\documentclass{revtex4}
\usepackage{amsmath,amsthm}

\usepackage{graphicx}

\usepackage{amsfonts}
\newtheorem{theorem}{Theorem}[section]
\newtheorem{lemma}[theorem]{Lemma}
\newtheorem{proposition}[theorem]{Proposition}

\theoremstyle{remark}
\newtheorem{remark}[theorem]{Remark}

\theoremstyle{definition}
\newtheorem{definition}[theorem]{Definition}

\theoremstyle{example}

\theoremstyle{notation}
\newtheorem{notation}[theorem]{Notation}

\begin{document}
\title{Quantum mechanics on profinite groups and partial order}            
\author{A. Vourdas}
\affiliation{Department of Computing,\\
University of Bradford, \\
Bradford BD7 1DP, United Kingdom}

\begin{abstract}
Inverse limits and profinite groups are used in a quantum mechanical context. Two cases are considered.
A quantum system with positions in the profinite group ${\mathbb Z}_p$ and momenta in the
group ${\mathbb Q}_p/{\mathbb Z}_p$; and a quantum system with positions in the profinite group ${\widehat {\mathbb Z}}$ and momenta in the
group ${\mathbb Q}/{\mathbb Z}$.
The corresponding Schwatz-Bruhat spaces of wavefunctions and
the Heisenberg-Weyl groups are discussed.
The sets of subsystems of these systems are studied from the point of view of partial order theory.
It is shown that they are directed-complete partial orders.
It is also shown that they are topological spaces with $T_0$ topologies, and this is used to define continuity of various physical quantities.
The physical meaning of 
profinite groups, non-Archimedean metrics, partial orders and $T_0$ topologies, in a quantum mechanical context, is discussed.
\end{abstract}

\maketitle

\newpage
\section*{Contents}
\begin{enumerate}
\item
Introduction
\item
Preliminaries
\begin{enumerate}
\item
Pontryagin duality
\item
Partial orders
\begin{itemize}
\item
The partially ordered sets ${\mathbb N}$, ${\mathbb N}^{(p)}$ and ${\mathbb N}(n)$
\end{itemize}
\item
The supernatural numbers as a directed-complete partially ordered set
\item
$T_0$ topological spaces
\item
$p$-adic numbers
\item
The group ${\widehat {\mathbb Z}}$
\end{enumerate}
\item
Inverse and direct limits in a quantum mechanical context
\begin{enumerate}
\item
The profinite group ${\mathbb Z}_p$ and its Pontryagin dual group ${\mathbb Q}_p/{\mathbb Z}_p$
\begin{itemize}
\item
${\mathbb Z}_p$ as inverse limit
\item
${\mathbb Q}_p/{\mathbb Z}_p$ as direct limit
\item
The chains ${\cal G}^{(p)}$ and ${\cal G}_S^{(p)}$ of Pontryagin dual pairs of groups
\end{itemize}
\item
The profinite group ${\widehat {\mathbb Z}}$ and its Pontryagin dual group ${\mathbb Q}/{\mathbb Z}$
\begin{itemize}
\item
${\widehat {\mathbb Z}}$ as inverse limit
\item
${\mathbb Q}/{\mathbb Z}$ as direct limit
\item
The directed partially ordered sets  ${\cal G}$ and ${\cal G}_S$ of Pontryagin dual pairs of groups
\end{itemize}
\end{enumerate}
\item
The system $\Sigma[{\mathbb Z}_p, ({\mathbb Q}_p/{\mathbb Z}_p)]$
\begin{enumerate}
\item
Locally constant functions and functions with compact support
\item
Integrals
\item
Fourier transforms
\item
The Schwartz-Bruhat space ${\mathfrak B}[{\mathbb Z}_p, ({\mathbb Q}_p/{\mathbb Z}_p)]$
\begin{itemize}
\item
Time evolution
\end{itemize}

\item
The Heisenberg-Weyl group ${\bf HW}[({\mathbb Q}_p/{\mathbb Z}_p),
{\mathbb Z}_p, ({\mathbb Q}_p/{\mathbb Z}_p)]$
\begin{itemize}
\item
Coherent states
\end{itemize}
\item
Parity operators
\item
The profinite Heisenberg-Weyl group 
${\bf HW}({\mathbb Z}_p, {\mathbb Z}_p, {\mathbb Z}_p)$
\end{enumerate}
\item
Subsystems of $\Sigma [{\mathbb Z}_p, ({\mathbb Q}_p/{\mathbb Z}_p)]$
\begin{enumerate}
\item
Subsystems and supersystems
\item
The subsystems $\Sigma[{\mathbb Z}(p^n),{\mathbb Z}(p^n)]$
\begin{itemize}
\item
The finite Heisenberg-Weyl group ${\bf HW}[{\mathbb Z}(p^n),{\mathbb Z}(p^n),{\mathbb Z}(p^n)]$
\end{itemize}
\item
Embeddings and their compatibility
\item
The chains $\Sigma ^{(p)}$ and $\Sigma _S^{(p)}$ as $T_0$ topological spaces
\item
Physical importance of the non-Archimedean metric and of the profinite topology
\end{enumerate}
\item
The system $\Sigma[{\widehat {\mathbb Z}}, ({\mathbb Q}/{\mathbb Z})]$
\begin{enumerate}
\item
The Schwartz-Bruhat space ${\mathfrak B}[{\widehat {\mathbb Z}}, ({\mathbb Q}/{\mathbb Z})]$
as the restricted tensor product of 
${\mathfrak B}[{\mathbb Z}_p, ({\mathbb Q}_p/{\mathbb Z}_p)]$
\item
The Heisenberg-Weyl group ${\bf HW}[({\mathbb Q}/{\mathbb Z}),
{\widehat{\mathbb Z}}, ({\mathbb Q}/{\mathbb Z})]$
as the restricted direct product of ${\bf HW}[({\mathbb Q}_p/{\mathbb Z}_p),
{\mathbb Z}_p, ({\mathbb Q}_p/{\mathbb Z}_p)]$
\begin{itemize}
\item
Coherent states
\end{itemize}
\item
Parity operators
\item
The profinite Heisenberg-Weyl group 
${\bf HW}({\widehat {\mathbb Z}},{\widehat {\mathbb Z}},{\widehat {\mathbb Z}})$
\end{enumerate}
\item
Subsystems of
$\Sigma [{\widehat {\mathbb Z}},({\mathbb Q}/{\mathbb Z})]$
\begin{enumerate}
\item
The subsystems $\Sigma [{\mathfrak Z}(n),\widetilde {\mathfrak Z}(n)]$ with $n \in {\mathbb N}_S$
\item
The subsystems $\Sigma[{\mathbb Z}(\ell),{\mathbb Z}(\ell)]$ with $\ell \in {\mathbb N}$, and their factorization
\begin{itemize}
\item
The finite Heisenberg-Weyl group ${\bf HW}[{\mathbb Z}(\ell),{\mathbb Z}(\ell),{\mathbb Z}(\ell)]$
\end{itemize}
\item
Embeddings and their compatibility
\item
The partially ordered sets $\Sigma $ and $\Sigma _S$ as $T_0$ topological spaces
\item
Physical importance of the $T_0$-topology
\item
Continuity of physical quantities in $\Sigma [{\mathbb Z}(n),{\mathbb Z}(n)]$ as a function of $n$
\item
Physical importance of the profinite topology and of the Schwartz-Bruhat space
\end{enumerate}
\item
Discussion
\end{enumerate}

\section{Introduction}

The starting point for quantum mechanics is a pair of Abelian groups $(G,\widetilde G)$ which are 
Pontryagin dual to each other, so that we can use one of them as the group of positions 
and the other as the group of momenta.  We use the notation
$\Sigma(G,\widetilde G)$ for a quantum system where the position variable takes 
values in the group $G$ and the momentum variable takes 
values in the Pontryagin group $\widetilde G$.  
Then $G\times \widetilde G$ is the phase space of this system. 

We are interested in a set of quantum systems with positions in the Abelian groups $G_n$
and momenta in their Pontryagin dual groups $\widetilde {G_n}$.
By taking the inverse limit ${\mathfrak G}$ of the groups $G_n$ \cite{PRO1,PRO2,PRO3},
and the direct limit  ${\widetilde {\mathfrak G}}$ of the corresponding Pontryagin dual groups $\widetilde {G_n}$,
we get the groups $({\mathfrak G},{\widetilde {\mathfrak G}})$ which are Pontryagin dual to each other.
In a quantum mechanical context, 
the corresponding quantum system $\Sigma ({\mathfrak G},{\widetilde {\mathfrak G}})$ is at the `edge' of the quantum systems
$\Sigma(G_n,\widetilde {G _n})$, i.e., it is the smallest system which contains
all systems in the set $\{\Sigma(G_n,\widetilde {G _n})\}$.

In this paper we are interested in the case that the $G_n$ are finite Abelian groups, and then the inverse limit 
${\mathfrak G}$ is a profinite group.
Profinite groups have been studied extensively by the mathematical community in the last few years\cite{PRO1,PRO2,PRO3}.
They are topological groups at the edge of finite groups, i.e., between 
finite and infinite groups.

In the case $G_n={\mathbb Z}(n)$ (the integers modulo $n$), we get the finite quantum systems 
$\Sigma({\mathbb Z}(n), {\mathbb Z}(n))$. They
have been studied extensively in the last few years (see reviews in refs\cite{F1,F2,F3,F4,F5} and also refs\cite{To1,To2,To3}).
We factorize $n$ in terms of powers of prime numbers as $n=\prod p_i^{e_i}$ where $p_i$ are prime numbers and $i=1,...,\ell$. Then
the quantum formalism for $\Sigma({\mathbb Z}(n), {\mathbb Z}(n))$,
can be factorized in terms of the quantum formalisms for $\Sigma({\mathbb Z}(p_i^{e_i}), {\mathbb Z}(p_i^{e_i}))$ \cite{facto1,facto2}.
This factorization is based on the Chinese remainder theorem.
In this sense the quantum systems $\Sigma({\mathbb Z}(p_i^{e_i}), {\mathbb Z}(p_i^{e_i}))$ are the buildings blocks of all the
finite quantum systems $\Sigma({\mathbb Z}(n), {\mathbb Z}(n))$.
We note here that there is also a lot of work on quantum systems 
$\Sigma(GF(p^n),GF(p^n))$
with variables in the Galois fields $GF(p^n)$ (reviewed in \cite{Galois}), but this is not relevant in the present paper.

We first consider the systems $\Sigma({\mathbb Z}(p^e), {\mathbb Z}(p^e))$ (with fixed prime number $p$).
The inverse limit of ${\mathbb Z}(p^e)$ is the profinite group ${\mathbb Z}_p$ ($p$-adic integers)
and the direct limit of their Pontryagin dual groups which are also
${\mathbb Z}(p^e)$, is ${\mathbb Q}_p/{\mathbb Z}_p$ (where ${\mathbb Q}_p$ is the field of $p$-adic numbers). 
Therefore in this case
we get the system $\Sigma[{\mathbb Z}_p, ({\mathbb Q}_p/{\mathbb Z}_p)]$.

We also consider the systems $\Sigma({\mathbb Z}(n), {\mathbb Z}(n))$.
The inverse limit of ${\mathbb Z}(n)$ is the profinite group ${\widehat {\mathbb Z}}$ (defined later)
and the direct limit of their Pontryagin dual groups which are also
${\mathbb Z}(n)$, is ${\mathbb Q}/{\mathbb Z}$ (rational numbers on a circle). 
Therefore in this case
we get the system $\Sigma[{\widehat {\mathbb Z}}, ({\mathbb Q}/{\mathbb Z})]$.

General references on p-adic numbers are \cite{NP1,NP2} and on other related number fields \cite{N1,N2,N3,N4}.
Quantum mechanics and quantum field theory on p-adic numbers 
have been studied in refs \cite{r1,r2,r3,r4,r5,r6,r7,r8,r9,r10,r11,r12,r13,r14,r15,r16} and they have been applied to
the physics at the Planck scale, condensed matter\cite{ap1,ap2,ap3,ap4}, etc.
Applications of $p$-adic numbers to (classical) computer science have been discussed in \cite {comp}.
Wavelets on p-adic numbers have been studied in \cite{wav1,wav2,wav3,wav4}.
Related mathematical work (e.g., Fourier transforms of functions on $p$-adic numbers, generalized functions, etc) has been presented in 
\cite{b1,b2,b3,b4,b5,b6,b7}.
We also mention harmonic analysis on adeles \cite{AD1,AD2,AD3,AD4}; harmonic analysis on number fields in Tate's thesis (see \cite{N1}); and the 
Langlands programme \cite {L1} which studies deep connections between automorphic forms and other areas.
In the present paper, we discuss some of these ideas in the context of quantum mechanics,
taking into account recent developments by the mathematical community on profinite groups.

We also make a link between the above ideas and partial order theory \cite{Bi,Sz,Gr}, which has not been explored in the literature.
We define the partial order `subsystem'
in the set of the quantum systems $\{\Sigma (G_n , \widetilde G_n)\}$.
We say that $\Sigma (E,\widetilde E)$ is a subsystem of $\Sigma (G,\widetilde G)$ if 
$\widetilde E$ is a subgroup of $\widetilde G$ (the relationship between their Pontryagin dual groups $E$ and $G$ is discussed below).
The system $\Sigma ({\mathfrak G},{\widetilde {\mathfrak G}})$ is a maximum element which when added to this set,
it makes it a directed-complete partially ordered set.
The concept of directed-complete partially ordered set, is used in domain theory which plays an important role 
in theoretical computer science \cite{D1,D2,D3}.
The partial order theory approach, makes precise our statement earlier that the system $\Sigma ({\mathfrak G},{\widetilde {\mathfrak G}})$ is
at the edge of the quantum systems $\{\Sigma(G_n,\widetilde {G _n})\}$.
We mention here that 
lattices have been used to describe quantum logic \cite{LO1,LO2,LO3,LO4}.

It is known for a long time\cite {T1,T2} that there is a strong link between partial order theory and topology (e.g., \cite{TO1,TO2,TO3,TO4,TO5}).
We make this link clear in our context \cite{vourdas}, by making the sets $\{\Sigma(G_n,\widetilde {G _n})\}$
topological spaces with a $T_0$ topology.
From a physical point of view, this topology can be used to define continuity of several physical quantities in $\Sigma({\mathbb Z}(n), {\mathbb Z}(n))$, 
as a function of the dimension of the system $n$.
Intuitively, we expect a physical quantity (e.g., entropy) as defined in the systems $\{\Sigma(G_n,\widetilde {G _n})\}$
to be a continuous function of $n$.
But this requires appropriate topology, and the $T_0$-topology does this.

The work applies to quantum mechanics three different but linked areas of mathematics:
\begin{itemize}
\item
Profinite groups (and $p$-adic groups). They are topological groups with Hausdorff, totally disconnected and compact topology. 
Therefore  we can define locally constant functions with compact support on them, and also Haar measure for integrals.
\item
Partial order theory applied to sets of quantum systems $\{\Sigma(G_n,\widetilde {G _n})\}$. 
With it, smaller systems are embedded into larger ones and the $\Sigma ({\mathfrak G},{\widetilde {\mathfrak G}})$ is at the edge of
$\{\Sigma(G_n,\widetilde {G _n})\}$ (i.e., it is the smallest quantum system that contains all the $\Sigma(G_n,\widetilde {G _n})$).
\item
$T_0$-topologies applied to sets of quantum systems $\{\Sigma(G_n,\widetilde {G _n})\}$. These sets are topological spaces with
the systems $\Sigma (G_n,{\widetilde G}_n)$ as `points'.
Open sets contain some systems and their subsystems, and
closed sets contain some systems and their supersystems. 
The separation axioms of the $T_0$-topology, reflect very basic logical relations between the quantum systems.
The $T_0$-topology can be used to define continuity of a physical quantity in the various $\Sigma(G_n,\widetilde {G _n})$, as a function of $n$.
We note that the $T_0$-topology on the set $\{\Sigma(G_n,\widetilde {G _n})\}$, should not be confused with the topologies of the profinite groups which are Hausdorff topologies.
\end{itemize}
The emphasis is on linking these ideas to each other, in a quantum mechanical context.

In section 2 we present briefly background material in order to establish the notation.
This includes partial order theory, $T_0$ topologies, $p$-adic numbers, etc.
In section 3 we present the concepts of inverse and direct limits with emphasis on their physical importance in a quantum mechanical context.
These important mathematical concepts have not been used widely in physics, and one of the aims of this article is to 
show how they can be used in a physical context.

In section 4, we study the quantum system $\Sigma[{\mathbb Z}_p, ({\mathbb Q}_p/{\mathbb Z}_p)]$.
We define the Schwartz-Bruhat space of functions for this system, and we study the corresponding Heisenberg-Weyl group
${\bf HW}[({\mathbb Q}_p/{\mathbb Z}_p),
{\mathbb Z}_p, ({\mathbb Q}_p/{\mathbb Z}_p)]$.
We also give a series of properties of the displacement and parity operators in the phase space of this system.
In section 5, we study the set of subsystems of $\Sigma[{\mathbb Z}_p, ({\mathbb Q}_p/{\mathbb Z}_p)]$
as a partially ordered set and also as a $T_0$-topological space.
 
In section 6, we study the quantum system $\Sigma[{\widehat {\mathbb Z}}, ({\mathbb Q}/{\mathbb Z})]$.
We define the Schwartz-Bruhat space, and present a phase space formalism which includes the Heisenberg-Weyl group ${\bf HW}[({\mathbb Q}/{\mathbb Z}),
{\widehat{\mathbb Z}}, ({\mathbb Q}/{\mathbb Z})]$ and several
properties of the displacement and parity operators.
In section 7, we study the set of subsystems of $\Sigma[{\widehat {\mathbb Z}}, ({\mathbb Q}/{\mathbb Z})]$
as a partially ordered set. We also study this set as a $T_0$-topological space and show that
the topology reflects fundamental logical relations between the subsystems. 
With this topology, a physical quantity which is defined in the various systems $\Sigma(G_n,\widetilde {G _n})$, 
is shown to be a continuous function of $n$.
We conclude in section 8 with a discussion of our results.

\section{Preliminaries}\label{A}

${\mathbb R}$ denotes the real numbers;
${\mathbb Q}$ the rational numbers;
${\mathbb Z}$ the integers.
${\mathbb Z}^+$ are the positive integers, ${\mathbb Z}_0^+$ are the non-negative integers, and
${\mathbb N}={\mathbb Z}^+-\{1\}$.
$a|b$ where $a,b\in {\mathbb Z}^+$, denotes that $a$ is a divisor of $b$. 

$\Pi$ is the set of prime numbers. 
If $n$ is factorized in terms of powers of prime numbers as
\begin{eqnarray}\label{factorize10}
n=p_1^{e_1}...p_\ell^{e_\ell};\;\;\;\;n\in {\mathbb Z}^+,
\end{eqnarray}
we will use the notation
\begin{eqnarray}
\Pi(n)=\{p_1,...,p_\ell\};\;\;\;\;E(n)=\{e_1,...,e_\ell\}.
\end{eqnarray}

${\mathbb Z}(n)$ is the ring of integers
modulo $n$. 
We will use the notation
\begin{eqnarray}
\omega _n(\alpha)=\exp \left(\frac{i2\pi \alpha}{n}\right );\;\;\;\;\;\;\alpha \in {\mathbb Z}(n)
\end{eqnarray}
for the characters of ${\mathbb Z}(n)$ (the roots of unity).
Then
\begin{eqnarray}
\frac{1}{n}\sum _{\alpha \in {\mathbb Z}(n)}\omega _n(\alpha \beta)=\delta _n(\beta,0);\;\;\;\;\;\;
\beta \in {\mathbb Z}(n)
\end{eqnarray}
where $\delta _n(\beta,0)$ is Kronecker's delta.

$C(n)$ is the multiplicative group
\begin{eqnarray}
C(n)=\{\omega _{n}(\alpha)|\alpha \in {\mathbb Z}(n)\}\cong {\mathbb Z}(n).
\end{eqnarray}
The notation $G_1\le G_2$ for two groups, means that $G_1$ is a subgroup of $G_2$.
If $m|n$, then ${\mathbb Z}(m)\le {\mathbb Z}(n)$ (and $C(m)\le C(n)$).
If $n$ is factorized as in Eq.(\ref{factorize10}), then
\begin{eqnarray}
{\mathbb Z}(n)\cong {\mathbb Z}(p_1^{e_1})\times ...\times {\mathbb Z}(p_\ell^{e_\ell});\;\;\;\;
C(n)\cong C(p_1^{e_1})\times ...\times C(p_\ell^{e_\ell})
\end{eqnarray} 

\subsection{Pontryagin duality}\label{Pont}

Let $G$ be an Abelian group and $\widetilde G$ its Pontryagin dual group (i.e. the group of its characters).
We consider a subgroup $E$ of $G$ and study briefly the relation between their Pontryagin dual groups
$\widetilde E$ and $\widetilde G$.
The annihilator ${\rm Ann} _{\widetilde G}(E)$ of $E$ is a subgroup of $\widetilde G$, such that for all $a\in E$ and all $b\in {\rm Ann} _{\widetilde G}(E)$,
the characters $\chi _b(a)=1$. It is easily seen that
\begin{eqnarray}
F\le E\le G\;\;\rightarrow\;\;{\rm Ann} _{\widetilde G}(F)\ge {\rm Ann} _{\widetilde G}(E).
\end{eqnarray}

The theory of Pontryagin duality (e.g., \cite{Po}) proves that
the Pontryagin dual group of $E$ is isomorphic to $\widetilde G/{\rm Ann} _{\widetilde G}(E)$; that 
the Pontryagin dual group of $G/E$ is isomorphic to ${\rm Ann} _{\widetilde G}(E)$; and that the annihilator
${\rm Ann} _{G}({\rm Ann} _{\widetilde G}(E))$ is isomorphic to $E$:
\begin{eqnarray}\label{ann}
\widetilde E\cong \widetilde G/{\rm Ann} _{\widetilde G}(E);\;\;\;\;
\widetilde{G/E}\cong {\rm Ann} _{\widetilde G}(E);\;\;\;\;
{\rm Ann} _{G}({\rm Ann} _{\widetilde G}(E))\cong E
\end{eqnarray}
We next consider a subgroup $F$ of $E$. Then
\begin{eqnarray}\label{47}
\widetilde F\cong \widetilde G/{\rm Ann} _{\widetilde G}(F)\cong 
\frac{\widetilde G/{\rm Ann} _{\widetilde G}(E)}{{\rm Ann} _{\widetilde G}(F)/{\rm Ann} _{\widetilde G}(E)}
\end{eqnarray}
Therefore
\begin{eqnarray}\label{4}
F\le E\le G\;\;\rightarrow\;\;
\widetilde F\cong 
\frac{\widetilde E}{{\rm Ann} _{\widetilde G}(F)/{\rm Ann} _{\widetilde G}(E)}
\end{eqnarray}
It is seen that the Pontryagin dual concept to `subgroup' is `quotient'.
Specific examples of these ideas, are discussed later.

\subsection{Partial orders}
A directed partially ordered set is a set ${\mathbb A}$ with a relation $\prec$ such that:
\begin{itemize}
\item
$a\prec a$, for all $a\in {\mathbb A}$; 
\item
if $a\prec b$ and $b\prec a$, then $a=b$;
\item
if $a\prec b$ and $b\prec c$, then $a\prec c$;
\item
for $a,b\in {\mathbb A}$ there exists $c\in {\mathbb A}$ such that $a\prec c$ and $b\prec c$.
\end{itemize}
We say that $a,b$ are comparable if $a\prec b$ or $b\prec a$.
A partially ordered set where any pair of elements is comparable, is a chain (total order).
A partially ordered set where no pair of elements is comparable, is an antichain.

Two partially ordered sets $({\mathbb A},\prec )$ and $({\mathbb B},\prec ')$ are order isomorphic,
if there is a bijective map $f$ from ${\mathbb A}$ to ${\mathbb B}$, and 
$f(a_1)\prec ' f(a_2)$, if and only if $a_1\prec  a_2$ (where $a_1,a_2\in {\mathbb A}$).
Below, for simplicity we use the same symbol $\prec$ for different partial orders.

An upper bound of a subset ${\mathbb B}$ of the partially ordered set ${\mathbb A}$, 
is an element $a\in {\mathbb A}$ such that $b\prec a$ for all $b\in {\mathbb B}$. 
If the set of all upper bounds of ${\mathbb B}$ has a smallest element, it is called the supremum 
(or least upper bound) of ${\mathbb B}$.

An element $m\in {\mathbb A}$ is called maximal, if there is no element $k \in {\mathbb A}$ such that $m\prec k$.
A partially ordered set might have many maximal elements.
A finite partially ordered set has at least one maximal element.
An infinite partially ordered set might have no maximal elements.
A different concept is the maximum element.
An element $t\in {\mathbb A}$ is called maximum element, if $k\prec t$ for all $k \in {\mathbb A}$.

\paragraph*{Directed-complete partial orders:}
We give the following proposition without proof\cite{D1,D2,D3}:
\begin{proposition}
Let ${\mathbb A}$ be a partially ordered set.
The following two statements are equivalent:
\begin{itemize}
\item[(1)]
Every directed subset of ${\mathbb A}$ has a supremum.
\item[(2)]
Every chain in ${\mathbb A}$ has a supremum. 
\end{itemize}
If they hold, then the set ${\mathbb A}$ is called directed-complete partial order (dcpo).
\end{proposition}
From Zorn's lemma follows that a directed-complete partially ordered set has at least one maximal element. 
A chain with a supremum is a complete chain.

\paragraph*{Finite partially ordered sets:}
The width (length) of a finite partially ordered set ${\mathbb A}$ is the cardinality of the largest antichain (chain) in 
${\mathbb A}$. Dilworth's theorem \cite{DIL} states that if the width of ${\mathbb A}$ is $n$, then it can be partitioned into $n$ chains.

\subsubsection{The partially ordered sets ${\mathbb N}$, ${\mathbb N}^{(p)}$ and ${\mathbb N}(n)$}\label{refe}
Let
\begin{eqnarray}
{\mathbb N}=\{2,3,4,...\};\;\;\;\;{\mathbb N}^{(p)}=\{p,p^2,p^3,...\}\subset {\mathbb N}
\end{eqnarray}
where $p$ is a prime number.
${\mathbb N}$ with divisibility as an order (i.e., $a\prec b$ if $a|b$)
is a directed partially ordered set, and
${\mathbb N}^{(p)}$ is a chain.
${\mathbb N}$ is not a directed-complete partial order.
For example the chain $p,p^2,p^3,...$ where $p\in \Pi$, has no supremum.
${\mathbb N}$ has no maximal elements. Below we enlarge this set in order to make it
a directed-complete partial order.

We also consider the following subset of ${\mathbb N}$:
\begin{eqnarray}
{\mathbb N}(n)=\{k\in {\mathbb N}\;|\;k|n\}
\end{eqnarray}
The cardinality of ${\mathbb N}(n)$ is $\sigma _0(n)-1$ (the divisor function minus $1$ because we have excluded $1$).
We factorize $n$ as in Eq.(\ref{factorize10}) and see easily that the partially ordered set $N(n)$ 
has width equal to $\ell$.
Then according to Dilworth's theorem, we can partition ${\mathbb N}(n)$ into $\ell$ chains. One such partition is:
\begin{eqnarray}\label{901}
&&{\mathbb M}_1(n)=\{p_1^{r_1}\;|\;1\le r_1\le e_1\}\nonumber\\
&&{\mathbb M}_2(n)=\{p_1^{r_1}p_2^{r_2}\;|\;0\le r_1\le e_1;\;\;1\le r_2\le e_2\}\nonumber\\
&&....\nonumber\\
&&{\mathbb M}_\ell (n)=\{p_1^{r_1},p_2^{r_2}...p_\ell ^{r_\ell}\;|\;0\le r_1\le e_1;... \;\;0\le r_{\ell-1}\le e_{\ell -1};\;\;1\le r_\ell \le e_\ell\}
\end{eqnarray}
We assume that the labelling is such that $\max (e_1,...,e_\ell)=e_\ell$. 
Then ${\mathbb M}_\ell (n)$ is the largest chain and the length of ${\mathbb N}(n)$ is $(e_1+1)...(e_{\ell-1}+1)e_\ell$.

\begin{remark}
Later we will consider quantum systems $\Sigma({\mathbb Z}(n), {\mathbb Z}(n))$ with $n\in {\mathbb N}$.
We have excluded $1$ from ${\mathbb N}$, because the physical meaning of the one-dimensional system 
$\Sigma({\mathbb Z}(1), {\mathbb Z}(1))$ is limited. 
From a mathematical point of view, the implication of the exclusion of $1$, is that
${\mathbb N}$ is not a lattice (the `meet' operation is not defined).
\end{remark}

\subsection{The supernatural numbers as a directed-complete partially ordered set:}

The set ${\mathbb N}_S$ of supernatural (Steinitz) numbers \cite{PRO1,PRO2} is:
\begin{eqnarray}
{\mathbb N}_S=\left \{n=\prod p^{e_p(n)}\;|\; p\in \Pi;\;\;\;\;
e_p\in {\mathbb Z}_0^+\cup \{\infty\}\right \}
\end{eqnarray}
The index $S$ indicates supernatural or Steinitz.
Here the exponents can take the value $\infty$, and the product might contain an infinite number of
prime numbers. 
If there is no danger of confusion, we will use the simpler notation $e_p$ for the exponents. 
In this set only
multiplication is well defined, and by definition
\begin{eqnarray}
p^\infty p^e=p^\infty;\;\;\;\;
e\in {\mathbb Z}_0^+\cup \{\infty\}.
\end{eqnarray}
For the reasons we explained earlier, we exclude $1$ from ${\mathbb N}_S$, i.e., we require that at least one of the exponents is 
different from zero. 

In the special case that all $e_p\ne \infty$ and only a finite number of them are different from zero, 
the $\prod p^{e_p}\in {\mathbb N}$ (i.e., ${\mathbb N}$ is a subset of ${\mathbb N}_S$).

\begin{notation}
Let $(e_p)$ (where $e_p\in {\mathbb Z}_0^+\cup \{\infty\}$) be an infinite sequence of exponents labelled by $p\in \Pi$.
The $(e_p)\prec (e_p')$ means that
$e_p\le e_p'$ for all $p$. 
By definition all numbers in ${\mathbb Z}_0^+$ are smaller than $\infty$.
\end{notation}
If $(e_p)\prec (e_p')$ then we say that $n=\prod p^{e_p}$ is a divisor of $n'=\prod p^{e_p'}$
and we denote it as $n|n'$ or as $n\prec n'$.
An element of ${\mathbb N}_S$, corresponding to the sequence where all $e_p=\infty$, is
\begin{eqnarray}
\Omega=\prod _{p\in \Pi} p^{\infty}
\end{eqnarray}
This is the maximum element in ${\mathbb N}_S$ (every element of ${\mathbb N}_S$ is a divisor of $\Omega$).
Let $\Pi_1$ be a subset (finite or infinite) of $\Pi$. We define the 
\begin{eqnarray}
\Omega (\Pi_1)=\prod  _{p\in \Pi _1}p^{\infty};\;\;\;\;\;\Omega(\Pi _1)|\Omega
\end{eqnarray}

\begin{notation}
Let $n=\prod p^{e_p(n)}$ (where $e_p(n)\in {\mathbb Z}_0^+\cup \{\infty\}$) be a supernatural number. 
We use the notation 
\begin{itemize}
\item
$\Pi ^{(\infty )}(n)$ for the set of prime numbers for which $e_p(n)=\infty$. 
\item
$\Pi ^{(\rm fin )}(n)$ for the set of prime numbers for which $0<e_p (n)<\infty$.
\item
$\Pi ^{(0)}(n)=\Pi-\Pi ^{(\infty )}(n)-\Pi ^{(\rm fin )}(n)$.
\end{itemize}
Then
\begin{eqnarray}\label{super}
n=\prod  _{p\in \Pi ^{(\infty )}(n)}p^{\infty}\prod _{p\in \Pi ^{(\rm fin )}(n)}p^{e_p(n)}=
\Omega \left (\Pi ^{(\infty )}(n)\right )\prod _{p\in \Pi ^{(\rm fin )}(n)}p^{e_p(n)}
\end{eqnarray}
If $n\in {\mathbb N}$ then $\Pi ^{(\infty )}(n)$ is the empty set and $\Pi ^{(\rm fin )}(n)=\Pi (n)$.
\end{notation}
The set ${\mathbb N}_S$ (ordered by divisibility) is a directed-complete partial order, with $\Omega$ as maximum element.
An example of a complete chain in ${\mathbb N}_S$ is
\begin{eqnarray}
{\mathbb N}_S^{(p)}=\{p, p^2, ...,p^\infty\}={\mathbb N}^{(p)}\cup \{p^\infty\}
\end{eqnarray}
Here the supremum is $p^\infty$.
Other examples of chains in ${\mathbb N}_S$ are
\begin{eqnarray}
&&p_1\prec p_1^2\prec ...\prec p_1^\infty\prec p_1^\infty p_2\prec p_1^\infty p_2^2\prec...\prec p_1^\infty p_2^\infty;\;\;\;\;
p_1,p_2\in \Pi\nonumber\\
&&p_1p_2\prec (p_1p_2)^2\prec ...\prec p_1^\infty p_2^\infty\nonumber\\
&&2^\infty\prec 2^\infty 3^\infty\prec 2 ^\infty 3^\infty 5^\infty \prec ...\prec \Omega
\end{eqnarray}
In the first and second chain the supremum is $p_1^\infty p_2^\infty$ and in the last chain the 
supremum is $\Omega$.

\subsection{$T_0$ topological spaces}\label{topology}

\begin{definition}
A topological space $(X,T_X)$ is a set $X$ and a collection $T_X$ of subsets of $X$ (called open sets), such that
\begin{itemize}
\item[(1)]
any finite intersection of elements in $T_X$ is also in $T_X$
\item[(2)]
any union of elements in $T_X$ is also in $T_X$
\item[(3)]
$X$ and $\emptyset$ are elements of $T_X$.
\end{itemize}
The elements of $X$ are called points.
If $U$ is an open set in $X$ then $X-U$ is a closed set.
A basis $B_X$ is a subset of $T_X$ such that every open set in $T_X$ can be written as a union of open sets in $B_X$.
\end{definition}
If $Y\subset X$, the topology on $X$ induces the subspace topology on $Y$ as follows:
\begin{eqnarray}
T_Y=\{Y\cap U\;|\;U\in T_X\}.
\end{eqnarray}

\begin{definition}
The topological space $({\mathbb N}_S, {T}_{{\mathbb N}_S})$
with the `divisor topology' ${T}_{{\mathbb N}_S}$ is generated by the base
\begin{eqnarray}
B_{{\mathbb N}_S}=\{\emptyset, {U}(n)\;|\;n\in {\mathbb N}_S\};\;\;\;\;\;
U(n)=\{m\in {{\mathbb N}_S}\;|\;m|n\};\;\;\;\;n=2,3,...
\end{eqnarray}
\end{definition}
Here the $\emptyset, {U}(n)$ and all their unions are the open sets in this topological space.
We note that in the general definition of a topological space we require that the 
intersection (resp., union) of a finite (resp., any) number of open sets is an open set.
In our case the restriction to a finite number of open sets is not needed, because for each point $n$ the ${U}(n)$
is the smallest open set which contains $n$. 
Therefore open and closed sets satisfy exactly the same conditions. Such topology is known as Alexandrov topology.

The $({\mathbb N}_S, {T}_{{\mathbb N}_S})$ induces subspace topologies on its subsets 
${\mathbb N}$, ${\mathbb N}_S^{(p)}$ and ${\mathbb N}^{(p)}$. For ${\mathbb N}$ this is generated by the base
\begin{eqnarray}
B_{{\mathbb N}}=\{\emptyset, {U}(n)\;|\;n=2,3,...\};\;\;\;\;\;
U(n)=\{m\in {{\mathbb N}}\;|\;m|n\};\;\;\;\;n=2,3,...
\end{eqnarray}
For ${\mathbb N}_S^{(p)}$ this is generated by the base
\begin{eqnarray}
B_{{\mathbb N}_S^{(p)}}=\{\emptyset, {U}(p^n)\;|\;n\in {\mathbb Z}^+\cup \{\infty\}\};\;\;\;\;\;
U(p^n)=\{p^m\in {{\mathbb N}_S^{(p)}}\;|\;m\le n\};\;\;\;\;n\in {\mathbb Z}^+\cup \{\infty\}
\end{eqnarray}

For ${\mathbb N}^{(p)}$ this is generated by the base
\begin{eqnarray}
B_{{\mathbb N}^{(p)}}=\{\emptyset, {U}(p^n)\;|\;n=1,2,...\};\;\;\;\;\;
U(p^n)=\{p^m\in {{\mathbb N}^{(p)}}\;|\;m\le n\};\;\;\;\;n=1,2,...
\end{eqnarray}

Below we give the separation axioms relevant to the present work
\begin{definition}
Let $({\mathbb A}, {T}_{\mathbb A})$ be a topological space.
\begin{itemize}
\item[(1)]
It is a $T_0$-space (Kolmogorov),
if for all pairs of distinct points $a,b$ there exist an open set $U\in {T}_{\mathbb A}$ such that
either $a\in U$ and $b\notin U$ or $a\notin U$ and $b\in U$.
\item[(2)]
It is a $T_1$-space (Frechet),
if for all pairs of distinct points $a,b$ there exist two open sets $U_1,U_2\in {T}_{\mathbb A}$ such that
\begin{eqnarray}
a\in U_1;\;\;\;\;b\notin U_1;\;\;\;\;a\notin U_2;\;\;\;\;b\in U_2.
\end{eqnarray}
\item[(3)]
It is a $T_2$-space (Hausdorff),
if for all pairs of distinct points $a,b$ there exist two open sets $U_1,U_2\in {T}_{\mathbb A}$ such that
\begin{eqnarray}
a\in U_1;\;\;\;\;b\notin U_1;\;\;\;\;a\notin U_2;\;\;\;\;b\in U_2;\;\;\;\;U_1\cap U_2=\emptyset
\end{eqnarray}
\end{itemize}
\end{definition}
Clearly if a topological space is $T_2$ then it is also $T_1$, and if it is $T_1$ it is also $T_0$.
\begin{proposition}\label{ttt}
\begin{itemize}
\mbox{}
\item[(1)]
The topological spaces $({\mathbb N}_S, {T}_{{\mathbb N}_S})$, $({\mathbb N}, {T}_{{\mathbb N}})$, $({\mathbb N}_S^{(p)}, {T}_{{\mathbb N}_S^{(p)}})$ and 
$({\mathbb N}^{(p)}, {T}_{{\mathbb N}^{(p)}})$, are $T_0$ spaces, but not $T_1$ spaces.
\item[(2)]
The topological spaces $({\mathbb N}, {T}_{{\mathbb N}})$ and 
$({\mathbb N}^{(p)}, {T}_{{\mathbb N}^{(p)}})$ are locally compact, but not compact.
The topological spaces $({\mathbb N}_S, {T}_{{\mathbb N}_S})$ and $({\mathbb N}_S^{(p)}, {T}_{{\mathbb N}_S^{(p)}})$ are compact.
\end{itemize}
\end{proposition}
\begin{proof}
\mbox{}
\begin{itemize}
\item[(1)]
We give the proof for the topological space $({\mathbb N}, {T}_{{\mathbb N}})$ and the proof for the other spaces is similar.
 
We consider a pair of elements $m,n\in {\mathbb N}$. If one of them is a divisor of the other, for example 
$m|n$, then the $U(m)$ is an open set such that $m\in U(m)$ and $n\notin U(m)$.
If none of the $m,n$ is a divisor of the other, then again $m\in U(m)$ and $n\notin U(m)$.
This proves that $({\mathbb N}, {T}_{{\mathbb N}})$ is a $T_0$ topological space.

$({\mathbb N}, {T}_{{\mathbb N}})$ is not a $T_1$ topological space,
because if $m|n$, every open set which contains $n$, also contains $m$.
\item[(2)]
We first prove that the topological space $({\mathbb N}, {T}_{{\mathbb N}})$ is locally compact.
For each point $n\in {\mathbb N}$, every open cover has a finite subcover which actually consist of one open set
$U(m)$ where $n|m$. This proves that $({\mathbb N}, {T}_{{\mathbb N}})$ is locally compact.
On the other hand, $({\mathbb N}, {T}_{{\mathbb N}})$ is not compact, because
its open cover 
\begin{eqnarray}
\bigcup_{n\in {\mathbb N}}U(n)={\mathbb N}
\end{eqnarray}
has no finite subcover. 
In a similar way we prove that $({\mathbb N}^{(p)}, {T}_{{\mathbb N}^{(p)}})$ is locally compact, but not compact.

We next prove that the topological space $({\mathbb N}_S, {T}_{{\mathbb N}_S})$ is compact.
In this case every open cover has a finite subcover. For example, the open cover 
\begin{eqnarray}
\bigcup_{n\in {\mathbb N}_S}U(n)={\mathbb N}_S
\end{eqnarray}
has the finite subcover $U(\Omega)={\mathbb N}_S$.
In a similar way we prove that $({\mathbb N}_S^{(p)}, {T}_{{\mathbb N}_S^{(p)}})$ is compact.

\end{itemize}
\end{proof}
\begin{remark}
For some authors compact spaces are by definition Hausdorff.
We do not include the requirement to be Hausdorff, in the definition of compact topological spaces.  
\end{remark}
The physical importance of this proposition will be discussed later.

\paragraph*{Topological spaces homeomorphic to $({\mathbb N}_S, {T}_{{\mathbb N}_S})$ and continuity:}
Below we will introduce several topological spaces $({\mathbb A}, {T}_{\mathbb A})$ which are homeomorphic to the topological space
$({\mathbb N}_S, {T}_{{\mathbb N}_S})$ (we denote this as (${\mathbb A}, 
{T}_{\mathbb A}) \sim ({\mathbb N}_S, {T}_{{\mathbb N}_S})$).
We construct them as follows.
We consider a set ${\mathbb A}$
with elements $a_n$ which are labelled with $n\in {\mathbb N}_S$ (i.e., there is a bijection between 
${\mathbb A}$ and ${\mathbb N}_S$). This induces
a partial order in ${\mathbb A}$ where $a_m\prec a_n$ if and only if $m|n$.
Then ${\mathbb A}$ is order isomorphic to ${\mathbb N}_S$.
We make ${\mathbb A}$ a topological space with topology ${T}_{\mathbb A}$
generated by the base
\begin{eqnarray}
B_{\mathbb A}=\{\emptyset, {U}_{\mathbb A}(a_n)\;|\;n\in {\mathbb N}_S\};\;\;\;\;\;
U_{\mathbb A}(a_n)=\{a_m\in {\mathbb A}\;|\;a_m\prec a_n\}
\end{eqnarray}
Then (${\mathbb A}, {T}_{\mathbb A}) \sim ({\mathbb N}_S, {T}_{{\mathbb N}_S})$.
Also the map from ${\mathbb N}_S$ to ${\mathbb A}$, where $n$ is mapped into $a_n$, is a continuous function.

In a similar way we will introduce topological spaces homeomorphic to 
$({\mathbb N}, {T}_{{\mathbb N}})$, to $({\mathbb N}_S^{(p)}, {T}_{{\mathbb N}_S^{(p)}})$,
and also to $({\mathbb N}^{(p)}, {T}_{{\mathbb N}^{(p)}})$.

\paragraph*{The topological spaces  ${\cal Z}$ and ${\cal Z}^{(p)}$:}
We consider the additive groups ${\mathbb Z}(n)$. Then 
${\mathbb Z}(m)\le {\mathbb Z}(n)$ if and only if $m|n$. 
Let ${\cal Z}$ be the set
\begin{eqnarray}\label{eq2}
{\cal Z}=\{{\mathbb Z}(2),{\mathbb Z}(3),...\}.
\end{eqnarray}
There is a bijective map from ${\mathbb N}$ to ${\cal Z}$, where $n$ is mapped into ${\mathbb Z}(n)$. This
induces the partial order `subgroup' into ${\cal Z}$.
Therefore ${\cal Z}$ is now order isomorphic to the partially ordered set ${\mathbb N}$.
Furthermore, it can become the topological space 
$({\cal Z}, T_{\cal Z})$ (as described above) which is a $T_0$-topological space, 
homeomeorphic to $({\mathbb N}, {T}_{{\mathbb N}})$.
The points in this topology are the groups ${\mathbb Z}(n)$, and
an open (resp. closed) set contains some groups and all their subgroups (resp. supergroups).

In a similar way the subset of ${\cal Z}$
\begin{eqnarray}\label{eq1}
{\cal Z}^{(p)}=\{{\mathbb Z}(p),{\mathbb Z}(p^2),...\}.
\end{eqnarray}
is a chain order isomorphic to ${\mathbb N}^{(p)}$, and it is also the topological space $({\cal Z}^{(p)}, T_{{\cal Z}^{(p)}})$ which is homeomorphic to 
$({\mathbb N}^{(p)}, {T}_{{\mathbb N}^{(p)}})$.

Later we will introduce larger topological spaces $({\cal Z}_S, T_{{\cal Z}_S}) \sim ({\mathbb N}_S, {T}_{{\mathbb N}_S})$
and $({\cal Z}_S^{(p)},T_{{\cal Z}_S^{(p)}})\sim ({\mathbb N}_S^{(p)}, {T}_{{\mathbb N}_S^{(p)}})$.

\subsection{$p$-adic numbers}

The field ${\mathbb Q}_p$ of $p$-adic numbers (where $p\in \Pi$), contains elements which  
can be written as
\begin{eqnarray}\label{1}
a_p=\sum _{\nu ={\rm ord}(a_p)}^{\infty} {\overline a} _{\nu} p^{\nu};\;\;\;\;\;\;\;
0\le {\overline a} _{\nu}\le p-1
\end{eqnarray}
${\rm ord}(a_p)$ is the ordinal or valuation of $a_p$.
Addition and multiplication is the usual addition and multiplication of series, together with the `carry' operation.
The ring ${\mathbb Z}_p$ of $p$-adic integers contains elements with ${\rm ord}(a_p)\ge 0$.
More generally the ring $p^n{\mathbb Z}_p$ contains elements with ${\rm ord}(a_p)\ge n$.
For $n,m\ge 0$ the $p^{-m}{\mathbb Z}_p/p^n{\mathbb Z}_p\cong {\mathbb Z}(p^{n+m})$. 
The $-1$ as a $p$-adic number is $(p-1)[1+p+p^2+...]$. 

The metric is non-Archimedean. The absolute value of $a_p$ is 
\begin{eqnarray}\label{abs}
|a_p|_p=p^{-{\rm ord}(a_p)}
\end{eqnarray}
and it has the usual properties, and also the property
\begin{eqnarray}
|a_p+b_p|_p\le \max (|a_p|_p,|b_p|_p).
\end{eqnarray}

We consider all absolute values  $|\lambda |_p$ of a fixed number $\lambda$, for all $p\in \Pi$.
Also let $|\lambda |_\infty$ be the `usual' absolute value of $\lambda$.
Ostrowski's theorem states that
\begin{eqnarray}\label{tgb}
|\lambda |_\infty \prod _{p\in \Pi}|\lambda |_p=1. 
\end{eqnarray}

Given a prime $p$, any rational number $\kappa/\lambda$, can be written as
\begin{eqnarray}
\frac{\kappa}{\lambda}=p^s\frac{\kappa_1}{\lambda_1}=p^s({\overline a} _0+{\overline a} _1p+{\overline a} _2p^2+...)
\end{eqnarray}
where any two of the $\kappa_1,\lambda_1,p$ are coprime integers. If $s\ge 0$ then $\kappa/\lambda$ is a $p$-adic integer.

\paragraph*{${\mathbb Q}_p/{\mathbb Z}_p$ as the Pontryagin dual group of ${\mathbb Z}_p$.} 
The ${\mathbb Q}_p/{\mathbb Z}_p$ contains fractional $p$-adic numbers.
Its elements are cosets and can be represented by 
\begin{eqnarray}\label{81}
{\mathfrak a}_p&=&\sum _{\nu =\nu_0}^{-1} {\overline a} _{\nu} p^{\nu}
\end{eqnarray}
They are defined modulo a $p$-adic integer.

${\mathbb Q}_p/{\mathbb Z}_p$ is isomorphic to the Pr\"ufer $p$-group (or quasi-cyclic group) $C(p^ {\infty})$:
\begin{eqnarray}
C(p^ {\infty})=\{\omega _{p^n}(\alpha )|\alpha \in {\mathbb Z}(p^n), n\in {\mathbb Z}^+\}
\cong {\mathbb Q}_p/{\mathbb Z}_p
\end{eqnarray}
Here ${\mathfrak a}_p \in {\mathbb Q}_p/{\mathbb Z}_p$ corresponds to $\exp(i2\pi {\mathfrak a}_p)$. 

\paragraph*{Characters:}
The product $a_p{\mathfrak b}_p$ where $a_p\in {\mathbb Z}_p$ and
${\mathfrak b}_p\in {\mathbb Q}_p/{\mathbb Z}_p$
is also a coset in ${\mathbb Q}_p/{\mathbb Z}_p$.
The Pontryagin dual group of ${\mathbb Z}_p$ is ${\mathbb Q}_p/{\mathbb Z}_p$.
This will also be seen below where ${\mathbb Z}_p$ (resp. ${\mathbb Q}_p/{\mathbb Z}_p$)
is studied as the inverse limit (resp. direct limit)
of the groups ${\mathbb Z}(p^n)$ (resp. of the Pontryagin dual groups to ${\mathbb Z}(p^n)$ 
which are isomorphic to ${\mathbb Z}(p^n)$).
Additive characters are given by
\begin{eqnarray}\label{81}
\chi _p(a_p{\mathfrak b}_p)=\exp (i2\pi a_p{\mathfrak b}_p)
\end{eqnarray}
\begin{remark}
The Pontryagin dual group of $p^{e_p}{\mathbb Z}_p$ is ${\mathbb Q}_p/(p^{-e_p}{\mathbb Z}_p)$.
Its elements are cosets represented with the
\begin{eqnarray}\label{160}
a_p=\sum _{\nu =\nu _0}^{-e_p-1} {\overline a} _{\nu} p^{\nu}.
\end{eqnarray}
In the case $e_p=\infty$ 
\begin{eqnarray}\label{600}
p^\infty {\mathbb Z}_p=\{0\};\;\;\;\;\;{\mathbb Q}_p/(p^{-\infty}{\mathbb Z}_p)=\{0\}.
\end{eqnarray}
The ${\mathbb Q}_p/(p^{-\infty}{\mathbb Z}_p)$ has one element which is the coset with all $p$-adic numbers, and we represent it with $0$.
\end{remark}

\subsection{The group ${\widehat {\mathbb Z}}$}

${\widehat {\mathbb Z}}$ is the additive group  
\begin{eqnarray}\label{0}
{\widehat {\mathbb Z}}=\prod _{p\in \Pi} {\mathbb Z}_p.
\end{eqnarray}
Its elements are
\begin{eqnarray}\label{467}
a=(a_2,...,a_p,...);\;\;\;\;a_p\in {\mathbb Z}_p;\;\;\;\;\;p\in \Pi.
\end{eqnarray}
Here addition is performed componentwise.
The set of integers ${\mathbb Z}$ can be embedded into ${\widehat {\mathbb Z}}$ as follows:
\begin{eqnarray}
{\mathbb Z}\ni n\;\;\rightarrow\;\;(n,n,...)\in {\widehat {\mathbb Z}}
\end{eqnarray}
Here on the right hand side, we have $n$ in the $p$-adic form (with $p=2,3,5...$).

\paragraph*{${\mathbb Q}/{\mathbb Z}$ as the Pontryagin dual group of ${\widehat {\mathbb Z}}$.}
The ${\mathbb Q}/{\mathbb Z}$ contains rational numbers on a circle, and it is
the Pontryagin dual group of ${\widehat {\mathbb Z}}$.
This will also be seen below where  ${\widehat {\mathbb Z}}$ (resp. ${\mathbb Q}/{\mathbb Z}$)
is studied as the inverse limit (resp. direct limit)
of the groups ${\mathbb Z}(n)$
(resp. of the Pontryagin dual groups to ${\mathbb Z}(n)$ 
which are isomorphic to ${\mathbb Z}(n)$). 

${\mathbb Q}/{\mathbb Z}$ is isomorphic to $\prod {\mathbb Q}_p/{\mathbb Z}_p$.
Indeed, an element of ${\mathbb Q}/{\mathbb Z}$
can be written as $\kappa/\lambda $ where $\kappa, \lambda$ are coprime integers and $\kappa <\lambda $. 
We write $\kappa/\lambda $ as 
\begin{eqnarray}
\frac{\kappa}{\lambda}=\sum _{p\in \Pi(\lambda)}\frac {\kappa _p}{p^{e_p}};\;\;\;\;p\in \Pi(\lambda);\;\;\;\;e_p\in E(\lambda)
\end{eqnarray}
and then write $\kappa _p/p^{e_p}$, in the $p$-adic form. We get
\begin{eqnarray}\label{468}
&&\frac{\kappa}{\lambda}=({\mathfrak a}_2,...,{\mathfrak a}_p,...);\;\;\;\;
{\mathfrak a}_p\in {\mathbb Q}_p/{\mathbb Z}_p\nonumber\\
&&p\in \Pi(\lambda)\;\;\rightarrow \;\;{\mathfrak a}_p=p^{-e_p}\kappa _p\nonumber\\
&&p\in \Pi-\Pi(\lambda)\;\;\rightarrow \;\;{\mathfrak a}_p=0
\end{eqnarray}
The set $\Pi (\lambda)$ is finite, and therefore
all but a finite number of the ${\mathfrak a}_p$, are equal to zero.
Addition is performed componentwise.

\paragraph*{Characters:}
Let ${\mathfrak b} =({\mathfrak b}_2,...,{\mathfrak b}_p,...)\in {\mathbb Q}/{\mathbb Z}$ and 
$a=(a_2,...,a_p,...)\in {\widehat {\mathbb Z}}$.
The product ${\mathfrak b}a=({\mathfrak b}_2a_2,...,{\mathfrak b}_pa_p,...)$ is an element of 
${\mathbb Q}/{\mathbb Z}$.
Additive characters in ${\mathbb Q}/{\mathbb Z}$ are given by
\begin{eqnarray}\label{81}
\chi (a{\mathfrak b})=\prod _{p\in \Pi}\chi _p(a_p {\mathfrak b}_p)
\end{eqnarray}
This converges because only a finite number of the 
${\mathfrak b}_p$ are 
different from zero.

\begin{remark}\label{rem12}
In the context of `fast Fourier transforms'
Good \cite{ER,G1,G2} used the Chinese remainder theorem to factorize characters 
$\omega _n(\mu\nu)$ where $\mu,\nu\in {\mathbb Z}(n)$ as follows.
We factorize $n$ as in Eq.(\ref{factorize10}), and we define the
\begin{eqnarray}\label{variables}
u_i=\frac{n}{p_i^{e_i}};\;\;\;t_iu_i=1({\rm mod}\; {p_i^{e_i}});\;\;\;\;
w_i=t_iu_i \in {\mathbb Z}(n).
\end{eqnarray}
where $p_i\in \Pi(n)$ and $e_i\in E(n)$.
We can show that
\begin{eqnarray}\label{135}
w_iw_j=\delta _{ij}w_j({\rm mod}\; n);\;\;\;\;\;\;w_iu_j=\delta _{ij}u_j({\rm mod}\; n).
\end{eqnarray}
We now consider two bijective maps between the isomorphic groups
\begin{eqnarray}\label{r15}
{\mathbb Z}(n)\cong {\mathbb Z}(p_1^{e_1})\times ...\times {\mathbb Z}(p_\ell^{e_\ell}).
\end{eqnarray} 
The first map is 
\begin{eqnarray}\label{m1}
&&\mu\;\leftrightarrow\;(\mu_1,...,\mu_\ell);\;\;\;\;\mu\in {\mathbb Z}(n);\;\;\;\;\mu_i\in {\mathbb Z}({p_i^{e_i}})
\nonumber\\
&&\mu_i=\mu\; ({\rm mod}\; {p_i^{e_i}});\;\;\;\;\;\mu=\sum _i\mu_iw_i.
\end{eqnarray}
The second map is:
\begin{eqnarray}\label{m2}
&&\nu\;\leftrightarrow\;({\widehat \nu}_1,...,{\widehat \nu}_\ell);\;\;\;\;\nu\in {\mathbb Z}(n)
;\;\;\;\;{\widehat \nu}_i\in {\mathbb Z}({p_i^{e_i}})
\nonumber\\
&&{\widehat \nu}_i=\nu t_i\;({\rm mod}\; {p_i^{e_i}});\;\;\;\;\;\frac{\nu}{n}=\sum _i \frac{{\widehat \nu}_i}{p_i^{e_i}}.
\end{eqnarray}
Using Eq.(\ref{135}) we easily prove
the following properties for the summation and product of two numbers in ${\mathbb Z}(n)$:
\begin{eqnarray}
&&\mu +\nu\;\leftrightarrow\;(\mu_1+\nu_1,...,\mu_\ell+\nu _\ell )\nonumber\\
&&\mu \nu\;\leftrightarrow\;(\mu_1\nu_1,...,\mu_\ell\nu _\ell ),
\end{eqnarray}
and also
\begin{eqnarray}
&&\mu +\nu\;\leftrightarrow\;(\widehat \mu_1+\widehat \nu_1,...,\widehat \mu_\ell+\widehat \nu _\ell )\nonumber\\
&&\mu \nu\;\leftrightarrow\;(\widehat \mu_1\nu_1,...,\widehat \mu_\ell\nu _\ell ).
\end{eqnarray}
Using Eq.(\ref{135}) we can prove that 
\begin{eqnarray}\label{82}
\omega _n(\mu\nu)=\omega _{p_1^{e_1}}\left ({\widehat \nu}_1 \mu_1\right )...
\omega _{p_\ell ^{e_\ell}}\left ({\widehat \nu}_\ell \mu_\ell \right )
\end{eqnarray}

Eq.(\ref{467}) can be viewed as generalization of Eq.(\ref{m1}).
Given $a\in {\mathbb Z}(n)$ we express it using Eq.(\ref{m1}) as
$(a _1,...,a _\ell)$ where $a _i\in {\mathbb Z}(p_i^{e_i})$
and then write $a_n$ as
\begin{eqnarray}\label{1001}
&&a=(a_2,a_3,a_5,...)\in {\widehat {\mathbb Z}}\nonumber\\
&&p_i\in \Pi(n)\;\;\rightarrow \;\;a_{p_i}=a _i\nonumber\\
&&p\in \Pi-\Pi(n)\;\;\rightarrow \;\;a_p=0.
\end{eqnarray}
Here $a_{p_i}$ is the $a_i$ expressed in the $p$-adic form.
Also Eq.(\ref{468}) can be viewed 
as generalization of Eq.(\ref{m2}).
Then Eq.(\ref{81}) is a generalization of Eq.(\ref{82}) where
the correspondence between the variables in these two equations is given by
\begin{eqnarray}
a_{p_i}\;\rightarrow \;\mu_i;\;\;\;\;{\mathfrak b}_{p_i}\;\rightarrow \;p_i^{-e_i}\widehat \nu_i.
\end{eqnarray}
The generalization involves both the exponents $e_i$ going to infinity, and the inclusion of all prime numbers.

\end{remark}

\begin{remark}
For later use, we partition ${\widehat {\mathbb Z}}$ into ${\widehat {\mathbb Z}}_{\rm even}$ and 
${\widehat {\mathbb Z}}_{\rm odd}$ as follows:
\begin{eqnarray}\label{467v}
&&{\widehat {\mathbb Z}}_{\rm even}=\{(a_2,...,a_p,...)|{\rm ord}(a_2)\ge 1\}\nonumber\\
&&{\widehat {\mathbb Z}}_{\rm odd}=\{(a_2,...,a_p,...)|{\rm ord}(a_2)=0\}
\end{eqnarray}
\end{remark}

\section{Inverse and direct limits in a quantum mechanical context}

Let $\{G_n\}$ be a set of finite Abelian groups, where the indices $n$ belong to some directed partially ordered set.
Also let $\{\widetilde G_n\}$ be the set of their Pontryagin dual groups.
Inverse limits of finite groups lead to profinite groups which are 
Hausdorff, compact and totally disconnected topological groups.
Let ${\mathfrak G}$ be the inverse limit  of $\{G_n\}$,  
and $\widetilde {\mathfrak G}$ the direct limit of $\{\widetilde G_n\}$.
Then $\widetilde {\mathfrak G}$ is the Pontryagin dual group of ${\mathfrak G}$.

In this section we consider two sets of such groups 
$\{{\mathbb Z}(p^n)\}$ and  $\{{\mathbb Z}(n)\}$.
All these groups are Pontryagin self-dual.
The inverse limit of $\{{\mathbb Z}(p^n)\}$ and $\{{\mathbb Z}(n)\}$
are the profinite groups ${\mathbb Z}_p$ and ${\widehat {\mathbb Z}}$, 
correspondingly.
The direct limits of $\{{\mathbb Z}(p^n)\}$ and $\{{\mathbb Z}(n)\}$ are the
${\mathbb Q}_p/{\mathbb Z}_p$ and  ${\mathbb Q}/{\mathbb Z}$ correspondingly.

The concept of  inverse (resp., direct) limit, involves homomorphisms between the groups 
$\{G_n\}$ (resp., $\{\widetilde G_n\}$) which are discussed below.
These homomorphisms will be used to define
a partial order in the set $\{(G_n,\widetilde G_n)\}$
with the relation `subgroup' among the $\widetilde G_n$.
Then the various $G_n$ are related through quotients as described in Eq.(\ref{4}).
The sets
\begin{eqnarray}
\{(G_n,\widetilde G_n)\;|\;n\in {\mathbb N}\}\subset \{(G_n,\widetilde G_n)\;|\;n\in {\mathbb N}_S\}
\end{eqnarray}
are directed partially ordered sets. The larger set $\{(G_n,\widetilde G_n)\;|\;n\in {\mathbb N}_S\}$
which contains the  $({\mathfrak G}, \widetilde {\mathfrak G})$ is a directed-complete partially ordered set.
Below the $G_n \times \widetilde G_n$ will be used as phase spaces of quantum systems.
The embeddings defined in this section, will be used to define embeddings of smaller quantum systems into larger ones.

\subsection{The profinite group ${\mathbb Z}_p$ and its Pontryagin dual group ${\mathbb Q}_p/{\mathbb Z}_p$}

\subsubsection{${\mathbb Z}_p$ as inverse limit}
We consider the 
$\{{\mathbb Z}(p^n)\}$ as additive topological groups with the discrete topology.
Here $n$ belongs to ${\mathbb Z}^+$ which with the usual order is a chain. 
For $k\le \ell$ we define the continous homomorphisms:
\begin{eqnarray}\label{hom1}
&&\varphi_{k\ell}:\;{\mathbb Z}(p^k)\;\leftarrow\;{\mathbb Z}(p^\ell);\;\;\;\;\;\;\;\;k\le \ell\nonumber\\
&&\varphi_{k\ell}(\alpha_{p^\ell})=\alpha_{p^k};\;\;\;\;\;\;\alpha_{p^\ell}= \alpha_{p^k} ({\rm mod}\; p^k)
;\;\;\;\;\;
\alpha_{p^\ell}\in {\mathbb Z}(p^\ell);\;\;\;\;\;\alpha_{p^k} \in {\mathbb Z}(p^k).
\end{eqnarray}
These homomorphisms are compatible in the sense that
if $k\le  \ell \le r$, then
$\varphi_{k\ell}\circ \varphi_{\ell r}=\varphi_{kr}$. Also $\varphi_{kk}={\bf 1}$.

The $\{{\mathbb Z}(p^\ell),\varphi _{k\ell}\}$
is an inverse system with inverse limit the ${\mathbb Z}_p$:
\begin{eqnarray}
\lim _{\longleftarrow} {\mathbb Z}(p^\ell)={\mathbb Z}_p.
\end{eqnarray}
The elements of the inverse limit of this inverse system are the sequences
\begin{eqnarray}\label{pt}
&&\alpha =(\alpha _p,\alpha _{p^2},...);\;\;\;\;\;\;\;\alpha _{p^\ell} \in {\mathbb Z}(p^\ell)\nonumber\\
&&k\le \ell\;\;\rightarrow\;\;\alpha _{p^\ell}=\alpha _{p^k}\;({\rm mod}\;p^k),
\end{eqnarray}
and addition is performed componentwise.
This is a different but equivalent representation to Eq.(\ref{1}), of $p$-adic numbers.
Indeed we can replace the sequence in Eq.(\ref{pt}) with the sequence
\begin{eqnarray}
\{\overline \alpha _0, \overline \alpha _1, \overline \alpha _2,...\};\;\;\;\;\;
{\overline a} _{\nu}=p^{-\nu}(\alpha_{p^{\nu +1}}-\alpha_{p^\nu});\;\;\;\;\;
\overline \alpha _0=\alpha _p,
\end{eqnarray}
which we rewrite as a sum as in Eq.(\ref{1}).

\paragraph*{Fundamental system of neighbourhoods of $0$:}
${\mathbb Z}_p$ is a profinite group and as such it is a topological group 
which is Hausdorff, compact and totally disconnected. The
\begin{eqnarray}
\{p^n{\mathbb Z}_p\;|\;n\in {\mathbb Z}^+\} 
\end{eqnarray}
are a fundamental system of neighbourhoods of $0$.
This topology is the same as the topology endowed by the $p$-adic metric in Eq.(\ref{abs}).
It is easily seen that
\begin{eqnarray}
{\mathbb Z}_p\ge p{\mathbb Z}_p\ge p^2{\mathbb Z}_p\ge ...
\end{eqnarray}

\paragraph*{Projections:}
There exist projections $\xi _k$ from ${\mathbb Z}_p$ to ${\mathbb Z}(p^k)$, given by the `truncated series'
\begin{eqnarray}\label{AA1}
\xi _k(a_p)=\sum _{\nu =0}^{k-1} {\overline a} _{\nu} p^{\nu}
\end{eqnarray}
These projections are compatible with the $\varphi_{k\ell}$:
\begin{eqnarray}
k\le \ell\;\;\rightarrow\;\;\varphi_{k\ell}\circ \xi _\ell=\xi_k
\end{eqnarray}

\subsubsection{${\mathbb Q}_p/{\mathbb Z}_p$ as direct limit}

We consider the Pontryagin dual groups to $\{{\mathbb Z}(p^n)\}$ used in the inverse limit of Eq.(\ref{hom1}).
They are the  $\{{\mathbb Z}(p^n)\cong C(p^n)\}$ (they are self-dual).
For $k\le \ell$ we define the homomorphisms (`embeddings')
\begin{eqnarray}\label{hom2}
&&{\widetilde \varphi } _{k\ell}:\;{\mathbb Z}(p^k)\;\rightarrow\;{\mathbb Z}(p^\ell);\;\;\;\;\;\;\;\;
k\le \ell,\nonumber\\
&& {\widetilde \varphi } _{k\ell}(\alpha _{p^k})=\alpha _{p^\ell};\;\;\;\;\;\;
\alpha _{p^\ell}= p^{\ell -k}\alpha _{p^k};\;\;\;\;\;
\alpha _{p^\ell}\in {\mathbb Z}(p^\ell);\;\;\;\;\;
\alpha _{p^k} \in {\mathbb Z}(p^k).
\end{eqnarray}
The ${\mathbb Z}(p^k)$ is a subgroup of ${\mathbb Z}(p^\ell)$,
and for a given element in ${\mathbb Z}(p^k)$ the homomorphism defines
the corresponding element in ${\mathbb Z}(p^\ell )$.
These homomorphisms are compatible in the sense that
if $k\le  \ell \le r$ then
${\widetilde \varphi}_{\ell r}\circ {\widetilde \varphi}_{k \ell }={\widetilde \varphi}_{kr}$
In addition to that, ${\widetilde \varphi}_{\ell \ell}={\bf 1}$.

The $\{{\mathbb Z}(p^\ell),{\widetilde \varphi} _{k\ell}\}$ 
is a direct system with direct limit
\begin{eqnarray}
\lim _{\longrightarrow} {\mathbb Z}(p^\ell)\cong {\mathbb Q}_p/{\mathbb Z}_p\cong C(p^ {\infty}).
\end{eqnarray}
The  direct limit is the disjoint union of all $C(p^ {\ell})\cong {\mathbb Z}(p^\ell)$
modulo an equivalence relation where 
we identify $\alpha_{p^k} \in {\mathbb Z}(p^k)$ with $p^{\ell -k}\alpha_{p^k} \in {\mathbb Z}(p^\ell)$.
The ${\mathbb Q}_p/{\mathbb Z}_p$ is a discrete group (as the Pontryagin dual group to ${\mathbb Z}_p$ which is compact).

There exist homomorphisms ${\widetilde \xi} _\ell$ from 
${\mathbb Z}(p^\ell)$ to 
${\mathbb Q}_p/{\mathbb Z}_p$:
\begin{eqnarray}\label{AA2}
{\widetilde \xi} _\ell (\alpha _{p^\ell})=p^{-\ell}\alpha_{p^\ell}
\end{eqnarray}
They are compatible in the sense that
\begin{eqnarray}\label{34f}
k\le \ell\;\;\rightarrow\;\;{\widetilde \xi} _\ell \circ {\widetilde \varphi}_{k\ell}  
={\widetilde \xi} _k.
\end{eqnarray}

\paragraph*{The topological space ${\cal Z}_S^{(p)}$}

We can now extend the topological space ${\cal Z}^{(p)}$ in Eq.(\ref{eq1}), by adding the $C(p^ {\infty})\cong {\mathbb Q}_p/{\mathbb Z}_p$
(all the $C(p^n)$ are subgroups of $C(p^ {\infty})$):
\begin{eqnarray}
{\cal Z}_S^{(p)}={\cal Z}^{(p)}\cup \{C(p^ {\infty})\}
\end{eqnarray}
We make it a topological space as described in section \ref{topology}, and 
$({\cal Z}_S^{(p)},T_{{\cal Z}_S^{(p)}})\sim ({\mathbb N}_S^{(p)}, {T}_{{\mathbb N}_S^{(p)}})$.

\subsubsection{The chains ${\cal G}^{(p)}$ and ${\cal G}_S^{(p)}$ of Pontryagin dual pairs of groups}

We consider the set
\begin{eqnarray}\label{zc5}
{\cal G}_S^{(p)}=\{({\mathbb Z}(p), C(p)), ({\mathbb Z}(p^2), C(p^2)),...\}
\cup \{({\mathbb Z}_p,C(p^\infty))\}
\end{eqnarray}
Each element contains a pair of groups which are Pontryagin dual to each other. 
We order the elements of ${\cal G}_S^{(p)}$ by ordering the second elements of the pairs with the relation `subgroup',
and we get the following chain which is order isomorphic to ${\mathbb N}_S^{(p)}$:
\begin{eqnarray}\label{zc5}
({\mathbb Z}(p), C(p))\prec ({\mathbb Z}(p^2), C(p^2))\prec ({\mathbb Z}(p^3), C(p^3))\prec ...
\prec ({\mathbb Z}_p,C(p^\infty))
\end{eqnarray}
The last pair $({\mathbb Z}_p,C(p^\infty))$ is the supremum of this chain.
We added it, so that the chain is complete.
The first elements of these pairs are ordered indirectly by the fact that their Pontryagin duals are ordered with the relation `subgroup'.
A more direct relation can be found using Eq.(\ref{ann}), which involve quotients.
As an example, we consider the first elements of the pairs $({\mathbb Z}(p^n), C(p^n))$ and $({\mathbb Z}_p,C(p^\infty))$.
In this case $G=C(p^\infty)\cong {\mathbb Q}_p/{\mathbb Z}_p$ and $E=C(p^n)$. Then
\begin{eqnarray}
&&\widetilde G={\mathbb Z}_p;\;\;\;\;\;
{\rm Ann} _{\widetilde G}(E)\cong p^n{\mathbb Z}_p;\;\;\;\;\widetilde E=\widetilde G/{\rm Ann} _{\widetilde G}(E)\cong {\mathbb Z}_p/p^n{\mathbb Z}_p
\cong {\mathbb Z}(p^n).
\end{eqnarray}
To see this we write the elements of $E=C(p^n)$ as ${\mathfrak a}_p=\overline a_{-n}p^{-n}+...+\overline a_{-1}p^{-1}$ where $0\le \overline a_{-i}\le p-1$
and then it is easily seen that $\chi _p({\mathfrak a}_pb_p)=1$ for all $b_p\in p^n{\mathbb Z}_p$.
Therefore ${\rm Ann} _{\widetilde G}(E)\cong p^n{\mathbb Z}_p$.
From this follows that the Pontryagin dual of $E=C(p^n)$ is ${\mathbb Z}_p/p^n{\mathbb Z}_p$ which is isomorphic to ${\mathbb Z}(p^n)$.
This shows explicitly the `quotient relationship' between 
the first elements of the pairs $({\mathbb Z}(p^n), C(p^n))$ and $({\mathbb Z}_p,C(p^\infty))$.

A subset of ${\cal G}_S^{(p)}$ is
\begin{eqnarray}
{\cal G}^{(p)}=\{({\mathbb Z}(p), C(p)), ({\mathbb Z}(p^2), C(p^2)),...\}
\end{eqnarray}
This is also a chain, but it has no supremum.
It is order isomorphic to ${\mathbb N}^{(p)}$.

\paragraph*{The topological spaces $({\cal G}^{(p)}, T_{{\cal G}^{(p)}})$ and $({\cal G}_S^{(p)}, T_{{\cal G}_S^{(p)}})$:}
There is a bijective map between ${\mathbb N}_S^{(p)}$ and ${\cal G}_S^{(p)}$:
\begin{eqnarray}
&&p^n\;\;\leftrightarrow \;\;({\mathbb Z}(p^n), C(p^n))\nonumber\\
&&p^\infty \;\;\leftrightarrow \;\;({\mathbb Z}_p,{\mathbb Q}_p/{\mathbb Z}_p).
\end{eqnarray}
Using the methodology of section \ref{topology} we can make the ${\cal G}_S^{(p)}$ a topological space and 
$({\mathbb N}_S^{(p)}, T_{{\mathbb N}_S^{(p)}})\sim ({\cal G}_S^{(p)}, T_{{\cal G}_S^{(p)}})$.
In a similar way we can make ${\cal G}^{(p)}$ a topological space and 
$({\mathbb N}^{(p)}, T_{{\mathbb N}^{(p)}})\sim ({\cal G}^{(p)}, T_{{\cal G}^{(p)}})$.

\paragraph*{Embeddings:}
The products ${\mathbb Z}(p^k)\times {\mathbb Z}(p^k)$  
will be used later as phase spaces of quantum systems.
We use the inverse of a restriction of the homomorphism $\varphi _{k\ell}$  in Eq.(\ref{hom1}), and 
the homomorphism $\widetilde \varphi _{k\ell}$ in Eq.(\ref{hom2}), 
to define embeddings among them.
\begin{definition}\label{def1}
\mbox{}
\begin{itemize}
\item[(1)]
For $k\le \ell$, $E_{k\ell}$ is the embedding of ${\mathbb Z}(p^k)\times {\mathbb Z}(p^k)$ into
${\mathbb Z}(p^\ell)\times {\mathbb Z}(p^\ell)$, by mapping the
$(\alpha _{p^k},\beta _{p^k}) \in {\mathbb Z}(p^k)\times {\mathbb Z}(p^k)$ 
into $(\alpha _{p^\ell},\beta _{p^\ell}) \in {\mathbb Z}(p^\ell)\times {\mathbb Z}(p^\ell)$, as follows:
\begin{eqnarray}\label{bn1}
&&E_{k\ell}:\;(\alpha _{p^k},\beta _{p^k})\;\rightarrow \;(\alpha _{p^\ell},\beta _{p^\ell})
;\;\;\;\;k\le \ell \nonumber\\
&&\alpha _{p^\ell}=\alpha _{p^k};\;\;\;\;\;
\beta _{p^\ell}=p^{\ell -k}\beta _{p^k}.
\end{eqnarray}
\item[(2)]
$E_{k\infty}$ is the embedding of
${\mathbb Z}(p^k)\times {\mathbb Z}(p^k)$ into ${\mathbb Z}_p\times ({\mathbb Q}_p/{\mathbb Z}_p)$
by mapping the
$(\alpha _{p^k},\beta _{p^k}) \in {\mathbb Z}(p^k)\times {\mathbb Z}(p^k)$ into
$(a_p,{\mathfrak b}_p)\in {\mathbb Z}_p\times {\mathbb Q}_p/{\mathbb Z}_p$, as follows:
\begin{eqnarray}\label{bn2}
&&E_{k\infty}:\;(\alpha _{p^k},\beta _{p^k})\;\rightarrow \;(a_p,{\mathfrak b}_p)\nonumber\\
&&\alpha _{p}=\alpha _{p^k}=\overline \alpha _0+...+\overline \alpha _{k-1}p^{k-1};\;\;\;\;
{\mathfrak b}_{p}=p^{-k}\beta _{p^k}.
\end{eqnarray}
\end{itemize}
\end{definition}
The $\varphi _{k\ell}$ is not invertible because it is not injective 
($p^{\ell-k}$ elements 
of ${\mathbb Z}(p^\ell)$ are mapped into the same element in ${\mathbb Z}(p^k)$). 
But in Eq.(\ref{bn1}) we use the inverse of a restriction of $\varphi _{k\ell}$, by mapping
the $\alpha _{p^k}\in {\mathbb Z}(p^k)$
into $\alpha _{p^\ell}=\alpha _{p^k}$ in ${\mathbb Z}(p^\ell)$.
We also use in Eq.(\ref{bn1}), the homomorphism $\widetilde \varphi _{k\ell}$.

It is easily seen that
\begin{eqnarray}\label{omega}
\omega _{p^\ell}\left ( \alpha _{p^\ell}\beta _{p^\ell}\right )
=\omega _{p^k}\left ( \alpha _{p^k}\beta _{p^k}\right );\;\;\;\;\;
\chi _{p}\left (a_{p}\mathfrak b _{p}\right )=\omega _{p^k}\left ( \alpha _{p^k}\beta _{p^k}\right ).
\end{eqnarray}
The embeddings $E_{k\ell}$ are compatible:
\begin{eqnarray}
k\le \ell \le r\;\;\rightarrow\;\;E_{\ell r}\circ E_{k\ell}=E_{kr}.
\end{eqnarray}
\begin{remark}\label{re}
In Eq.(\ref{bn1}), the map $\alpha _{p^k}\;\;\rightarrow\;\;\alpha _{p^\ell}=\alpha _{p^k}$
 is different from the map $\beta _{p^k}\;\;\rightarrow\;\;\beta _{p^\ell}=p^{\ell -k}\beta _{p^k}$
(the factor $p^{\ell -k}$).
This is related to the fact that the $\alpha$-variables are eventually linked to ${\mathbb Z}_p$, while the 
$\beta$-variables are eventually linked to its Pontryagin dual group ${\mathbb Q}_p/{\mathbb Z}_p$.
The extra factor $p^{\ell -k}$ with the $\beta$-variables is important for the equality in Eq.(\ref{omega}), 
and it is related to the fact that the Pontryagin dual concept to subgroup is quotient, as we explained earlier. 
We also discuss this later in remark \ref{last1}.
\end{remark}

\subsection{The profinite group ${\widehat {\mathbb Z}}$ and its Pontryagin dual group ${\mathbb Q}/{\mathbb Z}$}

\subsubsection{${\widehat {\mathbb Z}}$ as inverse limit}\label{nei}

We consider the additive topological groups $\{{\mathbb Z}(n)\}$ 
with the discrete topology. 
Here $n$ belongs to ${\mathbb Z}^+$ which with divisibility as a partial order, is a directed partially ordered set.
For $k |\ell$ 
we define the continuous homomorphisms:
\begin{eqnarray}\label{hom3}
&&\Phi_{k\ell}:\;{\mathbb Z}(k)\;\leftarrow\;{\mathbb Z}(\ell);\;\;\;\;\;\;\;\;k|\ell\nonumber\\
&&\Phi_{k\ell}(\alpha_\ell)=\alpha_k;\;\;\;\;\;\;\alpha_\ell= \alpha_k ({\rm mod}\;k);\;\;\;\;\;
\alpha_{\ell}\in {\mathbb Z}(\ell);\;\;\;\;\;\alpha_k \in {\mathbb Z}(k).
\end{eqnarray}
These homomorphisms are compatible in the sense that
if $k|\ell |r$, then
$\Phi_{k\ell}\circ \Phi_{\ell r}=\Phi_{kr}$. Also $\Phi_{kk}={\bf 1}$.

The $\{{\mathbb Z}(\ell),\Phi _{k\ell}\}$
is an inverse system and we call ${\widehat {\mathbb Z}}$ its inverse limit 
\begin{eqnarray}
\lim _{\longleftarrow} {\mathbb Z}(\ell)={\widehat {\mathbb Z}}.
\end{eqnarray}
Its elements can be represented as sequences
\begin{eqnarray}\label{Z}
a=(\alpha_1, \alpha_2,...);\;\;\;\;\alpha_k\in {\mathbb Z}(k);\;\;\;\;\alpha_\ell= \alpha_k ({\rm mod}\;k)
\end{eqnarray}
and addition is performed componentwise.
The Chinese remainder theorem and
the fact that  $\alpha_\ell= \alpha_k ({\rm mod}\;k)$ for $k |\ell$, show that only the elements with indices which are powers of a prime, 
are needed in order to define uniquely the sequence.
Therefore $a$ can also be written as
\begin{eqnarray}\label{ZA}
a=(a_2,a_3,a_5,...,a_p,...);\;\;\;\;a_p=(\alpha_p,\alpha _{p^2},...)\in {\mathbb Z}_p
\end{eqnarray}
$a_p$ is a $p$-adic integer (written in the representation of Eq.(\ref{pt})).

\paragraph*{Fundamental system of neighbourhoods of $0$ and the product topology:}
${\widehat {\mathbb Z}}$ is a profinite group and as such it is a topological group 
which is Hausdorff, compact and totally disconnected. The
\begin{eqnarray}
n{\widehat {\mathbb Z}}\cong \prod _{p\in \Pi} p^{e_p}{\mathbb Z}_p
\end{eqnarray}
are a fundamental system of neighbourhoods of $0$.
Here we have factorized $n$ as in Eq.(\ref{factorize10}), and therefore in the sequence $(e_p)$ all but a finite number of $e_p$, are equal to zero.
This is related to the fact that ${\widehat {\mathbb Z}}=\prod {\mathbb Z}_p$ as a topological group, has the 
product (Tychonoff) topology.
This means that the open sets in ${\widehat {\mathbb Z}}$
are $\prod U_p$ where $U_p$ is an open sets in ${\mathbb Z}_p$, and $U_p={\mathbb Z}_p$ for all but a finite number of $p$. 
This topology is important later.

\paragraph*{Projections:}
We factorize a positive integer $n$ in terms of prime numbers, as in Eq.(\ref{factorize10}).
There exist projections $\pi _n$ from ${\widehat {\mathbb Z}}$
to ${\mathbb Z}(n)$, where the $a=(a_2,a_3,a_5,...)\in {\widehat {\mathbb Z}}$ is mapped into
\begin{eqnarray}\label{BB1}
\pi _n(a)=\xi _{e_1}(a_{p_1})...\xi _{e_\ell}(a_{p_\ell})\in {\mathbb Z}(n);\;\;\;\;\;p_1,...,p_\ell\in \Pi(n);\;\;\;\;\;e_1,...,e_\ell\in E(n)
\end{eqnarray}
Here $\xi _{e_i}(a_{p_i})\in {\mathbb Z}(p_i^{e_i})$ (see Eq.(\ref{AA1})).
These projections are compatible with the $\Phi_{n\ell}$:
\begin{eqnarray}
n|\ell\;\;\rightarrow\;\;\Phi_{n\ell}\circ \pi _\ell=\pi_n
\end{eqnarray}

\subsubsection{${\mathbb Q}/{\mathbb Z}$ as direct limit}

We consider the Pontryagin dual groups to $\{{\mathbb Z}(n)\}$ used in the inverse limit of Eq.(\ref{hom3}).
They are the  
$\{{\mathbb Z}(n)\}$ (they are self-dual).
For $k |\ell$ 
we define the homomorphisms (`embeddings'):
\begin{eqnarray}\label{hom4}
&&\widetilde \Phi_{k\ell}:\;{\mathbb Z}(k)\;\rightarrow\;{\mathbb Z}(\ell);\;\;\;\;\;\;\;\;k|\ell\nonumber\\
&&\widetilde \Phi_{k\ell}(\alpha_k)=\alpha_\ell;\;\;\;\;\;\;\alpha_\ell= \frac{\ell}{k}\alpha_k;\;\;\;\;\;
\alpha_{\ell}\in {\mathbb Z}(\ell);\;\;\;\;\;\alpha_k \in {\mathbb Z}(k).
\end{eqnarray}
The ${\mathbb Z}(k)$ is a subgroup of ${\mathbb Z}(\ell)$, and for a given element in ${\mathbb Z}(k)$
the homomorphism defines the corresponding
element in ${\mathbb Z}(\ell )$.
These homomorphisms are compatible in the sense that
if $k|\ell |r$, then
$\widetilde \Phi_{\ell r}\circ \widetilde \Phi_{k\ell}=\widetilde \Phi_{kr}$. Also $\widetilde \Phi_{kk}={\bf 1}$.

The $\{{\mathbb Z}(k),{\widetilde \Phi} _{k\ell}\}$ 
is a direct system with direct limit
\begin{eqnarray}
\lim _{\longrightarrow} {\mathbb Z}(k)={\mathbb Q}/{\mathbb Z}.
\end{eqnarray}
The  direct limit is the disjoint union of all $C(k)\cong {\mathbb Z}(k)$
modulo an equivalence relation where 
we identify $\alpha_{k} \in {\mathbb Z}(k)$ with $\alpha _\ell=\frac{\ell}{k}\alpha_{k} \in {\mathbb Z}(\ell)$.
The ${\mathbb Q}/{\mathbb Z}$ is a discrete group (as the Pontryagin dual group to ${\widehat {\mathbb Z}}$ which is profinite and therefore compact).

For a positive integer $n$,
there exist homomorphisms ${\widetilde \pi} _n$ from ${\mathbb Z}(n)$ to
${\mathbb Q}/{\mathbb Z}$, where
\begin{eqnarray}\label{BB2}
{\widetilde \pi} _n (\alpha _n)=\frac{\alpha_n}{n}=({\mathfrak a}_2,{\mathfrak a}_3,{\mathfrak a}_5,...). 
\end{eqnarray}
The mechanism of expressing $\alpha_n/n$ as $({\mathfrak a}_2,{\mathfrak a}_3,{\mathfrak a}_5,...)$
is described in Eq.(\ref{468}).
The ${\widetilde \pi} _n$ are compatible in the sense that
\begin{eqnarray}
n|\ell\;\;\rightarrow\;\;\widetilde \pi _\ell \circ \widetilde \Phi_{n\ell}=\widetilde \pi_n.
\end{eqnarray}

\paragraph*{Subgroups of ${\mathbb Q}/{\mathbb Z}$ and their Pontryagin dual groups:}
Given a supernatural number $n$ factorized as in Eq.(\ref{super}), we consider the group
\begin{eqnarray}\label{x1}
{\mathfrak Z}(n)=\left [
{\mathbb Z}\left (\prod _{p\in \Pi ^{(\rm fin )}(n)}p^{e_p(n)}\right )\prod _{p\in \Pi ^{\infty }(n)}{\mathbb Z}_p\right ]
\cong \frac{\widehat {\mathbb Z}}{{\rm Ann}_{\widehat {\mathbb Z}}\widetilde {\mathfrak Z}(n)};\;\;\;\;
n\in {\mathbb N}_S,
\end{eqnarray}
and its Pontryagin dual group
\begin{eqnarray}\label{x2}
\widetilde {\mathfrak Z}(n)=\left [
C\left (\prod _{p\in \Pi ^{(\rm fin )}(n)}p^{e_p(n)}\right )\prod _{p\in \Pi ^{\infty }(n)}{\mathbb Q}_p/{\mathbb Z}_p\right ]
\le {\mathbb Q}/{\mathbb Z}.
\end{eqnarray}
We note here that
\begin{eqnarray}
{\mathbb Z}\left (\prod _{p\in \Pi ^{(\rm fin )}(n)}p^{e_p(n)}\right )\cong \prod _{p\in \Pi ^{(\rm fin )}(n)}{\mathbb Z}\left (p^{e_p(n)}\right )
\end{eqnarray}
and that Eq.(\ref{47}) has been used in Eq.(\ref{x1}).
The elements of ${\mathfrak Z}(n)$ are $(a_2,...,a_p,...)$ where $a_p=0$ for $p\in \Pi ^{(0)}(n)$.
Similarly the elements of $\widetilde {\mathfrak Z}(n)$ are $({\mathfrak a}_2,...,{\mathfrak a}_p,...)$ 
where ${\mathfrak a}_p=0$ for $p\in \Pi ^{(0)}(n)$.
In the special case that $n\in {\mathbb N}$, we get 
${\mathfrak Z}(n)\cong \widetilde {\mathfrak Z}(n)\cong {\mathbb Z}(n)$.

An example is
\begin{eqnarray}
&&{\mathfrak Z}(3^\infty 5^\infty)={\mathbb Z}_3\times {\mathbb Z}_5
\nonumber\\
&&\widetilde {\mathfrak Z}(3^\infty 5^\infty)=({\mathbb Q}_3/{\mathbb Z}_3)\times ({\mathbb Q}_5{\mathbb Z}_5).
\end{eqnarray}
The elements of ${\mathfrak Z}(3^\infty 5^\infty)$ are $(0,a_3,a_5,0,0,...)$ where $a_3 \in {\mathbb Z}_3$ and $a_5 \in {\mathbb Z}_5$.
Similarly the elements of $\widetilde {\mathfrak Z}(3^\infty 5^\infty)$ are $(0, {\mathfrak a}_3, {\mathfrak a}_5, 0,0,...)$ where 
${\mathfrak a}_3 \in {\mathbb Q}_3/{\mathbb Z}_3$ and ${\mathfrak a}_5 \in {\mathbb Q}_5/{\mathbb Z}_5$.
Also
\begin{eqnarray}
{\mathfrak Z}(\Omega)=\widehat {\mathbb Z};\;\;\;\;\widetilde {\mathfrak Z}(\Omega)={\mathbb Q}/{\mathbb Z}. 
\end{eqnarray}

\paragraph*{The topological space ${\cal Z}_S$:}
We can now extend the topological space ${\cal Z}$ in Eq.(\ref{eq2}). We consider the set
\begin{eqnarray}
{\cal Z}_S=\{\widetilde {\mathfrak Z}(n)\;|\;n\in {\mathbb N}_S\}.
\end{eqnarray}
There is a bijective map between ${\mathbb N}_S$ and ${\cal Z}_S$.
Using the methodology in section \ref{topology},
we make ${\cal Z}_S$ a topological space and 
$({\cal Z}_S,T_{{\cal Z}_S})\sim ({\mathbb N}_S, {T}_{{\mathbb N}_S})$.

\subsubsection{The directed partially ordered sets  ${\cal G}$ and ${\cal G}_S$ of Pontryagin dual pairs of groups}\label{c56}

We consider the set 
\begin{eqnarray}
{\cal G}_S=\{({\mathfrak Z}(n),\widetilde {\mathfrak Z}(n))\;|\;n\in {\mathbb N}_S\} 
\end{eqnarray}
of pairs of Pontryagin dual groups. 
We order the pairs by ordering the second elements $\widetilde {\mathfrak Z}(n)$ 
of the pairs with the relation `subgroup'. Then
\begin{eqnarray}\label{478}
&&k|\ell\;\;\rightarrow\;\;({\mathfrak Z}(k),\widetilde {\mathfrak Z}(k)) \prec
({\mathfrak Z}(\ell),\widetilde {\mathfrak Z}(\ell))
\end{eqnarray}
The set ${\cal G}_S$ is now a partially ordered set, which is order isomorphic to ${\mathbb N}_S$ (with divisibility as an order).
The first elements ${\mathfrak Z}(n)$ of these pairs, are partially ordered indirectly by the fact that their Pontryagin duals are 
partially ordered with the relation `subgroup'.
In order to find a more direct relation we can use Eq.(\ref{4}), but we are not discussing the details here.

All chains in ${\cal G}_S$ have a supremum, and therefore ${\cal G}_S$ is a directed-complete partial order.
The $(\widehat {\mathbb Z},{\mathbb Q}/{\mathbb Z})$ is maximum element in ${\cal G}_S$.
An example of a chain in ${\cal G}_S$ is 
\begin{eqnarray}\label{byu}
({\mathfrak Z}(k),\widetilde {\mathfrak Z}(k))\prec ({\mathfrak Z}(k^2),\widetilde {\mathfrak Z}(k^2))
\prec...\prec \left (
\prod _{p\in \Pi -\Pi ^{(0)}(k)}{\mathbb Z}_p,
\prod _{p\in \Pi -\Pi ^{(0)}(k)}{\mathbb Q}_p/{\mathbb Z}_p
\right ).
\end{eqnarray}
where $k\in {\mathbb N}$.
The supremum of this chain is the pair in the right hand side. 
Another chain in ${\cal G}_S$ is
\begin{eqnarray}\label{byu1}
({\mathfrak Z}(2^\infty),\widetilde {\mathfrak Z}(2^\infty))\prec 
({\mathfrak Z}(2^\infty 3^\infty),\widetilde {\mathfrak Z}(2^\infty 3^\infty))
\prec...\prec (\widehat {\mathbb Z},{\mathbb Q}/{\mathbb Z})
\end{eqnarray}
The $(\widehat {\mathbb Z},{\mathbb Q}/{\mathbb Z})$ is the supremum of this chain.

A subset of ${\cal G}_S$ is
\begin{eqnarray}
{\cal G}=\{({\mathbb Z}(n),{\mathbb Z}(n))\;|\;n\in {\mathbb N}\} 
\end{eqnarray}
This is a directed partially ordered set, , which is order isomorphic to ${\mathbb N}$. It is  not  a directed-complete partial order, and it
has no maximal elements.

\paragraph*{The topological spaces $({\cal G}_S, T_{{\cal G}_S})$ and $({\cal G}, T_{{\cal G}})$:}
There is a bijective map between ${\mathbb N}_S$ and ${\cal G}_S$: 
\begin{eqnarray}
&&n\;\;\leftrightarrow \;\;({\mathfrak Z}(n),\widetilde {\mathfrak Z}(n))\nonumber\\
&&\Omega \;\;\leftrightarrow \;\;({\widehat {\mathbb Z}},{\mathbb Q}/{\mathbb Z}).
\end{eqnarray}
Therefore we can make ${\cal G}_S$ a topological space and $({\mathbb N}_S, T_{{\mathbb N}_S})\sim ({\cal G}_S, T_{{\cal G}_S})$ 
(see section \ref{topology}).
In a similar way we make ${\cal G}$ a topological space and $({\mathbb N}, T_{{\mathbb N}})\sim ({\cal G}, T_{{\cal G}})$.
According to proposition \ref{ttt}, $({\cal G}_S, T_{{\cal G}_S})$ is a $T_0$-space which is compact,
and $({\cal G}, T_{{\cal G}})$ is a $T_0$-space which is locally compact.

\paragraph*{Embeddings:}
We consider the products ${\mathfrak Z}(k)\times \widetilde {\mathfrak Z}(k)$ 
which will be used later as phase spaces of quantum systems.
We define embeddings among them, using the 
inverse of a restriction of $\Phi _{k\ell}$ in Eq.(\ref{hom3}),
and the homomorphism $\widetilde \Phi _{k\ell}$
of Eq.(\ref{hom4}).
\begin{definition}\label{def2}
For $k|\ell$, ${\bf E}_{k\ell}$ is the embedding of
${\mathfrak Z}(k)\times \widetilde {\mathfrak Z}(k)$ into 
${\mathfrak Z}(\ell)\times \widetilde {\mathfrak Z}(\ell)$,
with the map
\begin{eqnarray}
&&{\bf E}_{k\ell}:\;\;(x,{\mathfrak p})\;\;\rightarrow\;\;(x',{\mathfrak p}');\;\;\;\;k|\ell\nonumber\\
&&x=(x_2,...,x_p,...)\in {\mathfrak Z}(k);\;\;\;\;
{\mathfrak p}=({\mathfrak p}_2,...,{\mathfrak p}_p,...)\in \widetilde {\mathfrak Z}(k)
\nonumber\\
&&x'=(x_2',...,x_p',...)\in {\mathfrak Z}(\ell);\;\;\;\;
{\mathfrak p}'=({\mathfrak p}_2',...,{\mathfrak p}_p',...)\in \widetilde {\mathfrak Z}(\ell),
\end{eqnarray}
where:
\begin{itemize}
\item[(1)]
if $p\in \Pi ^{(0)}(k)$, then $x_p'=0$ and ${\mathfrak p}_p'=0$.
\item[(2)]
if $p\in \Pi ^{({\rm fin})}(k) \cap \Pi ^{({\rm fin})}(\ell)$
\begin{eqnarray}
&&x_p'=x_p\;\;\;\;\;\;{\mathfrak p}_p'={\mathfrak p}_p p^{e_p(\ell)-e_p(k)}\nonumber\\
&&x_p, {\mathfrak p}_p\in {\mathbb Z}(p^{e_p(k)});\;\;\;\;x_p', {\mathfrak p}_p' \in {\mathbb Z}(p^{e_p(\ell)}),
\end{eqnarray}
as in Eq.(\ref{bn1}).
\item[(3)]
if $p\in \Pi ^{({\rm fin})}(k) \cap \Pi ^{({\infty})}(\ell)$ then 
\begin{eqnarray}
&&x_p'=x_p;\;\;\;\;\;{\mathfrak p}_p'={\mathfrak p}_p p^{-e_p(k)}\nonumber\\
&&x_p ,{\mathfrak p}_p \in {\mathbb Z}(p^{e_p});\;\;\;\;x_p'\in {\mathbb Z}_p;\;\;\;\;\;
{\mathfrak p}_p'\in {\mathbb Q}_p/{\mathbb Z}_p,
\end{eqnarray}
as in Eq.(\ref{bn2}).
\item[(4)]
if $p\in \Pi ^{(\infty )}(k)$ then $x_p'=x_p$ and ${\mathfrak p}_p'={\mathfrak p}_p$.
\end{itemize}
\end{definition}
The relations ${\mathfrak p}_p'={\mathfrak p}_p p^{e_p(\ell)-e_p(k)}$
and ${\mathfrak p}_p'={\mathfrak p}_p p^{-e_p(k)}$ above, are basically the 
homomorphism $\widetilde \Phi _{k\ell}$.
The relation $x_p'=x_p$ is the inverse of a restriction of $\Phi _{k\ell}$.
The $\Phi _{k\ell}$ is not invertible because it is not injective, but here we use a restriction of $\Phi _{k\ell}$ which is invertible.

It is easily seen that 
\begin{eqnarray}
\chi (x{\mathfrak p})=\chi (x'{\mathfrak p}').
\end{eqnarray}
The embeddings ${\bf E}_{k\ell}$ are compatible:
\begin{eqnarray}
k|\ell |m\;\;\rightarrow\;\;{\bf E}_{\ell m} \circ {\bf E}_{k\ell}={\bf E}_{km}.
\end{eqnarray}

\begin{remark}
As in remark \ref{re}, we note here that the map $x_p'=x_p$ is different from the map ${\mathfrak p}_p'={\mathfrak p}_p p^{-e_p(k)}$
(the factor $p^{-e_p(k)}$). We comment on this again in remark \ref{last}.
\end{remark}

\section{The system $\Sigma[{\mathbb Z}_p, ({\mathbb Q}_p/{\mathbb Z}_p)]$}

\subsection{Locally constant functions and functions with compact support}

\begin{definition}
Let $f_p(a_p)$ be a complex function of $a_p\in {\mathbb Q}_p$.
\begin{itemize}
\item[(1)]
$f_p(a_p)$ is locally constant with degree $n$, if $f_p(a_p+b_p)=f_p(a_p)$ for all $|b_p|_p\le p^{-n}$.
\item[(2)]
$f_p(a_p)$ has compact  support with degree $k$, if $f_p(a_p)=0$ for all $|a_p|_p> p^k$.
\end{itemize}
\end{definition}
\begin{remark}
A function $f_p(a_p)$ which is locally constant with degree $n$ and has compact  support with degree $k$,
is effectively defined on $p^{-k}{\mathbb Z}_p/p^{n}{\mathbb Z}_p\cong {\mathbb Z}(p^{n+k})$.
\end{remark}

\subsection{Integrals}
In integrals of functions over ${\mathbb Q}_p$ we use the Haar measure, normalized as:
\begin{eqnarray}\label{86}
\int _{{\mathbb Z}_p}da_p=1. 
\end{eqnarray}
The integral  over ${\mathbb Q}_p$, of a function $f_p(a_p)$
which is locally constant with degree $n$ and has compact  support with degree $k$,
is given by
\begin{eqnarray}\label{877}
\int _{{\mathbb Q}_p}f_p(a_p)da_p=p^{-n}
\sum f_p({\overline a}_{-k}p^{-k}+...+{\overline a}_{n-1}p^{n-1}).
\end{eqnarray}
The sum is here finite and it is over all ${\overline a}_{-k},...,{\overline a}_{n-1}$.
If we truncate the sum at $n'>n$, we get the same result, because 
the function is locally constant with degree $n$.
If we truncate the sum at $k'>k$, we get the same result because the function 
has constant  support with degree $k$.

\paragraph*{Change of variables:}

Let $f_p(a_p)$ be a function which is locally constant with degree $n$ and has compact  support with degree $k$.
We express $\lambda \in {\mathbb Z}^+$ as $\lambda=p^r\lambda _1$, where $p$ and $\lambda _1$ are coprime, and 
$|\lambda |_p=p^{-r}$.
Then the function $f'_p(a_p)=f_p(\lambda a_p)$ is
locally constant with degree $n-r$ and has compact  support with degree $k+r$. From this, and the relation
\begin{eqnarray}
\int _{{\mathbb Q}_p}f_p(a_p)da_p=\int _{{\mathbb Q}_p}f_p'(a_p)d(\lambda a_p),
\end{eqnarray}
follows that
\begin{eqnarray}\label{cha1}
d(\lambda a_p)=|\lambda |_p da_p=p^{-r} da_p;\;\;\;\;\;\;a_p\in {\mathbb Q}_p.
\end{eqnarray}
If $\lambda, p$ are coprime then $d(\lambda a_p)=da_p$. 

\paragraph*{Integrals over ${\mathbb Q}_p/{\mathbb Z}_p$:} 
Let $F_p({\mathfrak p}_p)$ be a complex function of ${\mathfrak p}_p \in {\mathbb Q}_p/{\mathbb Z}_p$, which has compact support with degree $k$.
Then
\begin{eqnarray}\label{89}
\int _{{\mathbb Q}_p/{\mathbb Z}_p}F_p({\mathfrak p}_p)d{\mathfrak p}_p=\sum
F_p({\overline {\mathfrak p}}_{-k}p^{-k}+{\overline {\mathfrak p}}_{-k+1}p^{-k+1}+...
+{\overline {\mathfrak p}}_{-1}p^{-1}).
\end{eqnarray}
The sum is finite and it is over all ${\overline {\mathfrak p}}_{-k},...,{\overline {\mathfrak p}}_{n-1}$.
The counting measure is used here. 
The function $F_p({\mathfrak p}_p)$ can be regarded 
as a functions $F_p(u_p)$ over ${\mathbb Q}_p$ which is
periodic:
\begin{eqnarray}\label{www10}
F_p(u_p+1)=F_p(u_p);\;\;\;\;\;\;u_p\in {\mathbb Q}_p. 
\end{eqnarray}
We write $u_p={\mathfrak p}_p+x_p$ where 
${\mathfrak p}_p \in {\mathbb Q}_p/{\mathbb Z}_p$ and $x_p \in {\mathbb Z}_p$.
Here we represent the coset ${\mathfrak p}_p$ with the element that has zero integer part.
Then the $F_p(u_p)$ does not depend on $x_p$ and integration of $F_p(u_p)$ over ${\mathbb Q}_p$ is given by
\begin{eqnarray}\label{n7}
\int _{{\mathbb Q}_p} du_p F_p(u_p)=
\int _{{\mathbb Q}_p/{\mathbb Z}_p}d{\mathfrak p}_p \int _{{\mathbb Z}_p}dx_p\;F_p({\mathfrak p}_p)=
\int _{{\mathbb Q}_p/{\mathbb Z}_p}d{\mathfrak p}_pF_p({\mathfrak p}_p).
\end{eqnarray}
The counting measure for integration over ${\mathbb Q}_p/{\mathbb Z}_p$ ensures that this relation holds.
Eq.(\ref{n7}) can be regarded as a Weil transform \cite{N3} which in the present context, takes functions from 
${\mathbb Q}_p$ to ${\mathbb Q}_p/{\mathbb Z}_p$.

\paragraph*{Change of variables:}
As above a change of variables is performed with the relation 
\begin{eqnarray}
d(\lambda {\mathfrak a}_p)=|\lambda |_p d{\mathfrak a}_p;\;\;\;\;\;{\mathfrak a}_p\in {\mathbb Q}_p/{\mathbb Z}_p;\;\;\;\;\;\lambda \in {\mathbb Z}^+
\end{eqnarray}
The following formula for the case $p=2$, is needed later: 
\begin{eqnarray}
2^{-n}\int _{{\mathbb Q}_2/2^{-n}{\mathbb Z}_2}d{\mathfrak a}_2F(2^n{\mathfrak a}_2)=\int _{{\mathbb Q}_2/{\mathbb Z}_2}d(2^n{\mathfrak a}_2)F(2^n{\mathfrak a}_2)=
\int _{{\mathbb Q}_2/{\mathbb Z}_2}d{\mathfrak a}_2'F({\mathfrak a}_2')
\end{eqnarray}
Here the ${\mathfrak a}_2$ takes values in ${\mathbb Q}_2/2^{-n}{\mathbb Z}_2$, and therefore the
$2^n{\mathfrak a}_2$ takes values in ${\mathbb Q}_2/{\mathbb Z}_2$.
We rewrite this relation in a way that it is valid for any $p$ as
\begin{eqnarray}\label{567}
|2^{n}|_p\int _{{\mathbb Q}_p/|2^{n}|_p{\mathbb Z}_p}d{\mathfrak a}_pF(2^n{\mathfrak a}_p)=\int _{{\mathbb Q}_p/{\mathbb Z}_p}d(2^n{\mathfrak a}_p)F(2^n{\mathfrak a}_p)=
\int _{{\mathbb Q}_p/{\mathbb Z}_p}d{\mathfrak a}_p'F({\mathfrak a}_p')
\end{eqnarray}
We recall here that for $p\ne 2$, we have $|2^{n}|_p=1$ and also $d(2^n{\mathfrak a}_p)=d({\mathfrak a}_p)$.

\paragraph*{Delta functions:}
In calculations it is useful to have delta functions.
$\delta _p(x_p)$ where $x_p\in {\mathbb Z}_p$ is a generalized function
(it is not locally constant).
Generalized functions in the present context are discussed in \cite{b7}.
We note that 
\begin{eqnarray}\label{44}
\delta_p(\lambda x_p)=\frac{\delta_p(x_p)}{|\lambda |_p};\;\;\;\;\;\lambda \in {\mathbb Z}^+.
\end{eqnarray}
We also introduce the function $\Delta _p({\mathfrak p}_p)$ where
${\mathfrak p}_p\in {\mathbb Q}_p/{\mathbb Z}_p$, such that $\Delta _p(0)=1$
and $\Delta _p({\mathfrak p}_p)=0$ if ${\mathfrak p}_p\ne 0$.
This is not a generalized function.
Then
\begin{eqnarray}
\int _{{\mathbb Z}_p} dx_pf _p(x_p)\delta_p(x_p-a_p)=f_p(a_p);\;\;\;\;\;
\int _{{\mathbb Q}_p/{\mathbb Z}_p} d{\mathfrak p}_pF _p({\mathfrak p}_p) \Delta _p({\mathfrak p}_p-{\mathfrak a}_p)
=F _p({\mathfrak a}_p).
\end{eqnarray}
Also
\begin{eqnarray}\label{44}
\int _{{\mathbb Z}_p} dx_p\chi _p(x_p{\mathfrak p}_p)=\Delta _p({\mathfrak p}_p);\;\;\;\;\;
\int _{{\mathbb Q}_p/{\mathbb Z}_p} d{\mathfrak p}_p\chi_p(x_p{\mathfrak p}_p)
=\delta_p(x_p).
\end{eqnarray}

\subsection{Fourier transforms}

The Fourier transform of the complex function $f_p(x_p)$ where $x_p\in {\mathbb Q}_p$, is given by
\begin{eqnarray}\label{qa1}
[{\mathfrak F}_pf_p](y_p)={\widetilde f}_p(y_p)=\int _{{\mathbb Q}_p}dx_p
\chi _p(-x_py_p)f_p(x_p);\;\;\;\;\;y_p\in {\mathbb Q}_p.
\end{eqnarray} 
\begin{proposition}
If $f_p(x_p)$ is locally constant with degree $n$ then its Fourier transform
${\widetilde f}_p(y_p)$ has compact support with the same degree $n$.
Also, if $f_p(x_p)$ has compact support with degree $k$ then 
${\widetilde f}_p(y_p)$  is locally constant with the same degree $k$.
\end{proposition}
\begin{proof}
If the function $f_p(x_p)$ is locally constant
with degree $n$, then for all $|\alpha _p|_p\le p^{-n}$ we get
\begin{eqnarray}
f_p(x_p+a_p)-f_p(x_p)=0\;\;\rightarrow\;\;
\int _{{\mathbb Q}_p/{\mathbb Z}_p}
d{\mathfrak p}_p\;\chi _p(x_p{\mathfrak p}_p){\tilde f}_p({\mathfrak p}_p)[1-\chi _p (\alpha _p{\mathfrak p}_p)]=0.
\end{eqnarray}
Here the Fourier transform of ${\tilde f}_p({\mathfrak p}_p)[1-\chi _p(\alpha _p{\mathfrak p}_p)]$ is zero, and therefore
${\tilde f}_p({\mathfrak p}_p)[1-\chi _p(\alpha _p{\mathfrak p}_p)]=0$.
But for $|\alpha _p|_p\le p^{-n}$ and $|{\mathfrak p}_p|> p^n$ the $1-\chi _p(\alpha _p{\mathfrak p}_p)\ne 0$
and therefore ${\tilde f}_p({\mathfrak p}_p)=0$.
This proves that the function ${\tilde f}_p({\mathfrak p}_p)$ has compact support
with degree $n$.

In a similar way we prove the second part of the proposition.

\end{proof}
Therefore, if $f_p(x_p)$ is locally constant with degree $n$
and has compact support with degree $k$ (in which case it can be regarded as a function 
on $p^{-k}{\mathbb Z}_p/p^{n}{\mathbb Z}_p\cong {\mathbb Z}(p^{n+k})$),
its Fourier transform ${\widetilde f}_p(y_p)$  is locally constant with degree $k$
and has compact support with degree $n$ (in which case it can be regarded as a function 
on $p^{-n}{\mathbb Z}_p/p^{k}{\mathbb Z}_p\cong {\mathbb Z}(p^{n+k})$).

\subsection{The Schwartz-Bruhat space ${\mathfrak B}[{\mathbb Z}_p, ({\mathbb Q}_p/{\mathbb Z}_p)]$}\label{L1}

We consider a quantum system in which the position variable $x_p$ takes values in ${\mathbb Z}_p$  
and the momentum variable ${\mathfrak p}_p$ takes values in the Pontryagin dual group 
${\mathbb Q}_p/{\mathbb Z}_p$. Its wavefunctions belong to the
Schwartz-Bruhat space ${\mathfrak B}[{\mathbb Z}_p, ({\mathbb Q}_p/{\mathbb Z}_p)]$ \cite{b1,b2,b3}, which 
is defined as follows:

\begin{definition}\label{def3}
The Schwartz-Bruhat space ${\mathfrak B}[{\mathbb Z}_p, ({\mathbb Q}_p/{\mathbb Z}_p)]$
can be defined by one of the following two ways, which are related through a Fourier transform, and are equivalent to each other:
\begin{itemize}
\item[(1)]
It consists of locally constant complex functions  $f_p(x_p)$ (where $x_p\in {\mathbb Z}_p$).
The fact that these functions are defined in ${\mathbb Z}_p$ implies that they also
have compact support with degree $k\le 0$.
The scalar product is given by
\begin{eqnarray}\label{qa1}
(f_p,g_p)=\int_{{\mathbb Z}_p}[f_p(x_p)]^*g_p(x_p)dx_p.
\end{eqnarray}  
\item[(2)]
It consists of complex functions  with compact support $F_p({\mathfrak p}_p)$ 
(where ${\mathfrak p}_p\in  {\mathbb Q}_p/{\mathbb Z}_p$).
The fact that these functions are defined in ${\mathbb Q}_p/{\mathbb Z}_p$ implies that 
they are also locally constant with degree $n\le 0$.
The scalar product of two such functions is given by
\begin{eqnarray}\label{qa2}
(F_p,G_p)=
\int _ {{\mathbb Q}_p/{\mathbb Z}_p}[F_p({\mathfrak p}_p)]^*
G_p({\mathfrak p}_p)d{\mathfrak p}_p.
\end{eqnarray} 

\end{itemize}
\end{definition}

The Fourier transform is given by 
\begin{eqnarray}\label{qa3}
({\mathfrak F}_pf_p)({\mathfrak p}_p)={\widetilde f}_p({\mathfrak p}_p) =\int _{{\mathbb Z}_p}dx_pf_p(x_p)\chi _p(-x_p{\mathfrak p}_p).
\end{eqnarray} 
The inverse Fourier transform is 
\begin{eqnarray}\label{qa4}
({\mathfrak F}_p^{-1}\widetilde f_p)(x_p)=f_p(x_p)=\int _{{\mathbb Q}_p/{\mathbb Z}_p}d{\mathfrak p}_p{\widetilde f}_p({\mathfrak p}_p)\chi _p(x_p{\mathfrak p}_p).
\end{eqnarray} 

According to Parceval's theorem $(f_p,g_p)=({\widetilde f}_p, {\widetilde g}_p)$. Also ${\mathfrak F}_p^4={\bf 1}$.

\subsubsection{Time evolution}

In general, time evolution can be described with unitary operators ${\cal U}(x_p,y_p;t)$:
\begin{eqnarray}\label{ev1}
\int _{{\mathbb Z}_p}dy_p\;{\cal U}(x_p,y_p;t)[{\cal U}(x'_p,y_p;t)]^*=\delta _p(x_p-x'_p).
\end{eqnarray}
Here $t\in {\mathbb R}$ is the time and $x_p,x'_p\in {\mathbb Z}_p$.
Then the wavefunction $f_p(x_p)$ evolves in time as 
\begin{eqnarray}\label{ev2}
f_p(x_p;t)=\int _{{\mathbb Z}_p}dy_p\;{\cal U}(x_p,y_p;t)f_p(y_p).
\end{eqnarray}  
In ref.\cite{r11} we have studied in the present context, evolution operators 
which are analogous to $[\exp(i\alpha \hat x)\exp(i\beta \hat p)]^t$
in the harmonic oscillator.

\subsection{The Heisenberg-Weyl group ${\bf HW}[({\mathbb Q}_p/{\mathbb Z}_p),
{\mathbb Z}_p, ({\mathbb Q}_p/{\mathbb Z}_p)]$}

The phase space of this system is ${\mathbb Z}_p \times ({\mathbb Q}_p/{\mathbb Z}_p)$.
The displacement operators $D_p({\mathfrak a}_p,b_p,{\mathfrak c}_p)$ 
act on the wavefunctions $f_p(x_p)$ and $F_p({\mathfrak p}_p)$, as follows:
\begin{eqnarray}\label{450}      
&&[D_p({\mathfrak a}_p,b_p,{\mathfrak c}_p)f_p](x_p)=\chi _p\left ({\mathfrak c}_p-{\mathfrak a}_pb_p+
2{\mathfrak a}_px_p\right )f_p(x_p-b_p)\nonumber\\
&&[D_p({\mathfrak a}_p,b_p,{\mathfrak c}_p)F_p]({\mathfrak p}_p)=
\chi _p\left ({\mathfrak c}_p+{\mathfrak a}_pb_p-b_p{\mathfrak p}_p\right )
F_p({\mathfrak p}_p-2{\mathfrak a}_p)
\nonumber\\
&&x_p,b_p\in {\mathbb Z}_p;\;\;\;\;\;{\mathfrak p}_p, {\mathfrak a}_p, {\mathfrak c}_p \in
{\mathbb Q}_p/{\mathbb Z}_p
\end{eqnarray}
It is easily proved that these two relations are equivalent to each other, and therefore they can be used as 
definition of the displacement operators.
It is also seen that 
\begin{eqnarray}\label{fth}
D_p({\mathfrak a}_p+1,b_p,{\mathfrak c}_p)=D_p({\mathfrak a}_p,b_p,{\mathfrak c}_p)
=D_p({\mathfrak a}_p,b_p,{\mathfrak c}_p+1),
\end{eqnarray}
and this is consistent with the fact that ${\mathfrak a}_p, {\mathfrak c}_p \in {\mathbb Q}_p/{\mathbb Z}_p$.

We next prove the multiplication rule
\begin{eqnarray}\label{137}
D_p({\mathfrak a}_p,b_p,{\mathfrak c}_p)D_p({\mathfrak a}'_p,b'_p,{\mathfrak c}'_p)=
D_p({\mathfrak a}_p+{\mathfrak a}'_p,b_p+b'_p,{\mathfrak c}_p+{\mathfrak c}'_p+({\mathfrak a}_pb'_p-{\mathfrak a}'_pb_p))
\end{eqnarray} 
Therefore these displacement operators form a representation 
of the Heisenberg-Weyl group, for which we use the notation
${\bf HW}({\mathbb Q}_p/{\mathbb Z}_p, {\mathbb Z}_p, {\mathbb Q}_p/{\mathbb Z}_p)$
which indicates clearly where the three variables belong.
It is easily seen that
\begin{eqnarray}\label{33}     
[D_p({\mathfrak a}_p,b_p,{\mathfrak c}_p)]^{\dagger}=D_p(-{\mathfrak a}_p,-b_p,-{\mathfrak c}_p);\;\;\;\;\;\;
D_p({\mathfrak a}_p,b_p,{\mathfrak c}_p)[D_p({\mathfrak a}_p,b_p,{\mathfrak c}_p)]^{\dagger}={\bf 1}
\end{eqnarray}
The following are subgroups of ${\bf HW}({\mathbb Q}_p/{\mathbb Z}_p, {\mathbb Z}_p, {\mathbb Q}_p/{\mathbb Z}_p)$:
\begin{eqnarray}\label{30}     
&&{\bf HW}_1({\mathbb Q}_p/{\mathbb Z}_p)=\{D_p({\mathfrak a}_p,0,0)\;|\;{\mathfrak a}_p \in {\mathbb Q}_p/{\mathbb Z}_p\}
\cong {\mathbb Q}_p/{\mathbb Z}_p\nonumber\\
&&{\bf HW}_2({\mathbb Z}_p)=\{D_p(0,b_p,0)\;|\;b_p\in {\mathbb Z}_p\}\cong {\mathbb Z}_p\nonumber\\
&&{\bf HW}_3({\mathbb Q}_p/{\mathbb Z}_p)=\{D_p(0,0,{\mathfrak c}_p)\;|\;{\mathfrak c}_p \in {\mathbb Q}_p/{\mathbb Z}_p\}\cong {\mathbb Q}_p/{\mathbb Z}_p.
\end{eqnarray}
${\bf HW}_2({\mathbb Z}_p)$ is a profinite group.
${\bf HW}_1({\mathbb Q}_p/{\mathbb Z}_p)$ and ${\bf HW}_3({\mathbb Q}_p/{\mathbb Z}_p)$ are discrete groups.
${\bf HW}_3({\mathbb Q}_p/{\mathbb Z}_p)$ is a normal subgroup of ${\bf HW}({\mathbb Q}_p/{\mathbb Z}_p, {\mathbb Z}_p, {\mathbb Q}_p/{\mathbb Z}_p)$.
\begin{proposition}\label{local}
${\bf HW}[({\mathbb Q}_p/{\mathbb Z}_p), {\mathbb Z}_p, ({\mathbb Q}_p/{\mathbb Z}_p)]$ is a locally compact topological group. 
\end{proposition}
\begin{proof}
The ${\bf HW}_1({\mathbb Q}_p/{\mathbb Z}_p)$ is a discrete topological group and therefore it is locally compact.
${\bf HW}_2({\mathbb Z}_p)\cong {\mathbb Z}_p$ is a profinite group and therefore it is a compact topological group.
Therefore the ${\bf HW}_1({\mathbb Q}_p/{\mathbb Z}_p)\times {\bf HW}_2({\mathbb Z}_p)$ 
with the product topology, is a locally compact topological group.

The ${\bf HW}_3({\mathbb Q}_p/{\mathbb Z}_p)$ is a normal subgroup of ${\bf HW}({\mathbb Q}_p/{\mathbb Z}_p, {\mathbb Z}_p, {\mathbb Q}_p/{\mathbb Z}_p)$
and
\begin{eqnarray}
\frac{{\bf HW}({\mathbb Q}_p/{\mathbb Z}_p, {\mathbb Z}_p, {\mathbb Q}_p/{\mathbb Z}_p)}{{\bf HW}_3({\mathbb Q}_p/{\mathbb Z}_p)} \cong 
{\bf HW}_1({\mathbb Q}_p/{\mathbb Z}_p)\times {\bf HW}_2({\mathbb Z}_p)
\end{eqnarray}
From the fact that both ${\bf HW}_1({\mathbb Q}_p/{\mathbb Z}_p)\times {\bf HW}_2({\mathbb Z}_p)$ and 
${\bf HW}_3({\mathbb Q}_p/{\mathbb Z}_p)$ are locally compact topological groups, follows that 
${\bf HW}({\mathbb Q}_p/{\mathbb Z}_p, {\mathbb Z}_p, {\mathbb Q}_p/{\mathbb Z}_p)$ is a locally compact topological group.
\end{proof} 
\begin{definition}
${\mathfrak d}_2(r_2)$ is the generalized function
\begin{eqnarray}
{\mathfrak d}_2(r_2)=\int _{{\mathbb Q}_2/{\mathbb Z}_2}d{\mathfrak p}_2 \chi _2(r_2{\mathfrak p}_2)
;\;\;\;\;r_2\in 2^{-1}{\mathbb Z}_2.
\end{eqnarray}
If ${\rm ord}(r_2)\ge 0$ (i.e., $r_2\in {\mathbb Z}_2$) then this function is $\delta _2(r_2)$.
If ${\rm ord}(r_2)=-1$ then
\begin{eqnarray}
{\mathfrak d}_2(r_2)=\int _{{\mathbb Q}_2/{\mathbb Z}_2}d{\mathfrak p}_2 \chi _2(r_2'{\mathfrak p}_2)
\chi _2(2^{-1}{\mathfrak p}_2);\;\;\;r_2=2^{-1}+r_2';\;\;\;\;r_2'\in {\mathbb Z}_2.
\end{eqnarray}
\end{definition}
\begin{definition}
Let $f_2(x_2)$ be a function in  ${\mathfrak B}[{\mathbb Z}_2, ({\mathbb Q}_2/{\mathbb Z}_2)]$, 
where $x_2\in {\mathbb Z}_2$.
Its transform ${\widehat f}_2(2^{-1}y_2)$ (where $y_2\in {\mathbb Z}_2$ and therefore
$2^{-1}y_2\in 2^{-1}{\mathbb Z}_2$) is defined as
\begin{eqnarray}\label{126}
{\widehat f}_2(2^{-1}y_2)=\int _{{\mathbb Z}_2}dx_2\;f_2(x_2) {\mathfrak d}_2(2^{-1}y_2-x_2);\;\;\;\;
y_2\in {\mathbb Z}_2
\end{eqnarray}
\end{definition}
\begin{remark}\label{rema}
If ${\rm ord}(y_2)\ge 1$ then ${\widehat f}_2(2^{-1}y_2)=f_2(2^{-1}y_2)$.
If ${\rm ord}(y_2)=0$ then
\begin{eqnarray}\label{997}
{\widehat f}_2(2^{-1}y_2)=\int _{{\mathbb Q}_2/{\mathbb Z}_2}d{\mathfrak p}_2 \chi _2(2^{-1}y_2{\mathfrak p}_2)\widetilde f_2({\mathfrak p}_2)
\end{eqnarray}
\end{remark}
\begin{definition}
An operator $\theta$ acts on a function $f_p(x_p)$ as
\begin{eqnarray}
(\theta f_p)(x_p)=\int _{{\mathbb Z}_p}dy_p \theta (x_p,y_p) f_p(y_p)
\end{eqnarray}
It also acts on its Fourier transform $\widetilde f_p({\mathfrak p}_p)$ as
\begin{eqnarray}
(\theta \widetilde f_p)({\mathfrak p}_p)
=\int _{{\mathbb Q}_p/{\mathbb Z}_p}d{\mathfrak p}'_p\widetilde \theta ({\mathfrak p}_p,{\mathfrak p}'_p)
\widetilde f_p({\mathfrak p}'_p) 
\end{eqnarray}
where
\begin{eqnarray}
\widetilde \theta ({\mathfrak p}_p,{\mathfrak p}'_p)
=\int _{{\mathbb Z}_p}dx_p \int _{{\mathbb Z}_p}dy_p \theta (x_p,y_p)
\chi_p[x_p({\mathfrak p}_p-{\mathfrak p}'_p)].
\end{eqnarray}
The trace of $\theta$ is given by
\begin{eqnarray}
{\rm tr}\theta=\int _{{\mathbb Z}_p}dy_p \theta (x_p,x_p)=
\int _{{\mathbb Q}_p/{\mathbb Z}_p}d{\mathfrak p}'_p\widetilde \theta ({\mathfrak p}_p,{\mathfrak p}_p)
\end{eqnarray}
\end{definition}
The following lemma gives a property of the integrands for $p=2$, in the integrals in
the proposition below. 
\begin{lemma}\label{lemma}
Let
\begin{eqnarray}
&&S({\mathfrak a}_2,b_2)=D_2({\mathfrak a}_2, b_2,0)\;\theta \;[D_2({\mathfrak a}_2, b_2,0)]^{\dagger}\nonumber\\
&&R({\mathfrak a}_2,b_2)=D_2({\mathfrak a}_2, b_2,0){\rm tr} [D_2(-{\mathfrak a}_2, -b_2,0)\theta]
\end{eqnarray}
where $\theta$ is a trace class operator.
Then
\begin{itemize}
\item[(1)]
\begin{eqnarray}\label{evenodd}
&&{\rm ord}(b_2)\ge 1\;\;\rightarrow\;\;D_2({\mathfrak a}_2+2^{-1},b_2,{\mathfrak c}_2)=D_2({\mathfrak a}_2,b_2,{\mathfrak c}_2)\nonumber\\
&&{\rm ord}(b_2)=0\;\;\rightarrow\;\;D_2({\mathfrak a}_2+2^{-1},b_2,{\mathfrak c}_2)=-D_2({\mathfrak a}_2,b_2,{\mathfrak c}_2).
\end{eqnarray}
\item[(2)]
\begin{eqnarray}
S({\mathfrak a}_2+2^{-1},b_2)=S({\mathfrak a}_2+2^{-1},b_2);\;\;\;\;\;R({\mathfrak a}_2+2^{-1},b_2)=R({\mathfrak a}_2+2^{-1},b_2)
\end{eqnarray}
\end{itemize}
\end{lemma}
The proof is straightforward.
In view of this lemma, it is not surprising that integration of ${\mathfrak a}_p$ below is over ${\mathbb Q}_p/|2|_p{\mathbb Z}_p$, i.e.,
for $p\ne 2$ it is over ${\mathbb Q}_p/{\mathbb Z}_p$ and for $p=2$ it is over ${\mathbb Q}_2/2^{-1}{\mathbb Z}_2$.
\begin{proposition}\label{di}
\mbox{}
\begin{itemize}
\item[(1)]
For any trace class operator $\theta$ acting on ${\mathfrak B}[{\mathbb Z}_p, ({\mathbb Q}_p/{\mathbb Z}_p)]$
\begin{eqnarray}\label{150AA}
|2|_p\int _{{\mathbb Q}_p/|2|_p{\mathbb Z}_p}d{\mathfrak a}_p\;\int _{{\mathbb Z}_p}db_p\;
D_p({\mathfrak a}_p, b_p,0)\;\theta \;[D_p({\mathfrak a}_p, b_p,0)]^{\dagger}={\bf 1}{\rm tr}\theta.
\end{eqnarray}
\item[(2)]
A trace class operator $\theta$ acting on ${\mathfrak B}[{\mathbb Z}_p, ({\mathbb Q}_p/{\mathbb Z}_p)]$, 
can be expanded in terms of displacement operators, as
\begin{eqnarray}\label{XAA}
\theta=|2|_p\int _{{\mathbb Q}_p/|2|_p{\mathbb Z}_p}d{\mathfrak a}_p\;\int _{{\mathbb Z}_p}db_p\;
D_p({\mathfrak a}_p, b_p,0){\rm tr} [D_p(-{\mathfrak a}_p, -b_p,0)\theta]
\end{eqnarray}
\item[(3)]
Let $A_p({\mathfrak a}_p)$ and $B_p(b_p)$ be the `marginal operators'
\begin{eqnarray}\label{t950}
A_p({\mathfrak a}_p)=\int _{{\mathbb Z}_p}db_p\;D_p({\mathfrak a}_p, b_p,0);\;\;\;\;\;
B_p(b_p)=|2|_p\int _{{\mathbb Q}_p/|2|_p{\mathbb Z}_p}d{\mathfrak a}_p\;D_p({\mathfrak a}_p, b_p,0).
\end{eqnarray}
For any $\widetilde g_p({\mathfrak p}_p),\widetilde f_p({\mathfrak p}_p)$ in 
${\mathfrak B}[{\mathbb Z}_p, ({\mathbb Q}_p/{\mathbb Z}_p)]$,  
\begin{itemize}
\item[(i)]
\begin{eqnarray}\label{951cs}
\left (\widetilde g_p,A_p({\mathfrak a}_p)\widetilde f_p\right )=[{\widetilde g}_p({\mathfrak a}_p)]^*\;
{\widetilde f}_p(-{\mathfrak a}_p);\;\;\;\;{\mathfrak a}_p\in {\mathbb Z}_p.
\end{eqnarray}
\item[(ii)]
\begin{eqnarray}\label{VV1}
&&p\ne 2\;\;\rightarrow\;\;\left (\widetilde g_p,B_p(b_p)\widetilde f_p\right )=[g_p\left(2^{-1}b_p\right)]^*\;f_p(-2^{-1}b_p);\;\;\;\;2^{-1}{b}_p\in {\mathbb Z}_p\nonumber\\
&&p=2\;\;\rightarrow\;\;\left (\widetilde g_2,B_2(b_2)\widetilde f_2\right )=[\widehat g_2\left(2^{-1}b_2\right)]^*\;\widehat f_2(-2^{-1}b_2)
\end{eqnarray}
The transformation of Eq.(\ref{126}), is used here and if ${\rm ord}(b_2) \ge 1$ then $\widehat f_2(-2^{-1}b_2)=f_2(-2^{-1}b_2)$
and also $\widehat g_2\left(2^{-1}b_2\right)=g_2\left(2^{-1}b_2\right)$.
\end{itemize}
\end{itemize}
\end{proposition}
\begin{proof}
\mbox{}
\begin{itemize}
\item[(1)]
We act with $D_p({\mathfrak a}_p, b_p,0)\;\theta \;[D_p({\mathfrak a}_p, b_p,0)]^{\dagger}$ 
on a function
$F_p({\mathfrak p}_p)\in {\mathfrak B}[{\mathbb Z}_p, ({\mathbb Q}_p/{\mathbb Z}_p)]$ and we get
\begin{eqnarray}
|2|_p\left [D_p({\mathfrak a}_p, b_p,0)\;\theta \;[D_p({\mathfrak a}_p, b_p,0)]^{\dagger}F_p\right ]({\mathfrak p}_p)=
&&|2|_p\int _{{\mathbb Q}_p/{\mathbb Z}_p}
d{\mathfrak p}'_p\;\chi _p(2{\mathfrak a}_pb_p-{\mathfrak p}_pb_p+{\mathfrak p}_p'b_p)\;\nonumber\\
&&\times\theta ({\mathfrak p}_p-2{\mathfrak a}_p,{\mathfrak p}_p')\;
F_p({\mathfrak p}_p'+2{\mathfrak a}_p).
\end{eqnarray}
Consequently, for any function $G_p({\mathfrak p}_p)\in {\mathfrak B}[{\mathbb Z}_p, ({\mathbb Q}_p/{\mathbb Z}_p)]$ we get
\begin{eqnarray}\label{nz}
&&|2|_p\int _{{\mathbb Q}_p/|2|_p{\mathbb Z}_p}d{\mathfrak a}_p\;\int _{{\mathbb Z}_p}db_p\;
\left (G_p,D_p({\mathfrak a}_p, b_p,0)\;\theta \;[D_p({\mathfrak a}_p, b_p,0)]^{\dagger}F_p\right )\nonumber\\&&=|2|_p
\int _{{\mathbb Q}_p/|2|_p{\mathbb Z}_p}d{\mathfrak a}_p\;\int _{{\mathbb Z}_p}db_p\;
\int _{{\mathbb Q}_p/{\mathbb Z}_p}d{\mathfrak p}_p
\int _{{\mathbb Q}_p/{\mathbb Z}_p}d{\mathfrak p}_p'\;[G_p({\mathfrak p}_p)]^*\;
\chi _p(2{\mathfrak a}_pb_p-{\mathfrak p}_pb_p+{\mathfrak p}'_pb_p)\nonumber\\&&\times \theta ({\mathfrak p}_p-2{\mathfrak a}_p,{\mathfrak p}_p')\;
F_p({\mathfrak p}_p'+2{\mathfrak a}_p)\nonumber\\&&=|2|_p
\int _{{\mathbb Q}_p/|2|_p{\mathbb Z}_p}d{\mathfrak a}_p\;
\int _{{\mathbb Q}_p/{\mathbb Z}_p}d{\mathfrak p}_p
\int _{{\mathbb Q}_p/{\mathbb Z}_p}d{\mathfrak p}_p'\;[G_p({\mathfrak p}_p)]^*\;\Delta _p(2{\mathfrak a}_p-{\mathfrak p}_p+{\mathfrak p}_p')
\nonumber\\&&\times \theta ({\mathfrak p}_p-2{\mathfrak a}_p,{\mathfrak p}_p')\;
F_p({\mathfrak p}_p'+2{\mathfrak a}_p)
\end{eqnarray}
We now change the variable ${\mathfrak a}_p$ into $2{\mathfrak a}_p$ taking into account Eq.(\ref{567}). 
Then the right hand side of Eq.(\ref{nz}) gives
\begin{eqnarray}\label{ve}
&&\int _{{\mathbb Q}_p/{\mathbb Z}_p}d(2{\mathfrak a}_p)\;
\int _{{\mathbb Q}_p/{\mathbb Z}_p}d{\mathfrak p}_p
\int _{{\mathbb Q}_p/{\mathbb Z}_p}d{\mathfrak p}_p'\;[G_p({\mathfrak p}_p)]^*\;\Delta _p(2{\mathfrak a}_p-{\mathfrak p}_p+{\mathfrak p}_p')
\theta ({\mathfrak p}_p-2{\mathfrak a}_p,{\mathfrak p}_p')\;
F_p({\mathfrak p}_p'+2{\mathfrak a}_p)\nonumber\\&&=
\int _{{\mathbb Q}_p/{\mathbb Z}_p}d{\mathfrak p}_p\;[G_p({\mathfrak p}_p)]^*\;F_p({\mathfrak p}_p)
\int _{{\mathbb Q}_p/{\mathbb Z}_p}d{\mathfrak p}_p'\;
\theta ({\mathfrak p}_p',{\mathfrak p}_p')=(G_p,F_p){\rm tr}(\theta).
\end{eqnarray}
This completes the proof.

\item[(2)]

We act with $D_p(-{\mathfrak a}_p, -b_p,0)$ on the kernel $\widetilde \theta ({\mathfrak p}_p,{\mathfrak p}_p')$ of the operator $\theta$ and we get
\begin{eqnarray}
[D_p(-{\mathfrak a}_p, -b_p,0)\widetilde \theta]({\mathfrak p}_p,{\mathfrak p}_p')=
\chi _p\left ({\mathfrak a}_pb_p+{\mathfrak p}_pb_p\right )\widetilde \theta ({\mathfrak p}_p+2{\mathfrak a}_p,{\mathfrak p}_p')
\end{eqnarray}
Therefore its trace is 
\begin{eqnarray}\label{950C}
{\rm tr} [D_p(-{\mathfrak a}_p, -b_p,0)\theta]&=&
\int_{{\mathbb Q}_p/{\mathbb Z}_p}d{\mathfrak p}_p\;\chi _p\left ({\mathfrak a}_pb_p+{\mathfrak p}_pb_p\right )
\widetilde \theta ({\mathfrak p}_p+2{\mathfrak a}_p,{\mathfrak p}_p)
\end{eqnarray}
We now act with the operator on the right hand side of Eq.(\ref{XAA}) on a function 
$F_p({\mathfrak p}_p)\in {\mathfrak B}[{\mathbb Z}_p, ({\mathbb Q}_p/{\mathbb Z}_p)]$
and we get
\begin{eqnarray}\label{950D}
&&|2|_p\int _{{\mathbb Q}_p/|2|_p{\mathbb Z}_p}d{\mathfrak a}_p\;\int _{{\mathbb Z}_p}db_p\;
\int_{{\mathbb Q}_p/{\mathbb Z}_p}d{\mathfrak p}_p\;\int_{{\mathbb Q}_p/{\mathbb Z}_p}d{\mathfrak p}_p'\;
\chi _p\left ({\mathfrak a}_pb_p+{\mathfrak p}_pb_p\right )
\widetilde \theta ({\mathfrak p}_p+2{\mathfrak a}_p,{\mathfrak p}_p)\nonumber\\&&\times
\chi _p\left ({\mathfrak a}_pb_p-{\mathfrak p}_p'b_p\right )
F_p({\mathfrak p}_p'-2{\mathfrak a}_p)\nonumber\\&&=|2|_p
\int _{{\mathbb Q}_p/|2|_p{\mathbb Z}_p}d{\mathfrak a}_p\;
\int_{{\mathbb Q}_p/{\mathbb Z}_p}d{\mathfrak p}_p\;\int_{{\mathbb Q}_p/{\mathbb Z}_p}d{\mathfrak p}_p'\;
\widetilde \theta ({\mathfrak p}_p+2{\mathfrak a}_p,{\mathfrak p}_p)\;\Delta_p(2{\mathfrak a}_p+{\mathfrak p}_p-{\mathfrak p}_p')
F_p({\mathfrak p}_p'-2{\mathfrak a}_p)\nonumber\\&&=
\int_{{\mathbb Q}_p/{\mathbb Z}_p}d(2{\mathfrak a}_p)\;\int_{{\mathbb Q}_p/{\mathbb Z}_p}d{\mathfrak p}_p\;
\widetilde \theta ({\mathfrak p}_p+2{\mathfrak a}_p,{\mathfrak p}_p)F_p({\mathfrak p}_p)\nonumber\\&&=
\int_{{\mathbb Q}_p/{\mathbb Z}_p}d{\mathfrak p}_p\;\int_{{\mathbb Q}_p/{\mathbb Z}_p}d{\mathfrak p}_p''\;
\widetilde \theta ({\mathfrak p}_p'',{\mathfrak p}_p)F_p({\mathfrak p}_p)
\end{eqnarray}
Therefore the right hand side of Eq.(\ref{XAA}) is equal to the operator $\theta$.
\item[(3)]
Eq.(\ref{951cs}) is proved as follows:
\begin{eqnarray}
\int _{{\mathbb Z}_p}db_p\;(\widetilde g_p,D_p({\mathfrak a}_p, b_p,0)\widetilde f_p)
&=&\int _{{\mathbb Q}_p/{\mathbb Z}_p} d{\mathfrak p}_p
\int _{{\mathbb Z}_p}db_p\;[\widetilde g_p({\mathfrak p}_p)]^*
\chi _p\left ({\mathfrak a}_pb_p-{\mathfrak p}_pb_p\right )\widetilde f_p({\mathfrak p}_p -2{\mathfrak a}_p)\nonumber\\
&=&\int _{{\mathbb Q}_p/{\mathbb Z}_p}d{\mathfrak p}_p\;
[\widetilde g_p({\mathfrak p}_p)]^*
\Delta_p({\mathfrak a}_p-{\mathfrak p}_p)\widetilde f_p({\mathfrak p}_p-2{\mathfrak a}_p)\nonumber\\&=&
[\widetilde g({\mathfrak a}_p)]^*\;\widetilde f_p(-{\mathfrak a}_p)
\end{eqnarray}
Eq.(\ref{VV1}) is proved as follows:
\begin{eqnarray}\label{5g}
|2|_p\int _{{\mathbb Q}_p/|2|_p{\mathbb Z}_p}d{\mathfrak a}_p\;(\widetilde g_p,D_p({\mathfrak a}_p, b_p,0)\widetilde f_p)
&=&\int _ {{\mathbb Q}_p/{\mathbb Z}_p}d{\mathfrak p}_p
\int _{{\mathbb Q}_p/{\mathbb Z}_p}d(2{\mathfrak a}_p)\;[\widetilde g({\mathfrak p}_p)]^*\nonumber\\&\times &
\chi _p\left ({\mathfrak a}_pb_p-{\mathfrak p}_pb_p\right )\widetilde f_p({\mathfrak p}_p-2{\mathfrak a}_p)
\end{eqnarray}
We change the variable $2{\mathfrak a}_p$ to
${\mathfrak p}_p'={\mathfrak p}_p-2{\mathfrak a}_p$ and we get
\begin{eqnarray}\label{6g}
\int _ {{\mathbb Q}_p/{\mathbb Z}_p}d{\mathfrak p}_p\;[\widetilde g_p({\mathfrak p}_p)]^*\chi _p\left (-2^{-1}b_p{\mathfrak p}_p\right )
\int _{{\mathbb Q}_p/{\mathbb Z}_p}d{\mathfrak p}_p'\;
\chi _p\left (-2^{-1}b_p{\mathfrak p}_p'\right )\widetilde f_p({\mathfrak p}_p').
\end{eqnarray}
For $p\ne 2$, these integrals are inverse Fourier transforms and we prove the first of Eqs.(\ref{VV1}) (in this case $2^{-1}b_p\in {\mathbb Z}_p$).
For $p=2$, we take into account Eq.(\ref{997}), to prove the second of Eqs.(\ref{VV1}).

\end{itemize}
\end{proof}
\begin{remark}
In the case $p=2$, in the integrals of Eqs(\ref{150AA}), (\ref{XAA}) 
we integrate the variable ${\mathfrak a}_2$ over ${\mathbb Q}_2/2^{-1}{\mathbb Z}_2$.
It is easily seen from lemma \ref{lemma}, that if we integrate ${\mathfrak a}_2$ over ${\mathbb Q}_2/{\mathbb Z}_2$, we get the quoted result times $2$.
Also in $B_2(b_2)$ in Eq.(\ref{t950}), if we integrate ${\mathfrak a}_2$ over ${\mathbb Q}_2/{\mathbb Z}_2$,
we will get the quoted result times $2$ or zero, in the cases that ${\rm ord}(b_2)\ge 1$ or ${\rm ord}(b_2)= 0$, correspondingly.

\end{remark}
\subsubsection{Coherent states}
The set of coherent states consists of the states 
\begin{eqnarray}\label{451v}      
f_{\rm coh}(x_p|{\mathfrak a}_p,b_p)=
[D_p({\mathfrak a}_p,b_p,0)g_p](x_p);\;\;\;\;\;
{\mathfrak a}_p\in {\mathbb Q}_p/{\mathbb Z}_p;\;\;\;\;\;\;b_p\in {\mathbb Z}_p
\end{eqnarray}
where $g_p(x_p)$ is a state in ${\mathfrak B}[{\mathbb Z}_p, ({\mathbb Q}_p/{\mathbb Z}_p)]$,
which we normalize so that
$(g_p,g_p)=1$. A special case of Eq.(\ref{150AA}), is the following resolution of the identity in terms of 
coherent states:
\begin{eqnarray}\label{50A}
|2|_p\int _{{\mathbb Q}_p/|2|_p{\mathbb Z}_p}d{\mathfrak a}_p\;\int _{{\mathbb Z}_p}db_p\;
f_{\rm coh}(x_p|{\mathfrak a}_p,b_p)\;
[f_{\rm coh}(x'_p|{\mathfrak a}_p,b_p)]^*=\delta _p(x_p-x_p').
\end{eqnarray}

\subsection{Parity operators}\label{B7}

The parity operator is given by
\begin{eqnarray}\label{nmk}
&&P_p({\mathfrak a}_p,b_p)=[D_p({\mathfrak a}_p,b_p,0)]^\dagger\;{\mathfrak F}_p^2\;D_p({\mathfrak a}_p,b_p,0)=
[D_p(2{\mathfrak a}_p,2b_p,0)]^\dagger\;{\mathfrak F}_p^2
={\mathfrak F}_p^2\;D_p(2{\mathfrak a}_p,2b_p,0)\nonumber\\
&&[P_p({\mathfrak a}_p,b_p)]^2 ={\bf 1};\;\;\;\;\;[P_p({\mathfrak a}_p,b_p)]^\dagger=P_p({\mathfrak a}_p,b_p).
\end{eqnarray}
It acts on the functions $f_p(x_p)\in {\mathfrak B}[{\mathbb Z}_p, ({\mathbb Q}_p/{\mathbb Z}_p)]$,
or $\widetilde f_p({\mathfrak p}_p)\in {\mathfrak B}[{\mathbb Z}_p, ({\mathbb Q}_p/{\mathbb Z}_p)]$, as follows:
\begin{eqnarray}\label{azx}
&&P_p({\mathfrak a}_p,b_p)f_p(x_p)=\chi _p(-4{\mathfrak a}_pb_p-4{\mathfrak a}_px_p)f_p(-x_p-2b_p)\nonumber\\
&&P_p({\mathfrak a}_p,b_p)\widetilde f_p({\mathfrak p}_p)=\chi _p(4{\mathfrak a}_pb_p+2{\mathfrak p}_pb_p)
\widetilde f_p(-{\mathfrak p}_p-4{\mathfrak a}_p).
\end{eqnarray}
It is easily seen that
\begin{eqnarray}\label{3400A}     
P_p\left ({\mathfrak a}_p+\frac{1}{4},b_p\right )=P_p({\mathfrak a}_p,b_p).
\end{eqnarray}
In view of this property, it is not surprising that integration of ${\mathfrak a}_p$ below is over ${\mathbb Q}_p/|4|_p{\mathbb Z}_p$, i.e.,
for $p\ne 2$ it is over ${\mathbb Q}_p/{\mathbb Z}_p$ and for $p=2$ it is over ${\mathbb Q}_2/2^{-2}{\mathbb Z}_2$.
\begin{proposition}\label{pa}
\mbox{}
\begin{itemize}
\item[(1)]
\begin{eqnarray}\label{agg2}
P_p({\mathfrak a},b)= |2|_p\int _{{\mathbb Q}_p/|2|_p{\mathbb Z}_p}d{\mathfrak a}'_p\int _{{\mathbb Z}_p}db'_p\;
D_p({\mathfrak a}'_p,b'_p,0)
\chi _p(2{\mathfrak a}'_pb_p-2{\mathfrak a}_pb'_p)
\end{eqnarray}
\item[(2)]
For any trace class operator $\theta$ acting on ${\mathfrak B}[{\mathbb Z}_p, ({\mathbb Q}_p/{\mathbb Z}_p)]$
\begin{eqnarray}\label{a50AA}
|4|_p\int _{{\mathbb Q}_p/|4|_p{\mathbb Z}_p}d{\mathfrak a}_p\;\int _{{\mathbb Z}_p}db_p\;
P_p({\mathfrak a}_p, b_p)\;\theta \;P_p({\mathfrak a}_p, b_p)={\bf 1}{\rm tr}\theta.
\end{eqnarray}
\item[(3)]
A trace class operator $\theta$ acting on ${\mathfrak B}[{\mathbb Z}_p, ({\mathbb Q}_p/{\mathbb Z}_p)]$, 
can be expanded in terms of displacement operators, as
\begin{eqnarray}\label{a950AA}
\theta=|8|_p\int _{{\mathbb Q}_p/|4|_p{\mathbb Z}_p}d{\mathfrak a}_p\;\int _{{\mathbb Z}_p}db_p\;
P_p({\mathfrak a}_p, b_p){\rm tr} [\theta P_p({\mathfrak a}_p, b_p)]
\end{eqnarray}
\item[(4)]
Let ${\cal A}_p({\mathfrak a}_p)$ and ${\cal B}_p(b_p)$ be the `marginal operators'
\begin{eqnarray}\label{aA950}
{\cal A}_p({\mathfrak a}_p)=|2|_p\int _{{\mathbb Z}_p}db_p\;P_p({\mathfrak a}_p, b_p);\;\;\;\;\;
{\cal B}_p(b_p)=|4|_p\int _{{\mathbb Q}_p/|4|_p{\mathbb Z}_p}d{\mathfrak a}_p\;P_p({\mathfrak a}_p, b_p).
\end{eqnarray}
For any $\widetilde g_p({\mathfrak p}_p),\widetilde f_p({\mathfrak p}_p)$ in 
${\mathfrak B}[{\mathbb Z}_p, ({\mathbb Q}_p/{\mathbb Z}_p)]$, 
\begin{itemize}
\item[(i)]
\begin{eqnarray}\label{CC1}
(\widetilde g_p,{\cal A}_p({\mathfrak a}_p)\widetilde f_p)=[{\widetilde g}_p(-2{\mathfrak a}_p)]^*\;
{\widetilde f}_p(-2{\mathfrak a}_p);\;\;\;\;{\mathfrak a}_p\in {\mathbb Q}_p/{\mathbb Z}_p.
\end{eqnarray}
\item[(ii)]
\begin{eqnarray}\label{CC2}
(\widetilde g_p,{\cal B}_p(b_p)\widetilde f_p)=[\widetilde g_p(-b_p)]^*\;\widetilde f_p(-b_p);\;\;\;\;{b}_p\in {\mathbb Z}_p\label{12S1}.
\end{eqnarray}
\end{itemize}
\end{itemize}
\end{proposition}
\begin{proof}
\mbox{}
\begin{itemize}
\item[(1)]
We act with the right hand side of Eq.(\ref{agg2}) on an arbitrary function 
$F_p({\mathfrak p}_p)\in {\mathfrak B}[{\mathbb Z}_p, ({\mathbb Q}_p/{\mathbb Z}_p)]$, and we get 
\begin{eqnarray}\label{agg21}
&&|2|_p\int _{{\mathbb Q}_p/|2|_p{\mathbb Z}_p}d{\mathfrak a}_p'\;
\chi _p(2{\mathfrak a}_p'b_p)F_p({\mathfrak p}_p-2{\mathfrak a}_p')
\int _{{\mathbb Z}_p}db_p'\;
\chi _p[b_p'({\mathfrak a}_p'-{\mathfrak p}_p-2{\mathfrak a}_p)]\nonumber\\&&=
\int _{{\mathbb Q}_p/{\mathbb Z}_p}d(2{\mathfrak a}_p')\;
\chi _p(2{\mathfrak a}_p'b_p)F_p({\mathfrak p}_p-2{\mathfrak a}_p')
\Delta _p({\mathfrak a}_p'-{\mathfrak p}_p-2{\mathfrak a}_p)\nonumber\\&&=
\chi _p(4b_p{\mathfrak a}_p+2b_p{\mathfrak p}_p)\;F_p(-{\mathfrak p}_p-4{\mathfrak a}_p)=P_p({\mathfrak a},b)F_p({\mathfrak p}_p).
\end{eqnarray}
\item[(2)]
We rewrite Eq.(\ref{150AA}) as
\begin{eqnarray}\label{a150AA}
|4|_p\int _{{\mathbb Q}_p/|4|_p{\mathbb Z}_p}d{\mathfrak a}_p\;\int _{{\mathbb Z}_p}db_p\;
D_p(2{\mathfrak a}_p, 2b_p,0)\;\theta \;[D_p(2{\mathfrak a}_p, 2b_p,0)]^{\dagger}={\bf 1}{\rm tr}\theta.
\end{eqnarray}
Then we multiply each side with ${\mathfrak F}_p^2$ on the left and with ${\mathfrak F}_p^2$ on the right, and we get
\begin{eqnarray}\label{a50AA}
|4|_p\int _{{\mathbb Q}_p/|4|_p{\mathbb Z}_p}d{\mathfrak a}_p\;\int _{{\mathbb Z}_p}db_p\;
P_p({\mathfrak a}_p, b_p)\;\theta \;P_p({\mathfrak a}_p, b_p)={\mathfrak F}_p^4{\rm tr}\theta={\bf 1}{\rm tr}\theta.
\end{eqnarray}  
\item[(3)]
We first point out that 
\begin{eqnarray}\label{950CC}
{\rm tr} [P_p({\mathfrak a}, b)\theta ]&=&
\int_{{\mathbb Q}_p/{\mathbb Z}_p}d{\mathfrak p}_p\int_{{\mathbb Q}_p/{\mathbb Z}_p}d{\mathfrak p}_p'\;
\chi _p\left (4{\mathfrak a}_pb_p+2{\mathfrak p}_pb_p\right )\widetilde  \theta (-{\mathfrak p}_p-4{\mathfrak a},{\mathfrak p}_p')
\Delta _p({\mathfrak p}_p+{\mathfrak p}_p'+4{\mathfrak a}_p)\nonumber\\&=&
\int_{{\mathbb Q}_p/{\mathbb Z}_p}d{\mathfrak p}_p\;
\widetilde \theta ({\mathfrak p}_p ,-{\mathfrak p}_p-4{\mathfrak a}_p)\chi _p(-4{\mathfrak a}_pb_p-2{\mathfrak p}_pb_p).
\end{eqnarray}
We now act with the operator on the right hand side of Eq.(\ref{a950AA}) on a function 
$F_p({\mathfrak p}_p)\in {\mathfrak B}[{\mathbb Z}_p, ({\mathbb Q}_p/{\mathbb Z}_p)]$, and we get
\begin{eqnarray}\label{950DD}
&&|8|_p\int _{{\mathbb Q}_p/|4|_p{\mathbb Z}_p}d{\mathfrak a}_p\;\int _{{\mathbb Z}_p}db_p\;
\int_{{\mathbb Q}_p/{\mathbb Z}_p}d{\mathfrak p}_p'\;
\widetilde \theta ({\mathfrak p}_p' ,-{\mathfrak p}_p'-4{\mathfrak a}_p)\chi _p(-4{\mathfrak a}_pb_p-2{\mathfrak p}_p'b_p)\nonumber\\&&\times
\chi (4{\mathfrak a}_pb_p+2{\mathfrak p}_pb_p)
F_p(-{\mathfrak p}_p-4{\mathfrak a}_p)\nonumber\\&&=|4|_p
\int _{{\mathbb Q}_p/|4|_p{\mathbb Z}_p}d{\mathfrak a}_p\;
\int_{{\mathbb Q}_p/{\mathbb Z}_p}d{\mathfrak p}_p'\;
\widetilde \theta ({\mathfrak p}_p' ,-{\mathfrak p}_p'-4{\mathfrak a}_p)
\Delta _p({\mathfrak p}_p-{\mathfrak p}_p')
F_p(-{\mathfrak p}_p-4{\mathfrak a}_p)\nonumber\\&&=|4|_p
\int _{{\mathbb Q}_p/|4|_p{\mathbb Z}_p}d{\mathfrak a}_p\;
\widetilde \theta ({\mathfrak p}_p ,-{\mathfrak p}_p-4{\mathfrak a}_p)
F_p(-{\mathfrak p}_p-4{\mathfrak a}_p)
=\int_{{\mathbb Q}_p/{\mathbb Z}_p}d{\mathfrak p}_p'
\widetilde \theta ({\mathfrak p}_p,{\mathfrak p}_p')F_p({\mathfrak p}_p')
\end{eqnarray}
Eq.(\ref{567}) has been used in the change of variables.
This completes the proof.
\item[(4)]
\begin{eqnarray}\label{951c}
&&|2|_p\int _{{\mathbb Z}_p}db_p\;(\widetilde g _p,P({\mathfrak a}_p, b_p)\widetilde f _p)\nonumber\\
&&=|2|_p\int _ {{\mathbb Q}_p/{\mathbb Z}_p}d{\mathfrak p}_p
\int _{{\mathbb Z}_p}db_p\;[\widetilde g({\mathfrak p}_p)]^*
\chi _p\left (4{\mathfrak a}_p b_p+2{\mathfrak p}_p b_p\right )\widetilde f_p(-{\mathfrak p}_p-4{\mathfrak a}_p)\nonumber\\
&&=\int _ {{\mathbb Q}_p/{\mathbb Z}_p}d{\mathfrak p}_p\; [\widetilde g_p({\mathfrak p}_p)]^*\;
\Delta _p(2{\mathfrak a}_p+{\mathfrak p}_p )\;\widetilde f _p(-{\mathfrak p}_p -4{\mathfrak a}_p)
=[\widetilde g _p(-2{\mathfrak a}_p]^*\;\widetilde f_p(-2{\mathfrak a}_p)
\end{eqnarray}
This proves Eq.(\ref{CC1}).

Also
\begin{eqnarray}\label{951v}
|4|_p\int _{{\mathbb Q}_p/|4|_p{\mathbb Z}_p}d{\mathfrak a}_p\;(\widetilde g_p,P({\mathfrak a}_p, b_p)\widetilde f_p)
=&&|4|_p\int _ {{\mathbb Q}_p/{\mathbb Z}_p}d{\mathfrak p}_p
\int _{{\mathbb Q}_p/|4|_p{\mathbb Z}_p}d{\mathfrak a}_p\;[g_p({\mathfrak p}_p)]^*\nonumber\\&&
\times \chi _p\left (4{\mathfrak a}_pb_p+2{\mathfrak p}_pb_p\right )F_p(-{\mathfrak p}_p-4{\mathfrak a}_p)
\end{eqnarray}
Taking into account Eq.(\ref{567}), we change variables from ${\mathfrak a}_p$ to
${\mathfrak p}_p'=-{\mathfrak p}_p-4{\mathfrak a}_p$, and we get 
\begin{eqnarray}\label{951a}
\int _ {{\mathbb Q}_p/{\mathbb Z}_p}d{\mathfrak p}_p\;[g({\mathfrak p}_p)]^*\chi _p\left (b{\mathfrak p}_p\right )
\int _{{\mathbb Q}_p/{\mathbb Z}_p}d{\mathfrak p}_p'\;
\chi _p\left (-b{\mathfrak p}_p'\right )F_p({\mathfrak p}_p')=[{\widetilde g}(-b)]^*\;{\widetilde f}(-b).
\end{eqnarray}
This proves Eq.(\ref{CC2}).

\end{itemize}
\end{proof}

\subsection{The profinite Heisenberg-Weyl group ${\bf HW}({\mathbb Z}_p,{\mathbb Z}_p,{\mathbb Z}_p)$}

We define the profinite Heisenberg-Weyl group ${\bf HW}({\mathbb Z}_p,{\mathbb Z}_p,{\mathbb Z}_p)$
as the inverse limit of the finite Heisenberg-Weyl groups ${\bf HW}[{\mathbb Z}(p^n),{\mathbb Z}(p^n),{\mathbb Z}(p^n)]$.
This is different from the ${\bf HW}[({\mathbb Q}_p/{\mathbb Z}_p), {\mathbb Z}_p, ({\mathbb Q}_p/{\mathbb Z}_p)]$
discussed earlier.

We consider the ${\bf HW}[{\mathbb Z}(p^\ell),{\mathbb Z}(p^\ell),{\mathbb Z}(p^\ell)]$ (discussed in section\ref{HW1}) as multiplicative topological groups 
with the discrete topology.
For $k\le \ell$ we define the homomorphisms
\begin{eqnarray}
\psi_{k\ell}:\;\;{\bf HW}[{\mathbb Z}(p^k),{\mathbb Z}(p^k),{\mathbb Z}(p^k)]\;\leftarrow\;\;
{\bf HW}[{\mathbb Z}(p^\ell),{\mathbb Z}(p^\ell),{\mathbb Z}(p^\ell)];\;\;\;\;\;\;k\le \ell,
\end{eqnarray} 
where
\begin{eqnarray}\label{500}
&&\psi_{k\ell}[D^{(p^\ell)}(\alpha _{p^\ell}, \beta _{p^\ell}, \gamma _{p^\ell})]
=D^{(p^k)}(\alpha _{p^k}, \beta _{p^k}, \gamma _{p^k})\nonumber\\
&&\alpha _{p^k}=\varphi_{k\ell}(\alpha _{p^\ell});\;\;\;\;\;\; 
\beta _{p^k}=\varphi_{k\ell}(\beta _{p^\ell});\;\;\;\;\;\; 
\gamma _{p^k}=\varphi_{k\ell}(\gamma _{p^\ell}).
\end{eqnarray}
The map $\varphi_{k\ell}$ has been defined in Eq.(\ref{hom1}).
The $\psi_{k\ell}$ are compatible in the sense that for $k\le \ell\le r$
we get
$\psi_{\ell r}\circ \psi_{k\ell}=\psi_{k\ell}$.
Also $\psi_{r r}={\bf 1}$ and therefore the
$\{HW[{\mathbb Z}(p^\ell),{\mathbb Z}(p^\ell),{\mathbb Z}(p^\ell)],\psi _{k\ell}\}$ is an inverse system,
whose inverse limit we denote as ${\bf HW}({\mathbb Z}_p,{\mathbb Z}_p,{\mathbb Z}_p)$.
The elements of this group are the sequences
\begin{eqnarray}\label{4m}
&&{\mathfrak D}_p(a_p,b_p,c_p)=(D^{(p)}(\alpha _{p}, \beta _{p}, \gamma _{p}), D^{(p^2)}(\alpha _{p^2}, \beta _{p^2}, \gamma _{p^2}),...)\nonumber\\
&&a_p=(\alpha _{p}, \alpha _{p^2}, ...);\;\;\;\;\;b_p=(\beta _{p}, \beta _{p^2}, ...);\;\;\;\;\;c_p=(\gamma _{p}, \gamma _{p^2}, ...)\nonumber\\
&&a_p,b_p,c_p \in {\mathbb Z}_p;\;\;\;\;\alpha _{p^k}, \beta _{p^k}, \gamma _{p^k} \in {\mathbb Z}(p^k).
\end{eqnarray} 
The $a_p,b_p,c_p$ are $p$-adic integers written here in the representation of Eq.(\ref{pt}).
Multiplication is performed componentwise, and we can prove that
\begin{eqnarray}
{\mathfrak D}_p(a_p,b_p,c_p){\mathfrak D}_p(a_p',b_p',c_p')={\mathfrak D}_p(a_p+a_p',b_p+b_p',c_p+c_p'+a_pb_p'-a_p'b_p).
\end{eqnarray} 
The subgroups $\{{\bf HW}(p^n{\mathbb Z}_p, p^n{\mathbb Z}_p, p^{n}{\mathbb Z}_p)\}$ are a fundamental system of neighbourhoods of the identity of the
profinite group ${\bf HW}({\mathbb Z}_p,{\mathbb Z}_p,{\mathbb Z}_p)$.
\begin{remark}\label{last1}
We have noted earlier in Eq.(\ref{bn1}), that the map between the $\alpha$-variables, is different from the map between the $\beta$-variables. 
This is because the $\alpha$-variables are eventually embedded into ${\mathbb Z}_p$ and the $\beta$-variables into its Pontryagin dual group
${\mathbb Q}_p/{\mathbb Z}_p$.  
In contrast, in Eq.(\ref{500}), the map between the $\alpha$-variables is the same as the map between the $\beta$-variables, 
because they are both associated with ${\mathbb Z}_p$.

\end{remark}
\begin{remark}
We consider the following subgroups of ${\bf HW}({\mathbb Z}_p,{\mathbb Z}_p,{\mathbb Z}_p)$:
\begin{eqnarray}
&&\{{\mathfrak D}_p(a_p,0,0)\;|\;a_p\in {\mathbb Z}_p\}\cong {\mathbb Z}_p\nonumber\\
&&\{{\mathfrak D}_p(0,b_p,0)\;|\;b_p\in {\mathbb Z}_p\}\cong {\mathbb Z}_p.
\end{eqnarray} 
Clearly these subgroups are {\bf not} Pontryagin dual to each other.
Therefore the ${\bf HW}({\mathbb Z}_p,{\mathbb Z}_p,{\mathbb Z}_p)$ cannot be linked with displacements of
dual variables and it is not relevant to quantum mechanics.
\end{remark}

\section{Subsystems of $\Sigma [{\mathbb Z}_p, ({\mathbb Q}_p/{\mathbb Z}_p)]$}

\subsection{Subsystems and supersystems}

Let $(G,\widetilde G)$ be a pair of groups which are Pontryagin dual to each other.
Also let $(E,\widetilde E)$ be another pair of groups which are Pontryagin dual to each other,
and in addition to that $\widetilde E$ is a subgroup of $\widetilde G$.
In this case there is a quotient relationship between $E$ and $G$ which has been discussed in section \ref{Pont}.

We consider the quantum systems $\Sigma (G,\widetilde G)$ and $\Sigma (E,\widetilde E)$.
We say that $\Sigma (E,\widetilde E)$ is a subsystem of $\Sigma (G,\widetilde G)$ (and we denote this as
$\Sigma (E,\widetilde E) \prec \Sigma (G,\widetilde G)$), or equivalently that
$\Sigma (G,\widetilde G)$ is a supersystem of $\Sigma (E,\widetilde E)$.

We explain in detail below how $\Sigma (E,\widetilde E)$ can be embedded into $\Sigma (G,\widetilde G)$.
We also consider sets of many quantum systems $\{\Sigma (G_n,{\widetilde G}_n)\}$ where the binary relation `subsystem' is a partial order.
We show how these directed partially ordered sets can become directed-complete partially ordered sets, by adding appropriate  
`top elements'.

The set $\{\Sigma (G_n,{\widetilde G}_n)\}$ can become a topological space, with the use of the partial order `subsystem'.
In this topology the `points' are the systems $\Sigma (G_n,{\widetilde G}_n)$.
Open sets contain some systems and their subsystems.
This clearly obeys the requirements that the union of any number of open sets is an open set, and that the
intersection of a finite number of open sets is an open set.
Closed sets are the complements of open sets and contain some systems and their supersystems. 

We will show that this topological space is a $T_0$ space and 
that the topology formalizes some very fundamental logical relations between quantum systems and their subsystems and supersystems.

\subsection{The subsystems $\Sigma[{\mathbb Z}(p^n),{\mathbb Z}(p^n)]$}

\begin{definition}
The subspace ${\mathfrak B}[{\mathbb Z}(p^n),{\mathbb Z}(p^{n})]$
of ${\mathfrak B}[{\mathbb Z}_p, ({\mathbb Q}_p/{\mathbb Z}_p)]$
is defined by one of the following ways, which are equivalent to each other:
\begin{itemize}
\item[(1)]
It contains functions $f_p(x_p)$ (where $x_p \in {\mathbb Z}_p$), 
which are locally constant with degree $n$.
They can be regarded as functions 
\begin{eqnarray}\label{k11}
f^{(p^n)}({\cal X} _{p^n})=f_p(x_p);\;\;\;\;
{\cal X} _{p^n}=\xi _n(x_p)\in {\mathbb Z}_p/p^{n}{\mathbb Z}_p\cong {\mathbb Z}(p^{n}), 
\end{eqnarray}
where $\xi _n(x_p)$ is given in Eq.(\ref{AA1}).
In this case the scalar product of Eq.(\ref{qa1}) reduces to
\begin{eqnarray}
(f^{(p^n)},g^{(p^n)})=\frac{1}{p^n}\sum _{{\cal X} _{p^n}\in {\mathbb Z}(p^{n})}
[f^{(p^n)}({\cal X} _{p^n})]^*g^{(p^n)}({\cal X} _{p^n}).
\end{eqnarray}  
\item[(2)]
It contains functions $F_p({\mathfrak p}_p)$ (where ${\mathfrak p}_p \in {\mathbb Q}_p/{\mathbb Z}_p$)
which have compact support with degree $n$.
They can be regarded as functions $F_p({\cal P} _{p^n})$ where ${\cal P} _{p^n}\in 
p^{-n}{\mathbb Z}_p/{\mathbb Z}_p\cong {\mathbb Z}(p^{n})$.
The relationship between ${\mathfrak p}_p$ and ${\cal P} _{p^n}$ is described by the equation
\begin{eqnarray}\label{k12}
F^{(p^n)}({\cal P} _{p^n})=F_p({\mathfrak p}_p)=F_p(p^{-n}{\cal P} _{p^n});\;\;\;\;
{\mathfrak p}_p=\widetilde \xi _n({\cal P} _{p^n})=p^{-n}{\cal P} _{p^n}
\end{eqnarray}
where $\widetilde \xi _n({\cal P} _{p^n})$ is given in Eq.(\ref{AA2}).
In this case the scalar product of Eq.(\ref{qa2}) reduces to
\begin{eqnarray}
(F^{(p^n)},G^{(p^n)})=\sum _{{\cal P} _{p^n}\in {\mathbb Z}(p^{n})}[F^{(p^n)}({\cal P} _{p^n})]^*
G^{(p^n)}({\cal P} _{p^n})
\end{eqnarray}  
\end{itemize}
This subspace describes a finite quantum system $\Sigma[{\mathbb Z}(p^n),{\mathbb Z}(p^n)]$
with positions and momenta in ${\mathbb Z}(p^n)$.
\end {definition}
\paragraph*{Fourier transform:}The restriction of the Fourier transform ${\mathfrak F}_p$ in the subspace ${\mathfrak B}[{\mathbb Z}(p^n),{\mathbb Z}(p^{n})]$
is a finite Fourier transform, which we denote as 
${\mathfrak F}^{(p^n)}$:
\begin{eqnarray}\label{qa19}
{\widetilde f}^{(p^n)}({\cal P} _{p^n})=
\left ({\mathfrak F}^{(p^n)}f^{(p^n)}\right )({\cal P} _{p^n})=\frac{1}{p^n} 
\sum _{{\cal X} _{p^n}\in {\mathbb Z}(p^n)} f^{(p^n)}({\cal X} _{p^n})
\omega _{p^n} (-{\cal X} _{p^n} {\cal P} _{p^n})
;\;\;\;\;{\cal X} _{p^n}, {\cal P} _{p^n} \in {\mathbb Z}(p^n).
\end{eqnarray}

\subsubsection{The finite Heisenberg-Weyl group ${\bf HW}[{\mathbb Z}(p^n),{\mathbb Z}(p^n),{\mathbb Z}(p^n)]$:}\label{HW1}
 
We denote as $D^{(p^n)}(\alpha _{p^n}, \beta _{p^n}, \gamma _{p^n})$ the
restriction of the displacement operators $D_p({\mathfrak a}_p,b_p,{\mathfrak c}_p)$
in the subspace ${\mathfrak B}[{\mathbb Z}(p^n),{\mathbb Z}(p^{n})]$. 
The relationship between $\beta _{p^n}$ and $b_p$ is analogous to that in Eq.(\ref{k11}), and the 
relationship between $\gamma _{p^n}$ and
${\mathfrak c}_p$ is analogous to that in Eq.(\ref{k12}). 
For ${\mathfrak a}_p$ we consider two cases.
If $p\ne 2$ then $2^{-1}$ is a $p$-adic integer, the ${\mathfrak a}_p \in p^{-n}{\mathbb Z}_p/{\mathbb Z}_p \cong {\mathbb Z}(p^n)$ and the 
relationship between $\alpha_{p^n}$ and
${\mathfrak a}_p$ is analogous to that in Eq.(\ref{k12}). 
For $p=2$, we note that $2{\mathfrak a_2}$ enters in ${\widetilde f}_2({\mathfrak p}_2-2{\mathfrak a_2})$ in
Eqs(\ref{450}). Consequently $2{\mathfrak a}_2 \in 2^{-n}{\mathbb Z}_2/{\mathbb Z}_2\cong {\mathbb Z}(2^n)$
and therefore ${\mathfrak a}_2 \in 2^{-n-1}{\mathbb Z}_2/2^{-1}{\mathbb Z}_2\cong {\mathbb Z}(2^n)$.
We summarize this as follows:
\begin{eqnarray}\label{fd} 
&&{\mathfrak c}_p \in p^{-n}{\mathbb Z}_p/{\mathbb Z}_p\cong {\mathbb Z}(p^n);\;\;\;\;\;
b_p\in {\mathbb Z}_p/p^n{\mathbb Z}_p\cong {\mathbb Z}(p^n)\nonumber\\
&&p\ne 2\;\;\rightarrow\;\;{\mathfrak a}_p\in p^{-n}{\mathbb Z}_p/{\mathbb Z}_p\cong {\mathbb Z}(p^n)\nonumber\\
&&p=2\;\;\rightarrow\;\;{\mathfrak a}_2 \in 2^{-n-1}{\mathbb Z}_2/2^{-1}{\mathbb Z}_2\cong {\mathbb Z}(2^n)
\end{eqnarray}
Then
\begin{eqnarray}\label{4xc}
&&p\ne 2\;\;\rightarrow\;\;D^{(p^n)}(\alpha _{p^n}, \beta _{p^n}, \gamma _{p^n})
=D_p\left (p^{-n}\alpha _{p^n}, \beta _{p^n}, p^{-n}\gamma _{p^n}\right )\nonumber\\
&&p=2\;\;\rightarrow\;\;D^{(2^n)}(\alpha _{2^n}, \beta _{2^n}, \gamma _{2^n})
=D_2\left (2^{-n-1}\alpha _{2^n}, \beta _{2^n}, 2^{-n}\gamma _{2^n}\right )
\end{eqnarray}
Consequently Eqs(\ref{450}) reduce to
\begin{eqnarray}\label{450me}      
&&[D^{(p^n)}(\alpha _{p^n}, \beta _{p^n}, \gamma _{p^n})f^{(p^n)}]({\cal X} _{p^n})=\omega _{p^n}\left (\gamma _{p^n}-\alpha _{p^n}\beta _{p^n}+
2\alpha _{p^n}{\cal X} _{p^n}\right )f^{(p^n)}({\cal X} _{p^n}-\beta _{p^n})\nonumber\\
&&[D^{(p^n)}(\alpha _{p^n}, \beta _{p^n}, \gamma _{p^n}){\widetilde f}^{(p^n)}]({\cal P} _{p^n})=
\omega _{p^n}\left (\gamma _{p^n}+\alpha _{p^n}\beta _{p^n}-\beta _{p^n}{\cal P} _{p^n}\right )
{\widetilde f}^{(p^n)}({\cal P} _{p^n}-2\alpha _{p^n})\nonumber\\
&&\alpha _{p^n}, \beta _{p^n}, \gamma _{p^n}, {\cal X} _{p^n}, {\cal P} _{p^n} \in {\mathbb Z}(p^n).
\end{eqnarray}
The operators $D_p^{(p^n)}(\alpha _{p^n}, \beta _{p^n}, \gamma _{p^n})$  form a group which we denote
as ${\bf HW}[{\mathbb Z}(p^n),{\mathbb Z}(p^n),{\mathbb Z}(p^n)]$.
We next use the notation
\begin{eqnarray}  
&&X=D^{(p^n)}(0,1,0)=D_p(0,1,0)\nonumber\\   
&&p\ne 2\;\;\rightarrow\;\;Z=D^{(p^n)}(2^{-1},0,0)=D_p\left (p^{-n}2^{-1},0,0\right )\nonumber\\
&&p= 2\;\;\rightarrow\;\;Z=D^{(2^n)}(1,0,0)=D_2\left (2^{-n-1},0,0\right )
\end{eqnarray}
We note that in $D^{(p^n)}(2^{-1},0,0)$ the $2^{-1}$ exists in ${\mathbb Z}(p^n)$ because $p\ne 2$.

Now Eqs(\ref{450me}) give
\begin{eqnarray}\label{finite}      
&&[Zf^{(p^n)}]({\cal X} _{p^n})=\omega _{p^n}\left ({\cal X} _{p^n}\right )f^{(p^n)}({\cal X} _{p^n});\;\;\;\;\;
[Xf^{(p^n)}]({\cal X} _{p^n})=f^{(p^n)}({\cal X} _{p^n}-1)
\nonumber\\
&&[Z{\widetilde f}^{(p^n)}]({\cal P} _{p^n})={\widetilde f}^{(p^n)}({\cal P} _{p^n}-1);\;\;\;\;\;
[X{\widetilde f}^{(p^n)}]({\cal X} _{p^n})=\omega _{p^n}\left (-{\cal P} _{p^n}\right ){\widetilde f}^{(p^n)}({\cal X} _{p^n})
\nonumber\\
&&Z^{\alpha _{p^n}}X^{\beta _{p^n}}=X^{\beta _{p^n}}Z^{\alpha _{p^n}}\omega _{p^n}\left (\alpha _{p^n}\beta _{p^n}\right );\;\;\;\;\;\;Z^{p^n}=X^{p^n}={\bf 1}.
\end{eqnarray}
These relations are valid for all values of $p$, and they are familiar in the formalism of finite quantum systems with variables in the ring ${\mathbb Z}(p^n)$
\cite{F1,F2,F3,F4,F5}.

\subsection{Embeddings and their compatibility}\label{EE1}

Using the embedding $E_{k\ell}$ of definition \ref{def1}, we define an embedding of the space
${\mathfrak B}[{\mathbb Z}(p^k),{\mathbb Z}(p^{k})]$
into the space ${\mathfrak B}[{\mathbb Z}(p^\ell),{\mathbb Z}(p^{\ell})]$.
For $k\le \ell$, we consider the linear maps
\begin{eqnarray}\label{abn1}
&&\overline E_{k\ell}:\;\;f^{(p^k)}({\cal X}_{p^k})\;\;\rightarrow \;\;f^{(p^\ell )}({\cal X} _{p^\ell});\;\;\;\;
{\cal X} _{p^\ell}={\cal X} _{p^k}\nonumber\\
&&\overline E_{k\ell}:\;\;F^{(p^k)}({\cal P} _{p^k})\;\;\rightarrow \;\;F^{(p^\ell )}({\cal P} _{p^\ell});\;\;\;\;\;
{\cal P} _{p^\ell}=p^{\ell -k}{\cal P} _{p^k}.
\end{eqnarray}
We have explained earlier that here we use the $\widetilde \varphi _{k\ell}$ and the inverse of a restriction of $\varphi _{k\ell}$
(Eqs.(\ref{hom1}),(\ref{hom2})).
Also, the ${\mathfrak B}[{\mathbb Z}(p^k),{\mathbb Z}(p^{k})]$
is embedded into ${\mathfrak B}[{\mathbb Z}_p, ({\mathbb Q}_p/{\mathbb Z}_p)]$,
with the linear maps:
\begin{eqnarray}\label{abn2}
&&\overline E_{k\infty}:\;\;f^{(p^k)}({\cal X} _{p^k})\;\;\rightarrow \;\;f_p(x_p);\;\;\;\;
x _{p}={\cal X} _{p^k}=\overline {\cal X} _0+...+\overline {\cal X} _{k-1}p^{k-1}\nonumber\\
&&\overline E_{k\infty}:\;\;F^{(p^k)}({\cal P} _{p^k})\;\;\rightarrow \;\;F_p({\mathfrak p}_p);\;\;\;\;\;
{\mathfrak p}_{p}=p^{-k}{\cal P} _{p^k}.
\end{eqnarray}
These embeddings are compatible in the sense that
\begin{eqnarray}
&&k\le \ell \le m\;\;\rightarrow\;\;\overline E_{\ell m}\circ \overline E_{k\ell}=\overline E_{km}\nonumber\\
&&\overline E_{k\ell}\circ {\mathfrak F}^{(p^k)}={\mathfrak F}^{(p^\ell)}\circ \overline E_{k\ell}.
\end{eqnarray}

\paragraph*{Embeddings of the Heisenberg-Weyl groups:}
The ${\bf HW}[{\mathbb Z}(p^k),{\mathbb Z}(p^k),{\mathbb Z}(p^k)]$
is embedded into ${\bf HW}[{\mathbb Z}(p^\ell),{\mathbb Z}(p^\ell ),{\mathbb Z}(p^\ell )]$
as follows:
\begin{eqnarray}
&&D^{(p^k)}(\alpha _{p^k}, \beta _{p^k}, \gamma _{p^k})\;\;\rightarrow\;\;
D^{(p^\ell)}(\alpha _{p^\ell}, \beta _{p^\ell}, \gamma _{p^\ell});\;\;\;\;k\le \ell\nonumber\\
&&\alpha _{p^\ell}=\alpha _{p^k};\;\;\;\;\;\;
\beta _{p^\ell}=p^{\ell -k}\beta _{p^k};\;\;\;\;\;\;
\gamma _{p^\ell}=\gamma _{p^k}.
\end{eqnarray}
Here also we use the $\widetilde \varphi _{k\ell}$ and the inverse of a restriction of $\varphi _{k\ell}$
(Eqs.(\ref{hom1}),(\ref{hom2})).

The
${\bf HW}[{\mathbb Z}(p^n),{\mathbb Z}(p^n),{\mathbb Z}(p^n)]$ is embedded
into ${\bf HW}[({\mathbb Q}_p/{\mathbb Z}_p),
{\mathbb Z}_p, ({\mathbb Q}_p/{\mathbb Z}_p)]$ by mapping
$D^{(p^n)}(\alpha _{p^n}, \beta _{p^n}, \gamma _{p^n})$ into
$D_p({\mathfrak a}_p,b_p,{\mathfrak c}_p)$ as in Eq.(\ref{4xc}).
Then we can prove the compatibility condition
\begin{eqnarray}
&&\overline E_{k\ell}\circ D^{(p^k)}(\alpha _{p^k}, \beta _{p^k}, \gamma _{p^k})=
D^{(p^\ell)}(\alpha _{p^\ell}, \beta _{p^\ell}, \gamma _{p^\ell})\circ \overline E_{k\ell}\nonumber\\
&&\overline E_{k\infty}\circ D^{(p^k)}(\alpha _{p^k}, \beta _{p^k}, \gamma _{p^k})=
D_p({\mathfrak a}_p,b_p,{\mathfrak c}_p)\circ \overline E_{k\infty}
\end{eqnarray}
Therefore the whole quantum formalism for the system $\Sigma [{\mathbb Z}(p^n),{\mathbb Z}(p^n)]$
is embedded into the quantum formalism for $\Sigma [{\mathbb Z}_p, ({\mathbb Q}_p/{\mathbb Z}_p)]$.
In this sense the quantum formalism for $\Sigma [{\mathbb Z}_p, ({\mathbb Q}_p/{\mathbb Z}_p)]$
generalizes the quantum formalism for $\Sigma [{\mathbb Z}(p^n),{\mathbb Z}(p^n)]$
by allowing arbitrarily large exponents $n$.
The Schwartz-Bruhat space (the fact that the functions have compact support and are locally constant)
ensures that the integrals converge.

\subsection{The chains $\Sigma ^{(p)}$ and $\Sigma _S^{(p)}$ as $T_0$ topological spaces}

The order between the $({\mathbb Z}(p^k),C(p^k))$ in Eq.(\ref{zc5}),
induces the order `subspace' between the $\Sigma[{\mathbb Z}(p^k),{\mathbb Z}(p^k)]$.
So the sets
\begin{eqnarray}
&&\Sigma ^{(p)}=\{\Sigma[{\mathbb Z}(p^k),{\mathbb Z}(p^{k})]|k\in {\mathbb Z}^+\}\nonumber\\
&&\Sigma _S^{(p)}=\Sigma ^{(p)}\cup \{\Sigma[{\mathbb Z}_p, ({\mathbb Q}_p/{\mathbb Z}_p)]\} 
\end{eqnarray}
with the order `subsystem' ($\prec$), are infinite chains.
The $\Sigma _S^{(p)}$ is a complete chain with 
$\Sigma[{\mathbb Z}_p, ({\mathbb Q}_p/{\mathbb Z}_p)]$ as
supremum:
\begin{eqnarray}
\Sigma[{\mathbb Z}(p),{\mathbb Z}(p)]\prec \Sigma[{\mathbb Z}(p^2),{\mathbb Z}(p^{2})]
\prec...\prec \Sigma[{\mathbb Z}_p, ({\mathbb Q}_p/{\mathbb Z}_p)].
\end{eqnarray}
The chain $\Sigma _S^{(p)}$ is order isomorphic to the chain ${\mathbb N} _S^{(p)}$
(with divisibility as an order).
Similarly, the chain $\Sigma ^{(p)}$ is order isomorphic to the chain ${\mathbb N} ^{(p)}$.

\paragraph*{The topological spaces $(\Sigma ^{(p)}, T_{\Sigma ^{(p)}})$ and $(\Sigma _S^{(p)}, T_{\Sigma _S^{(p)}})$:}
There is a bijective map between ${\mathbb N} ^{(p)}$ and $\Sigma ^{(p)}$.
Using this we can make $\Sigma ^{(p)}$ a topological space (see section \ref{topology}), and 
$(\Sigma ^{(p)}, T_{\Sigma ^{(p)}})\sim ({\mathbb N} ^{(p)}, T_{{\mathbb N} ^{(p)}})$.
The `points' in this topology are the systems $\Sigma[{\mathbb Z}(p^n),{\mathbb Z}(p^{n})]$,
and an open (resp., closed set) contains some systems and their subsystems (resp. supersystems).

Similarly, there is a bijective map between ${\mathbb N} _S^{(p)}$ and $\Sigma _S^{(p)}$:
\begin{eqnarray}
&&p^n\;\;\leftrightarrow \;\;\Sigma[{\mathbb Z}(p^n),{\mathbb Z}(p^{n})]\nonumber\\
&&p^\infty\;\;\leftrightarrow \;\;\Sigma [{\mathbb Z}_p, ({\mathbb Q}_p/{\mathbb Z}_p)]
\end{eqnarray}
Therefore, we can make $\Sigma _S^{(p)}$ a topological space and 
$(\Sigma _S^{(p)}, T_{\Sigma _S^{(p)}})\sim ({\mathbb N} _S^{(p)}, T_{{\mathbb N} _S^{(p)}})$.
According to proposition \ref{ttt}, $(\Sigma _S^{(p)}, T_{\Sigma _S^{(p)}})$ is a $T_0$-space which is compact,
and $(\Sigma ^{(p)}, T_{\Sigma ^{(p)}})$ is a $T_0$-space which is locally compact.
The physical meaning of the $T_0$ topology in this context is discussed later.

In summary, from a partial order theory point of view, the chain $\Sigma _S^{(p)}$ is complete, while the chain
$\Sigma ^{(p)}$ is not complete.
From a topological point of view, $\Sigma _S^{(p)}$ is a compact $T_0$-space, while
$\Sigma ^{(p)}$ is a locally compact $T_0$-space (which is not compact).

\subsection{Physical importance of the non-Archimedean metric and of the profinite topology}\label{a45}

The topology endowed by the $p$-adic metric, is the same as the profinite topology, and it is needed for the quantum formalism of the system
$\Sigma[{\mathbb Z}_p,({\mathbb Q}_p/{\mathbb Z}_p)]$.
The $x_p+p^n{\mathbb Z}_p$ where $n$ is a large positive integer, is a small neighbourhood around $x_p$.
Based on this we introduced functions which are locally constant in a neighbourhood of $x_p$.
The fact that ${\mathbb Z}_p$ is totally disconnected is also important for having non-trivial locally constant functions.
For example, in real numbers a locally constant function is constant, i.e., it is trivial.

We also need the non-Archimedean metric to define functions with compact support, i.e., functions which are zero at large distances.
The Schwartz-Bruhat space contains functions which are locally constant with degree $n$ and have constant support with degree $k$,
for all positive integers $n,k$.
For fixed $n,k$ these functions are effectively defined on $p^{-k}{\mathbb Z}_p/p^n{\mathbb Z}_p\cong {\mathbb Z}(p^{k+n})$ and we get the formalism of 
the system $\Sigma[{\mathbb Z}(p^{n+k}),{\mathbb Z}(p^{n+k})]$.
The Schwartz-Bruhat space contains functions with all degrees $n,k$ and in this sense the $\Sigma[{\mathbb Z}(p^\ell),{\mathbb Z}(p^\ell)]$
(for all $\ell\in {\mathbb Z}^+$) are subsystems of $\Sigma[{\mathbb Z}_p,({\mathbb Q}_p/{\mathbb Z}_p)]$.
We have used earlier partial order theory to show explicitly how smaller systems are embedded into larger systems.

\section{The system $\Sigma[{\widehat {\mathbb Z}},({\mathbb Q}/{\mathbb Z})]$}\label{B}

In this section we study a quantum system in which the position variable $x=(x_2,x_3,x_5,...)$ 
takes values in ${\widehat {\mathbb Z}}$ (where $x_p\in {\mathbb Z}_p$)
and the momentum variable ${\mathfrak p}=({\mathfrak p}_2,{\mathfrak p}_3,{\mathfrak p}_5,...)$ 
takes values in the Pontryagin dual group 
${\mathbb Q}/{\mathbb Z}$ (where ${\mathfrak p}_p\in {\mathbb Q}_p/{\mathbb Z}_p$). 
Two important concepts which are used below to describe the relationship between the quantum formalism for 
$\Sigma[{{\mathbb Z}_p},({\mathbb Q}_p/{\mathbb Z}_p)]$ and the quantum formalism for 
$\Sigma[{\widehat {\mathbb Z}},({\mathbb Q}/{\mathbb Z})]$, are the restricted tensor product of spaces and the 
restricted direct product of locally compact groups \cite{N4,b1}.

The wavefunctions in the position representation are finite linear combinations of 
products $f(x)=\prod _{p\in \Pi} f_p(x_p)$.
Integrals over ${\widehat {\mathbb Z}}$ of these functions are given by the products
\begin{eqnarray}\label{tt}
\int _ {\widehat {\mathbb Z}}f(x)dx=\prod _{p\in \Pi} \int _{{\mathbb Z}_p}f_p(x_p)dx_p.
\end{eqnarray}

The wavefunctions in the momentum representation are finite linear combinations of 
products $F({\mathfrak p})=\prod _{p\in \Pi} F_p({\mathfrak p}_p)$.
Integrals over ${\mathbb Q}/{\mathbb Z}$ of these functions are given by the products
\begin{eqnarray}\label{mn}
\int _ {{\mathbb Q}/{\mathbb Z}}F({\mathfrak p})d{\mathfrak p}
=\prod _{p\in \Pi} \int _{{\mathbb Q}_p/{\mathbb Z}_p}F_p({\mathfrak p}_p)d{\mathfrak p}_p
\end{eqnarray}
The Schwartz-Bruhat space ${\mathfrak B}[{\widehat {\mathbb Z}},{\mathbb Q}/{\mathbb Z})]$ is defined below in 
such a way so that the above integrals converge.
We first discuss some technical details.

\paragraph*{Change of variables:}
A change in the variables $x'=\lambda x$ or 
${\mathfrak p}'=\lambda {\mathfrak p}$ 
where $\lambda \in {\mathbb Z}^+$, 
is performed as follows:
\begin{eqnarray}\label{ppp10}
&&d(\lambda x)=\prod _{p\in \Pi} d(\lambda x_p)=\prod _{p\in \Pi} |\lambda |_p\; dx_p
=\frac{1}{|\lambda |_\infty }dx;\;\;\;\;\;\;x\in \widehat {\mathbb Z};\nonumber\\
&&d({\lambda \mathfrak p})=\prod _{p\in \Pi} d({\lambda \mathfrak p}_p)=\prod _{p\in \Pi} |\lambda |_p\; d{\mathfrak p}_p
=\frac{1}{|\lambda |_\infty }d{\mathfrak p};\;\;\;\;\;{\mathfrak p}\in {\mathbb Q}/{\mathbb Z}.
\end{eqnarray}
Ostrowski's theorem (in Eq.(\ref{tgb})) has been used here.

The following formula is used later: 
\begin{eqnarray}\label{pppp10}
\frac{1}{2}\int _{{\mathbb Q}/2^{-1}{\mathbb Z}}d{\mathfrak a}F(2{\mathfrak a})=\int _{{\mathbb Q}/{\mathbb Z}}d(2{\mathfrak a})F(2{\mathfrak a})=
\int _{{\mathbb Q}/{\mathbb Z}}d{\mathfrak a}'F({\mathfrak a}')
\end{eqnarray}
Here  the ${\mathfrak a}$ takes values in ${\mathbb Q}/2^{-1}{\mathbb Z}$, and therefore the
$2{\mathfrak a}$ takes values in ${\mathbb Q}/{\mathbb Z}$.

\paragraph*{Delta functions:}
Delta functions in the present context are given by
$\delta (x)=\prod \delta _p (x_p)$.
We note that 
\begin{eqnarray}
\delta (\lambda x)=\prod \delta _p (\lambda x_p)=\prod \frac{\delta _p (x_p)}{|\lambda |_p}=
|\lambda |_\infty \delta (x)
\end{eqnarray}
We also define the  $\Delta ({\mathfrak p})=\prod \Delta _p({\mathfrak p}_p)$.
This is equal to $0$ if ${\mathfrak p}\ne 0$ and it is equal to $1$ if ${\mathfrak p}=0$ 
(the zero in ${\mathbb Q}/{\mathbb Z}$ is the coset with all the integers).
Then
\begin{eqnarray}
\int _{\widehat {\mathbb Z}} dx\;f(x)\delta(x-a)=f(a);\;\;\;\;
\int _{{\mathbb Q}/{\mathbb Z}}d{\mathfrak p}\;F({\mathfrak p})\Delta ({\mathfrak p}-{\mathfrak a})=F({\mathfrak a})
\end{eqnarray}
Also
\begin{eqnarray}
\int _{\widehat {\mathbb Z}} dx\;\chi (x{\mathfrak p})=\Delta ({\mathfrak p});\;\;\;\;
\int _{{\mathbb Q}/{\mathbb Z}}d{\mathfrak p}\;\chi (x{\mathfrak p})=\delta (x).
\end{eqnarray}

\subsection{The Schwartz-Bruhat space ${\mathfrak B}[{\widehat {\mathbb Z}},({\mathbb Q}/{\mathbb Z})]$
as the restricted tensor product of 
${\mathfrak B}[{\mathbb Z}_p, ({\mathbb Q}_p/{\mathbb Z}_p)]$}\label{L2}

\begin{definition}\label{d1}
The Schwartz-Bruhat space ${\mathfrak B}[{\widehat {\mathbb Z}},({\mathbb Q}/{\mathbb Z})]$\cite{b1,b2,b3} is 
defined by one of the following two ways which are related through a Fourier transform, and are equivalent to each other:
\begin{itemize}
\item[(1)]
It consists of finite linear combinations of complex functions
$f(x)=\prod _{p\in \Pi} f_p(x_p)$ (where $x\in {\widehat {\mathbb Z}}$ and $x_p\in {\mathbb Z}_p$)
such that
\begin{itemize}
\item[(1A)] $f_p(x_p)$ are locally constant complex functions (these functions are defined 
in ${\mathbb Z}_p$ and therefore they have constant support),
\item[(1B)] $f_p(x_p)=1$ for all but a finite number of $p\in \Pi$.
\end{itemize}
The scalar product is given by
\begin{eqnarray}
(f,g)=\int _ {\widehat {\mathbb Z}}[f(x)]^*g(x)dx.
\end{eqnarray} 
\item[(2)]
It consists of finite linear combinations of complex functions 
$F({\mathfrak p})=\prod _{p\in \Pi} F_p({\mathfrak p}_p)$
(where ${\mathfrak p}\in {\mathbb Q}/{\mathbb Z}$ and ${\mathfrak p}_p \in {\mathbb Q}_p/{\mathbb Z}_p$)
such that
\begin{itemize}
\item[(2A)] $F_p({\mathfrak p}_p)$ are complex functions with compact support (these functions 
are defined in ${\mathbb Q}_p/{\mathbb Z}_p$ and therefore they are locally constant), 
\item[(2B)] $F_p({\mathfrak p}_p)=\Delta _p({\mathfrak p}_p)$ for all but a finite number of $p\in \Pi$.
\end{itemize}
The scalar product is given by
\begin{eqnarray}
(F,G)=\int _ {{\mathbb Q}/{\mathbb Z}}[F({\mathfrak p})]^*G({\mathfrak p})d{\mathfrak p}.
\end{eqnarray} 
\end{itemize}
\end{definition}

\paragraph*{The restricted tensor product:}
The Schwartz-Bruhat space ${\mathfrak B}[{\widehat {\mathbb Z}},({\mathbb Q}/{\mathbb Z})]$, is isomorphic to
the tensor product of the Schwartz-Bruhat spaces 
${\mathfrak B}[{\mathbb Z}_p,({\mathbb Q}_p/{\mathbb Z}_p)]$, subject to the restrictions (1B),(2B).
This is called restricted tensor product and it is crucial for the convergence of the integrals.
It is indicated with the usual tensor product notation with a prime.
\begin{eqnarray}
{\mathfrak B}[{\widehat {\mathbb Z}},({\mathbb Q}/{\mathbb Z})]
=\bigotimes_{p\in \Pi}'{\mathfrak B}[{\mathbb Z}_p,({\mathbb Q}_p/{\mathbb Z}_p)]
\end{eqnarray}
Although $f_p(x_p)=1$ for all but a finite number of $p\in \Pi$, the function is constant in small neighbourhoods rather than in large regions.
This is because of the topology (discussed in section \ref{nei}) where the open sets of ${\widehat {\mathbb Z}}$
are $\prod U_p$ where $U_p$ is an open sets in ${\mathbb Z}_p$, and $U_p={\mathbb Z}_p$ for all but a finite number of $p$.
In this sense the restricted tensor product needs that topology. 

The Fourier transform is given by
\begin{eqnarray}\label{qa25}
[{\mathfrak F}f]({\mathfrak p})=\widetilde f({\mathfrak p})=\int _{{\widehat {\mathbb Z}}}
dx\;\chi (-x{\mathfrak p})f(x);\;\;\;\;\;\;\;{\mathfrak p} \in {\mathbb Q}/{\mathbb Z}
\end{eqnarray}
and the inverse Fourier transform by 
\begin{eqnarray}\label{qa17}
[{\mathfrak F}^{-1}\widetilde f](x)=\int _{{\mathbb Q}/{\mathbb Z}}d{\mathfrak p}
\chi (x{\mathfrak p})f({\mathfrak p});\;\;\;\;\;s\in {\widehat {\mathbb Z}}
\end{eqnarray} 
Parceval's theorem states that $(f,g)=({\widetilde f},{\widetilde g})$.
We can also prove that ${\mathfrak F}^4={\bf 1}$.
We note that
\begin{eqnarray}
{\mathfrak F}=\bigotimes _{p\in \Pi} {\mathfrak F}_p.
\end{eqnarray}

Time evolution is described with relations analogous to Eqs.(\ref{ev1}),(\ref{ev2}).

\subsection{The Heisenberg-Weyl group 
${\bf HW}[({\mathbb Q}/{\mathbb Z}), {\widehat {\mathbb Z}}, ({\mathbb Q}/{\mathbb Z})]$
as the restricted direct product of ${\bf HW}[({\mathbb Q}_p/{\mathbb Z}_p),
{\mathbb Z}_p, ({\mathbb Q}_p/{\mathbb Z}_p)]$}\label{B3}

The displacement operators $D({\mathfrak a},b,{\mathfrak c})$ 
act on the functions $f(x)$ and $F({\mathfrak p})$ in ${\mathfrak B}[{\widehat {\mathbb Z}},({\mathbb Q}/{\mathbb Z})]$,
as follows:
\begin{eqnarray}\label{4500}      
&&[D({\mathfrak a},b,{\mathfrak c})f](x)=\chi \left ({\mathfrak c}-{\mathfrak a}b+
2{\mathfrak a}x\right )f(x-b)\nonumber\\
&&[D({\mathfrak a},b,{\mathfrak c})F]({\mathfrak p})=
\chi \left ({\mathfrak c}+{\mathfrak a}b-{\mathfrak p}b\right )F({\mathfrak p}-2{\mathfrak a})\nonumber\\
&&{\mathfrak a},{\mathfrak c}, {\mathfrak p}\in {\mathbb Q}/{\mathbb Z};\;\;\;\;\;b,x\in {\widehat {\mathbb Z}}.
\end{eqnarray}
These relations are analogous to Eqs.(\ref{450}).
We recall that in the ${\mathfrak a}=({\mathfrak a}_2,...,{\mathfrak a}_p,...)$ all but a finite number of the ${\mathfrak a}_p$ are equal to zero,
and similarly for the ${\mathfrak c}$ and ${\mathfrak p}$.
This is important in proving that the $[D({\mathfrak a},b,{\mathfrak c})f](x)$ and $[D({\mathfrak a},b,{\mathfrak c})F]({\mathfrak p})$ 
do obey the restrictions (1B) and (2B) in definition \ref{d1}, and therefore they belong to the space 
${\mathfrak B}[{\widehat {\mathbb Z}},({\mathbb Q}/{\mathbb Z})]$.

We can check that $D({\mathfrak a}+1,b,{\mathfrak c})=D({\mathfrak a},b,{\mathfrak c})=D({\mathfrak a},b,{\mathfrak c}+1)$
and this is consistent with the fact that ${\mathfrak a} \in {\mathbb Q}/{\mathbb Z}$.

The $D({\mathfrak a}, b,{\mathfrak c})$ form a representation 
of the Heisenberg-Weyl group which we denote as
${\bf HW}[({\mathbb Q}/{\mathbb Z}), {\widehat {\mathbb Z}}, ({\mathbb Q}/{\mathbb Z})]$.
The multiplication rule is analogous to Eq.(\ref{137}).
Also
\begin{eqnarray}\label{33}     
[D({\mathfrak a},b,{\mathfrak c})]^{\dagger}=D(-{\mathfrak a},-b,-{\mathfrak c});\;\;\;\;\;\;
D({\mathfrak a},b,{\mathfrak c})[D({\mathfrak a},b,{\mathfrak c})]^{\dagger}={\bf 1}.
\end{eqnarray}
In analogy with Eq.(\ref{30}), we consider the following subgroups of ${\bf HW}[({\mathbb Q}/{\mathbb Z}), {\widehat {\mathbb Z}}, ({\mathbb Q}/{\mathbb Z})]$,
\begin{eqnarray}\label{300}     
&&{\bf HW}_1({\mathbb Q}/{\mathbb Z})=\{D({\mathfrak a},0,0)\;|\;{\mathfrak a}\in {\mathbb Q}/{\mathbb Z}\}
\cong {\mathbb Q}/{\mathbb Z}\nonumber\\
&&{\bf HW}_2({\widehat {\mathbb Z}})=\{D_(0,b,0)\;|\;b\in {\widehat {\mathbb Z}}\}\cong {\widehat {\mathbb Z}}\nonumber\\
&&{\bf HW}_3({\mathbb Q}/{\mathbb Z})=\{D(0,0,{\mathfrak c})\;|\;{\mathfrak c} \in {\mathbb Q}/{\mathbb Z}\}\cong {\mathbb Q}/{\mathbb Z}.
\end{eqnarray}
\begin{definition}
Let $\{G_i\}$ be a set of locally compact groups, and $H_i$ be a compact subgroup of $G_i$.
We denote as $g_i$ the elements of $G_i$, and as $h_i$ the elements of $H_i$.
The restricted direct product of $G_i$ with respect to $H_i$ is the 
\begin{eqnarray}
\prod _{i}^{'}G_i=\{(g_1,g_2,...)\;|\;g_i=h_i\;{\rm for\; all\; but\; a \;finite \;number\; of \;the\; indices}\;i\}
\end{eqnarray}
\end{definition}
\begin{proposition}
\mbox{}
\begin{itemize}
\item[(1)]
${\bf HW}[({\mathbb Q}/{\mathbb Z}), {\widehat {\mathbb Z}}, ({\mathbb Q}/{\mathbb Z})]$ is a locally compact group.
\item[(2)]
${\bf HW}[({\mathbb Q}/{\mathbb Z}), {\widehat {\mathbb Z}}, ({\mathbb Q}/{\mathbb Z})]$ is the restricted direct product of 
${\bf HW}[({\mathbb Q}_p/{\mathbb Z}_p), {\mathbb Z}_p, ({\mathbb Q}_p/{\mathbb Z}_p)]$ with respect to
${\bf HW}_2({\mathbb Z}_p)\cong {\mathbb Z}_p$:
\begin{eqnarray}
{\bf HW}[({\mathbb Q}/{\mathbb Z}), {\widehat {\mathbb Z}}, ({\mathbb Q}/{\mathbb Z})]=\prod _{p\in \Pi}^{'}
{\bf HW}[({\mathbb Q}_p/{\mathbb Z}_p), {\mathbb Z}_p, ({\mathbb Q}_p/{\mathbb Z}_p)]
\end{eqnarray}
\end{itemize}
\end{proposition}
\begin{proof}
\mbox{}
\begin{itemize}
\item[(1)]
The proof is analogous to the proof of proposition \ref{local}.
\item[(2)]
We use the notation ${\mathfrak a}=({\mathfrak a}_2,...,{\mathfrak a}_p,...)$ and a similar notation for the other variables.
Then
\begin{eqnarray}
\chi \left ({\mathfrak c}+{\mathfrak a}b-{\mathfrak p}b\right )=\prod_{p\in \Pi}\chi _p\left ({\mathfrak c}_p+{\mathfrak a}_pb_p-{\mathfrak p}_pb_p\right ). 
\end{eqnarray}
In Eq.(\ref{4500}), the function $f(x)$  is a finite linear combination of $\prod f_p(x_p)$, and the function $F({\mathfrak p})$
is a finite linear combination of $\prod F_p({\mathfrak p}_p)$.
Therefore 
\begin{eqnarray}
D({\mathfrak a},b,{\mathfrak c})=\prod D_p({\mathfrak a}_p,b_p,{\mathfrak c}_p)
\end{eqnarray}
Since ${\mathfrak a},{\mathfrak c}$ belong to ${\mathbb Q}/{\mathbb Z}$,
all but a finite number of ${\mathfrak a}_p$, ${\mathfrak c}_p$ are equal to zero.
Therefore all but a finite number of the factors in the right hand side are $D_p(0,b_p,0)$ which are elements of
compact subgroup ${\bf HW}_2({\mathbb Z}_p)\cong {\mathbb Z}_p$ of ${\bf HW}[({\mathbb Q}/{\mathbb Z}), {\widehat {\mathbb Z}}, ({\mathbb Q}/{\mathbb Z})]$.
This completes the proof.
\end{itemize}
\end{proof}
\begin{definition}
Let $f(x)=\prod _{p\in \Pi} f_p(x_p)$ (where $x\in {\widehat{\mathbb Z}}$) be a function in 
${\mathfrak B}[{\widehat{\mathbb Z}}, ({\mathbb Q}/{\mathbb Z})]$.
Its transform ${\widehat f}(2^{-1}y)$ (where $y\in {\widehat{\mathbb Z}}$)
is defined as 
\begin{eqnarray}\label{cft}
{\widehat f}(2^{-1}y)={\widehat f}_2(2^{-1}y_2)\prod _{p\in \Pi-\{2\}}f_p(x_p).
\end{eqnarray}
The ${\widehat f}_2(2^{-1}y_2)$ has been defined in Eq.(\ref{126}).
If $x\in {\widehat{\mathbb Z}}_{\rm even}$ (so that $2^{-1}x\in {\widehat{\mathbb Z}}$),
then ${\widehat f}(2^{-1}y)=f(2^{-1}y)$.
This definition also applies to finite linear combinations of factorizable functions.
\end{definition}
\begin{lemma}\label{lemma}
Let
\begin{eqnarray}
&&N({\mathfrak a},b)=D({\mathfrak a}, b,0)\;\theta \;[D({\mathfrak a}, b,0)]^{\dagger}\nonumber\\
&&M({\mathfrak a},b)=D({\mathfrak a}, b,0){\rm tr} [D(-{\mathfrak a}, -b,0)\theta]
\end{eqnarray}
where $\theta$ is a trace class operator.
Then
\begin{itemize}
\item[(1)]
\begin{eqnarray}\label{evenodd}
&&b\in {\widehat {\mathbb Z}}_{\rm even}\;\;\rightarrow\;\;D({\mathfrak a}+2^{-1},b,{\mathfrak c}_2)=D({\mathfrak a},b,{\mathfrak c})\nonumber\\
&&b\in {\widehat {\mathbb Z}}_{\rm odd}\;\;\rightarrow\;\;D({\mathfrak a}+2^{-1},b,{\mathfrak c})=-D({\mathfrak a},b,{\mathfrak c}).
\end{eqnarray}
\item[(2)]
\begin{eqnarray}
N({\mathfrak a}+2^{-1},b)=N({\mathfrak a}+2^{-1},b);\;\;\;\;\;M({\mathfrak a}+2^{-1},b)=M({\mathfrak a}+2^{-1},b)
\end{eqnarray}

\end{itemize}
\end{lemma}
The proof is straightforward.
This lemma explains why the integration of ${\mathfrak a}$ below, is over ${\mathbb Q}/2^{-1}{\mathbb Z}$.
\begin{proposition}\label{T3}
\mbox{}
\begin{itemize}
\item[(1)]
For any trace class operator $\theta$ acting on ${\mathfrak B}[{\widehat{\mathbb Z}}, ({\mathbb Q}/{\mathbb Z})]$
\begin{eqnarray}\label{95AB}
\frac{1}{2}\int _{{\mathbb Q}/2^{-1}{\mathbb Z}}d{\mathfrak a}\;\int _{\widehat {\mathbb Z}}db\;
D({\mathfrak a}, b,0)\;\theta \;[D({\mathfrak a}, b,0)]^{\dagger}={\bf 1}{\rm tr}\theta.
\end{eqnarray}
\item[(2)]
A trace class operator $\theta$ acting on ${\mathfrak B}[{\widehat{\mathbb Z}}, ({\mathbb Q}/{\mathbb Z})]$ 
can be expanded in terms of displacement operators, as
\begin{eqnarray}\label{950A}
\theta=\frac{1}{2}\int _{{\mathbb Q}/2^{-1}{\mathbb Z}}d{\mathfrak a}\;\int _{\widehat {\mathbb Z}}db\;
D({\mathfrak a}, b,0){\rm tr} [\theta D(-{\mathfrak a}, -b,0)];\;\;\;\;\;
{\mathfrak a}\in {\mathbb Q}/{\mathbb Z};\;\;\;\;\;\;b\in {\widehat {\mathbb Z}}
\end{eqnarray}
\item[(3)]
Let $A({\mathfrak a})$ and $B(b)$ be the `marginal operators'
\begin{eqnarray}\label{950}
A({\mathfrak a})=\int _{{\widehat {\mathbb Z}}}db\;D({\mathfrak a}, b,0);\;\;\;\;\;
B(b)=\frac{1}{2}\int _{{\mathbb Q}/(2^{-1}{\mathbb Z})}d{\mathfrak a}\;D({\mathfrak a}, b,0).
\end{eqnarray}
For any $\widetilde g({\mathfrak p}),\widetilde f({\mathfrak p})$ in 
${\mathfrak B}[{\widehat {\mathbb Z}}, ({\mathbb Q}/{\mathbb Z})]$, 
\begin{eqnarray}
(\widetilde g,A({\mathfrak a})\widetilde f)=[{\widetilde g}({\mathfrak a})]^*\;
{\widetilde f}(-{\mathfrak a});\;\;\;\;{\mathfrak a}\in {\mathbb Q}/{\mathbb Z},
\end{eqnarray}
and also
\begin{eqnarray}\label{d45}
(\widetilde g,B(b)\widetilde f)=[\widehat g\left(2^{-1}b\right)]^*\;\widehat f(-2^{-1}b);\;\;\;\;\;
b\in {\widehat {\mathbb Z}}.
\end{eqnarray}
The transformation of Eq.(\ref{cft}), is used here 
and if $b\in {\widehat{\mathbb Z}}_{\rm even}$ 
then ${\widehat f}(-2^{-1}b)=f(-2^{-1}b)$ and also $\widehat g\left(2^{-1}b\right)=g\left(2^{-1}b\right)$.
\end{itemize}
\end{proposition}
\begin{proof}
The proof is similar to the proof of proposition \ref{di}.
The constants in the corresponding formulas are different,
because the change of variables in proposition \ref{di} is performed with Eq.(\ref{cha1}),
and the change of variables here is performed with Eqs.(\ref{ppp10}), (\ref{pppp10}).
\end{proof}
\begin{remark}
The result in Eq.(\ref{d45}) is different in the cases that $b\in {\widehat{\mathbb Z}}_{\rm even}$ and $b\in {\widehat{\mathbb Z}}_{\rm odd}$.
Analogous results for finite quantum systems, are also different for the cases of even and odd dimensional systems\cite{LE,ZA}.
Similar behaviour is also seen with the displacement operators for quantum mechanics on a circle\cite{ZV}. 
\end{remark}
\subsubsection{coherent states}
Coherent states are the states 
\begin{eqnarray}\label{451v}      
f_{\rm coh}(x)|{\mathfrak a},b)=[D({\mathfrak a},b,{\mathfrak c})f](x);\;\;\;\;\;
{\mathfrak a}\in {\mathbb Q}/{\mathbb Z};\;\;\;\;\;\;x,b\in {\widehat {\mathbb Z}}
\end{eqnarray}
$f({\mathfrak x})$ is an arbitrary state in ${\mathfrak B}[{\widehat{\mathbb Z}}, ({\mathbb Q}/{\mathbb Z})]$, 
which we normalize so that $(f,f)=1$. From Eq.(\ref{95AB})
follows the `resolution of the identity' property:
\begin{eqnarray}
\frac{1}{2}\int _{{\mathbb Q}/{\mathbb Z}}d{\mathfrak a}\;\int _{\widehat {\mathbb Z}}db\;f_{\rm coh}(x|{\mathfrak a},b)\;
[f_{\rm coh}(x'|{\mathfrak a},b)]^*
=\delta ({x}-{x}').
\end{eqnarray}

\subsection{Parity operators}\label{B7}

The parity operator is given by 
\begin{eqnarray}\label{nmk1}
&&P({\mathfrak a},b)=[D({\mathfrak a},b,0)]^\dagger\;{\mathfrak F}^2\;D({\mathfrak a},b,0)=
[D(2{\mathfrak a},2b,0)]^\dagger\;{\mathfrak F}^2
={\mathfrak F}^2\;D(2{\mathfrak a},2b,0)\nonumber\\
&&[P({\mathfrak a},b)]^2 ={\bf 1};\;\;\;\;\;[P({\mathfrak a},b)]^\dagger=P({\mathfrak a},b).
\end{eqnarray}
It acts on the functions $f(x)$ or 
$\widetilde f({\mathfrak p})$
in ${\mathfrak B}[{\widehat {\mathbb Z}}, ({\mathbb Q}/{\mathbb Z})]$,
as follows:
\begin{eqnarray}\label{azx1}
&&P({\mathfrak a},b)f(x)=\chi (-4{\mathfrak a}b-4{\mathfrak a}x)f(-x-2b)\nonumber\\
&&P({\mathfrak a},b)\widetilde f({\mathfrak p})=\chi (4{\mathfrak a}b+2{\mathfrak p}b)
\widetilde f(-{\mathfrak p}-4{\mathfrak a}).
\end{eqnarray}
These relations are analogous to Eqs.(\ref{nmk}),(\ref{azx}).
We can prove that
\begin{eqnarray}\label{3400A}     
P\left ({\mathfrak a}+\frac{1}{4},b\right )=P({\mathfrak a},b)
\end{eqnarray}
Due to this property, the integration of ${\mathfrak a}$ below is over ${\mathbb Q}/4^{-1}{\mathbb Z}$.

\begin{proposition}\label{pa1}
\mbox{}
\begin{itemize}
\item[(1)]
\begin{eqnarray}\label{gg2}
P({\mathfrak a},b)= \frac{1}{2}\int _{{\mathbb Q}/2^{-1}{\mathbb Z}}d{\mathfrak a}'\int _{\widehat {\mathbb Z}}db'\;
D({\mathfrak a}',b',0)
\chi (2{\mathfrak a}'b-2{\mathfrak a}b')
\end{eqnarray}
\item[(2)]
For any trace class operator $\theta$ acting on ${\mathfrak B}[{\widehat {\mathbb Z}}, 
({\mathbb Q}/{\mathbb Z})]$
\begin{eqnarray}\label{50AA}
\frac{1}{4}\int _{{\mathbb Q}/4^{-1}{\mathbb Z}}d{\mathfrak a}\;\int _{\widehat {\mathbb Z}}db\;
P({\mathfrak a}, b)\;\theta \;P({\mathfrak a}, b)={\bf 1}{\rm tr}\theta.
\end{eqnarray}
\item[(3)]
A trace class operator $\theta$ acting on ${\mathfrak B}[{\mathbb Z}_p, ({\mathbb Q}_p/{\mathbb Z}_p)]$, 
can be expanded in terms of displacement operators, as
\begin{eqnarray}\label{950AA}
\theta=\frac{1}{8}\int _{{\mathbb Q}/4^{-1}{\mathbb Z}}d{\mathfrak a}\;\int _{\widehat{\mathbb Z}}db\;
P({\mathfrak a}, b){\rm tr} [\theta P({\mathfrak a}, b)]
\end{eqnarray}
\item[(4)]
Let ${\cal A}({\mathfrak a})$ and ${\cal B}(b)$ be the `marginal operators'
\begin{eqnarray}\label{A950}
{\cal A}({\mathfrak a})=\frac{1}{2}\int _{\widehat {\mathbb Z}}db\;P({\mathfrak a}, b);\;\;\;\;\;
{\cal B}(b)=\frac{1}{4}\int _{{\mathbb Q}/4^{-1}{\mathbb Z}}d{\mathfrak a}\;P({\mathfrak a}, b).
\end{eqnarray}
For any $\widetilde g({\mathfrak p}),\widetilde f({\mathfrak p})$ in 
${\mathfrak B}[{\widehat {\mathbb Z}}, ({\mathbb Q}/{\mathbb Z})]$, 
\begin{itemize}
\item[(i)]
\begin{eqnarray}\label{A951}
(\widetilde g,{\cal A}({\mathfrak a})\widetilde f)=2[{\widetilde g}(-2{\mathfrak a})]^*\;
{\widetilde f}(-2{\mathfrak a});\;\;\;\;{\mathfrak a}\in {\widehat {\mathbb Z}}.
\end{eqnarray}
\item[(ii)]
\begin{eqnarray}
(\widetilde g,{\cal B}(b)\widetilde f)=4[g(-b)]^*\;f(-b);\;\;\;\;{b}\in {\widehat{\mathbb Z}}\label{12S1}.
\end{eqnarray}
\end{itemize}
\end{itemize}
\end{proposition}
\begin{proof}
The proof is similar to the proof of proposition \ref{pa}.
The only difference is that the change of variables in proposition \ref{pa} is performed with Eq.(\ref{cha1}),
and the change of variables here is performed with Eq.(\ref{ppp10}).
Consequently, the constants in the corresponding formulas are different.
\end{proof}
\begin{remark}
The Weyl function ${\mathfrak W}_f({\mathfrak a},b)$ and the Wigner function $W_f({\mathfrak a},b)$ are defined as
\begin{eqnarray}
{\mathfrak W}_f({\mathfrak a},b)=(f,D({\mathfrak a},b,0)f);\;\;\;\;\;W_f({\mathfrak a},b)=(f,P({\mathfrak a},b)f).
\end{eqnarray}
They are intimately connected to
displacement and parity operators, and in
this sense propositions \ref{di}, \ref{T3}, are also properties of the Weyl functions
and propositions \ref{pa}, \ref{pa1}, are also properties of the Wigner functions.
\end{remark}

\subsection{The profinite Heisenberg-Weyl group 
${\bf HW}({\widehat {\mathbb Z}},{\widehat {\mathbb Z}},{\widehat {\mathbb Z}})$}

We define the profinite Heisenberg-Weyl group ${\bf HW}({\widehat {\mathbb Z}},{\widehat {\mathbb Z}},{\widehat {\mathbb Z}})$
as the inverse limit of the finite Heisenberg-Weyl groups ${\bf HW}[{\mathbb Z}(n),{\mathbb Z}(n),{\mathbb Z}(n)]$.
This is different from the ${\bf HW}[({\mathbb Q}/{\mathbb Z}), {\widehat {\mathbb Z}}, ({\mathbb Q}/{\mathbb Z})]$ discussed earlier.

We consider the ${\bf HW}[{\mathbb Z}(\ell),{\mathbb Z}(\ell),{\mathbb Z}(\ell)]$ (discussed in section\ref{HW2}) as multiplicative topological groups 
with the discrete topology.
For $k|\ell$ we define the homomorphisms
\begin{eqnarray}
\Psi_{k\ell}:\;\;{\bf HW}[{\mathbb Z}(k),{\mathbb Z}(k),{\mathbb Z}(k)]\;\leftarrow\;\;
{\bf HW}[{\mathbb Z}(\ell),{\mathbb Z}(\ell),{\mathbb Z}(\ell)];\;\;\;\;\;\;k| \ell,
\end{eqnarray} 
where
\begin{eqnarray}\label{600}
&&\Psi_{k\ell}[D^{(\ell)}(\alpha _{\ell}, \beta _{\ell}, \gamma _{\ell})]
=D^{(k)}(\alpha _{k}, \beta _{k}, \gamma _{k})\nonumber\\
&&\alpha _{k}=\Phi_{k\ell}(\alpha _{\ell});\;\;\;\;\;\; 
\beta _{k}=\Phi_{k\ell}(\beta _{\ell});\;\;\;\;\;\; 
\gamma _{k}=\Phi_{k\ell}(\gamma _{\ell}).
\end{eqnarray}
The map $\Phi_{k\ell}$ has been defined in Eq.(\ref{hom3}).
The $\Psi_{k\ell}$ are compatible and the
$\{HW[{\mathbb Z}(\ell),{\mathbb Z}(\ell),{\mathbb Z}(\ell)],\Psi _{k\ell}\}$ is an inverse system,
whose inverse limit we denote as ${\bf HW}({\widehat{\mathbb Z}}, {\widehat{\mathbb Z}}, {\widehat{\mathbb Z}})$.
The elements of this group are the sequences
\begin{eqnarray}
&&{\mathfrak D}(a,b,c)=(D^{(2)}(\alpha _2, \beta _2, \gamma _2),D^{(3)}(\alpha _{3}, \beta _{3}, \gamma _3),...)\nonumber\\
&&a=(\alpha _2, \alpha _3,...);\;\;\;\;\;b=(\beta _2, \beta _3,...);\;\;\;\;\;c=(\gamma _2, \gamma _3,...)\nonumber\\
&&a,b,c \in {\widehat {\mathbb Z}};\;\;\;\;\alpha _k, \beta _k, \gamma _k \in {\mathbb Z}(k)
\end{eqnarray}
The $a,b,c$ are elements ${\widehat {\mathbb Z}}$ written in the representation of  Eq.(\ref{Z}).

We now follow an argument analogous to the one that took us from the representation of Eq.(\ref{Z}), to the representation of Eq.(\ref{ZA}).
In Eq.(\ref{900}) below, we show explicitly how the $D^{(k)}(\alpha _{k}, \beta _{k}, \gamma _{k})$ factorizes
as $\prod D^{(p^{e_p})}({\widehat \alpha} _{p^{e_p}}, \beta _{p^{e_p}}, {\widehat \gamma} _{p^{e_p}})$
where $p\in \Pi(k)$ and $e_p\in E(k)$.
Therefore the above sequence is uniquely defined by the elements with $k$ which is equal to the power of a prime.
But a sequence that contains the elements with $k=p^n$, is the
${\mathfrak D}_p(a_p,b_p,c_p)$ defined in Eq.(\ref{4m}). Therefore
\begin{eqnarray}\label{product}
&&{\mathfrak D}(a,b,c)=({\mathfrak D}_2(a_2,b_2,c_2),..., {\mathfrak D}_p(a_p,b_p,c_p),...)\nonumber\\
&&a=(a_2,...,a_p,...);\;\;\;\;b=(b_2,...,b_p,...);\;\;\;\;c=(c_2,...,c_p,...);\;\;\;\;\;a_p,b_p,c_p\in {\mathbb Z}_p.
\end{eqnarray}
Multiplication is performed componentwise and we can show that
\begin{eqnarray}
{\mathfrak D}(a,b,c){\mathfrak D}(a',b',c')={\mathfrak D}(a+a',b+b',c+c'+ab'-a'b)
\end{eqnarray}
From Eq.(\ref{product}), it follows that 
\begin{eqnarray}
{\bf HW}({\widehat{\mathbb Z}}, {\widehat{\mathbb Z}}, {\widehat{\mathbb Z}})=\prod _{p\in \Pi}
{\bf HW}({\mathbb Z}_p, {\mathbb Z}_p, {\mathbb Z}_p).
\end{eqnarray}
This is direct product.
There is no need for a restricted direct product here, because the ${\bf HW}({\mathbb Z}_p, {\mathbb Z}_p, {\mathbb Z}_p)$
are profinite and therefore compact.
\paragraph*{Product topology:}
The ${\bf HW}({\widehat{\mathbb Z}}, {\widehat{\mathbb Z}}, {\widehat{\mathbb Z}})$ 
is a profinite group, and therefore it is Hausdorff, compact and totally disconnected topological group.
A fundamental system of neighbourhoods of its identity is the 
$\{{\bf HW}({n\widehat{\mathbb Z}}, n{\widehat{\mathbb Z}}, n{\widehat{\mathbb Z}})\}$.
We factorize them as
\begin{eqnarray}
{\bf HW}(n{\widehat{\mathbb Z}}, n{\widehat{\mathbb Z}}, n{\widehat{\mathbb Z}})=\prod _{p\in \Pi}
{\bf HW}(p^{e_p}{\mathbb Z}_p, p^{e_p}{\mathbb Z}_p, p^{e_p}{\mathbb Z}_p);\;\;\;\;n=\prod p^{e_p}.
\end{eqnarray}
All but a finite number of the exponents $e_p$ are zero, i.e., the neighbourhoods of 
${\bf HW}({\widehat{\mathbb Z}}, {\widehat{\mathbb Z}}, {\widehat{\mathbb Z}})$
are $\prod U_p$ where $U_p$ is a neighbourhood of ${\bf HW}({\mathbb Z}_p, {\mathbb Z}_p, {\mathbb Z}_p)$ 
and for all but a finite number of $p$ we have $U_p={\bf HW}({\mathbb Z}_p, {\mathbb Z}_p, {\mathbb Z}_p)$.
This shows that the topology of ${\bf HW}({\widehat{\mathbb Z}}, {\widehat{\mathbb Z}}, {\widehat{\mathbb Z}})$ is the
product (Tychonoff) topology.   
\begin{remark}
The two subgroups of ${\bf HW}({\mathbb Z}_p,{\mathbb Z}_p,{\mathbb Z}_p)$, that contain the elements $\{{\mathfrak D}(a,0,0)\}$
and $\{{\mathfrak D}(0,b,0)\}$ are both isomorphic to ${\widehat {\mathbb Z}}$ and therefore they
are {\bf not} Pontryagin dual to each other.
Consequently the ${\bf HW}({\widehat {\mathbb Z}},{\widehat {\mathbb Z}},{\widehat {\mathbb Z}})$ cannot be associated with displacements of
dual variables in quantum mechanics.
It is an example of a representation of the Heisenberg-Weyl group which is not relevant to quantum mechanics.
\end{remark}
\begin{remark}\label{last}
The maps in definition \ref{def2} between the $x$-variables and between the ${\mathfrak p}$-variables, are different.
This is because the $x$-variables are related to ${\widehat {\mathbb Z}}$ and the ${\mathfrak p}$-variables
are related to its Pontryagin dual group ${\mathbb Q}/{\mathbb Z}$.
In contrast, in Eq.(\ref{600}), the map between the $\alpha$-variables is the same as the map between the $\beta$-variables, 
because they are both associated with ${\widehat {\mathbb Z}}$.
\end{remark}

\section{Subsystems of
$\Sigma [{\widehat {\mathbb Z}},({\mathbb Q}/{\mathbb Z})]$}

\subsection{The subsystems $\Sigma [{\mathfrak Z}(n),\widetilde {\mathfrak Z}(n)]$ with $n \in {\mathbb N}_S$}

\begin{definition}\label{def24}
The subspace ${\mathfrak B}[{\mathfrak Z}(n),\widetilde {\mathfrak Z}(n)]$
of ${\mathfrak B}[\widehat {\mathbb Z}, ({\mathbb Q}/{\mathbb Z})]$
is defined as in definition \ref{d1} with $x\in {\mathfrak Z}(n)$ and 
${\mathfrak p} \in \widetilde {\mathfrak Z}(n)$.
The functions are finite linear combinations of $f^{(n)}(x)=\prod_{p\in \Pi}f_p(x_p)$
where in addition to the conditions in definition \ref{d1}, we have:
\begin{itemize}
\item[(1)]
if $p\in \Pi^{(\infty)}(n)$ the $f_p(x_p)$ is locally constant function.
\item[(2)]
if $p\in \Pi^{({\rm fin})}(n)$ the $f_p(x_p)$ is locally constant function with degree $e_p(n)$
\item[(3)] 
if $p\in \Pi^{(0)}(n)$ the $f_p(x_p)=1$.
\end{itemize}
Equivalently, the functions are linear combinations of $F^{(n)}({\mathfrak p})=\prod_{p\in \Pi}F_p({\mathfrak p}_p)$
where in addition to the conditions in definition \ref{d1}, we have:
\begin{itemize}
\item[(1)]
if $p\in \Pi^{(\infty)}(n)$ the $F_p({\mathfrak p}_p)$ is a function with compact support.
\item[(2)]
if $p\in \Pi^{({\rm fin})}(n)$ the $F_p({\mathfrak p}_p)$ is 
a function with compact support with degree $e_p(n)$
\item[(3)] 
if $p\in \Pi^{(0)}(n)$ the $F_p({\mathfrak p}_p)=\Delta _p({\mathfrak p}_p)$.
\end{itemize}
\end{definition}
The restriction of the Fourier transform ${\mathfrak F}$ in the subspace 
${\mathfrak B}[{\mathfrak Z}(n),\widetilde {\mathfrak Z}(n)]$ is
\begin{eqnarray}
{\mathfrak F}^{(n)}=\bigotimes _{p\in \Pi^{({\rm fin})}(n)}{\mathfrak F}^{(p^{e_p})}
\bigotimes _{p\in \Pi^{(\infty)}(n)}{\mathfrak F}_p
\end{eqnarray}
The restriction of the Heisenberg-Weyl group in the subspace 
${\mathfrak B}[{\mathfrak Z}(n),\widetilde {\mathfrak Z}(n)]$ is
${\bf HW}[\widetilde {\mathfrak Z}(n),{\mathfrak Z}(n),\widetilde {\mathfrak Z}(n)]$, where
\begin{eqnarray}
{\bf HW}[\widetilde{\mathfrak Z}(n),{\mathfrak Z}(n),\widetilde{\mathfrak Z}(n)]\cong \prod _{p\in \Pi^{({\rm fin})}(n)}
 {\bf HW}[{\mathbb Z}(p^{e_p}),{\mathbb Z}(p^{e_p}),{\mathbb Z}(p^{e_p})]
 \prod _{p\in \Pi^{(\infty)}(n)}^{'}{\bf HW}[{\mathbb Q}_p/{\mathbb Z}_p,{\mathbb Z}_p,{\mathbb Q}_p/{\mathbb Z}_p]
\end{eqnarray}
Its elements are
\begin{eqnarray}
D^{(n)}({\mathfrak a}, b, {\mathfrak c})=\bigotimes _{p\in \Pi^{({\rm fin})}(n)}D^{(p^{e_p})}(\alpha _{p^{e_p}}, \beta _{p^{e_p}},\gamma _{p^{e_p}})
\bigotimes _{p\in \Pi^{(\infty)}(n)}D_p({\mathfrak a}_p, b_p, {\mathfrak c}_p)
\end{eqnarray}
Below we study in detail the special case that $n$ is equal to $\ell \in {\mathbb N}$.

\subsection{The subsystems $\Sigma [{\mathbb Z}(\ell),{\mathbb Z}(\ell)]$ with $\ell \in {\mathbb N}$, and their factorization}\label{6vb}

If $n$ is equal to $\ell \in {\mathbb N}$ then the $\Pi^{(\infty)}(\ell)$ is the empty set and ${\mathfrak Z}(\ell)\cong 
\widetilde {\mathfrak Z}(\ell)\cong {\mathbb Z}(\ell)$.
In this case we are only interested in the components of $x =(x_2,x_3,...)\in {\widehat {\mathbb Z}}$
and in the components of ${\mathfrak p}=({\mathfrak p}_2, {\mathfrak p}_3,...)\in {\mathbb Q}/{\mathbb Z}$ which have index $p\in \Pi^{({\rm fin})}(\ell)$.
This is because for $p\in \Pi^{(0)}(\ell)$ the $f_p(x_p)=1$ and $F_p({\mathfrak p}_p)=\Delta _p({\mathfrak p}_p)$.
Then the functions $f(x)$ and $F({\mathfrak p})$ can be written as
\begin{eqnarray}\label{Z1}
&&f^{(\ell )}({\cal X}_\ell)=f(x);\;\;\;\;\;{\cal X}_\ell=\pi _\ell (x)=\prod _p\xi _{e_p}(x_p);\;\;\;\;p\in \Pi(\ell);\;\;\;\;e_p\in E(\ell)\nonumber\\
&&F^{(\ell )}({\cal P}_\ell)=F({\mathfrak p});\;\;\;\;\;{\mathfrak p}=\widetilde \pi _\ell({\cal P}_\ell)=p^{-\ell}{\cal P}_\ell
;\;\;\;\;\;{\cal X}_\ell, {\cal P}_\ell \in {\mathbb Z}(\ell)
\end{eqnarray}
where the $\pi _\ell (x)$ and $\widetilde \pi _\ell({\cal P}_\ell)$ have been given in eqs.(\ref{BB1}),(\ref{BB2}), correspondingly.

Using the definition \ref{d1} in the present context, we regard the functions $f(x)$ and $F({\mathfrak p})$ as finite linear combinations of the products
\begin{eqnarray}\label{Z2}
&&\prod f^{(p^{e_p})}({\cal X}_{p^{e_p}});\;\;\;\;\;
\prod F^{(p^{e_p})}({\widehat{\cal P}}_{p^{e_p}})\nonumber\\
&&{\cal X}_{p^{e_p}}, {\widehat{\cal P}}_{p^{e_p}}  \in {\mathbb Z}(p^{e_p});\;\;\;\;
p\in \Pi(\ell);\;\;\;\;\;e_p\in E(\ell)
\end{eqnarray}
The link between Eqs.(\ref{Z1}), (\ref{Z2}), is described with the relations:
\begin{eqnarray}\label{107}
&&f^{(\ell )}({\cal X}_\ell)=\prod f^{(p^{e_p})}({\cal X}_{p^{e_p}});\;\;\;\;\;
F^{(\ell )}({\cal P}_\ell)=\prod F^{(p^{e_p})}({\widehat {\cal P}}_{p^{e_p}})\nonumber\\
&&{\cal X}_{p^{e_p}}={\cal X}_\ell\;({\rm mod}\;p^{e_p});\;\;\;\;{\cal X}_\ell=\sum _p{\cal X}_{p^{e_p}}w_p;\;\;\;\;\;{\cal X}_\ell, {\cal P}_\ell \in {\mathbb Z}(\ell)\nonumber\\
&&{\widehat {\cal P}}_{p^{e_p}}={\cal P}_\ell t_p\;({\rm mod}\;p^{e_p});\;\;\;\;\frac{{\widehat {\cal P}}_\ell}{\ell}=\sum _p \frac{{\cal P}_{p^{e_p}}}{p^{e_p}}
\end{eqnarray}
The relationship between ${\cal X}_\ell$ and the ${\cal X}_{p^{e_p}}$ is analogous to Eq.(\ref{m1}), and 
the relationship between ${\cal P}_\ell$ and the ${\widehat {\cal P}}_{p^{e_p}}$ is analogous to Eq.(\ref{m2}).
The `hat' in the notation of ${\widehat {\cal P}}_{p^{e_p}}$ is consistent with the notation in Eq.(\ref{m2}), and emphasizes the fact that
the map for positions is different from the map for momenta. 
The variables $w_p,t_p$ are defined in Eq.(\ref{variables}).

With Eq.(\ref{107}), the $\ell$-dimensional space ${\mathfrak B}[{\mathbb Z}(\ell),{\mathbb Z}(\ell)]$ 
that describes the system $\Sigma[{\mathbb Z}(\ell),{\mathbb Z}(\ell)]$ with variables
in ${\mathbb Z}(\ell)$,
is expressed as the tensor product
\begin{eqnarray}\label{fa}
{\mathfrak B}[{\mathbb Z}(\ell),{\mathbb Z}(\ell)]=\bigotimes _p {\mathfrak B}[{\mathbb Z}(p^{e_p}),{\mathbb Z}(p^{e_p})];\;\;\;\;p\in \Pi(\ell);\;\;\;\;\;e_p\in E(\ell).
\end{eqnarray}
A direct proof of this, which 
is an extension of Good's formalism on fast Fourier transform and which is based on the Chinese remainder theorem (see remark \ref{rem12}), 
has been given in refs\cite{facto1,facto2}. 
In the context of the present paper we arrive at these ideas through
definition \ref{def24} for the Schwartz-Bruhat space.

The restriction of the Fourier transform ${\mathfrak F}$ in the subspace 
${\mathfrak B}[{\mathbb Z}(\ell),{\mathbb Z}(\ell)]$
is a finite Fourier transform, which is factorized in terms of finite Fourier transforms ${\mathfrak F}^{(p^{e_p})}$
in ${\mathfrak B}[{\mathbb Z}(p^{e_p}),{\mathbb Z}(p^{e_p})]$, as follows:
\begin{eqnarray}
&&{\mathfrak F}^{(\ell)}=\bigotimes _{p}{\mathfrak F}^{(p^{e_p})};\;\;\;\;p\in \Pi(\ell);\;\;\;\;\;e_p\in E(\ell)\nonumber\\
&&F({\cal P} _\ell)=
\left ({\mathfrak F}^{(\ell)}f\right )({\cal P} _{\ell})=\frac{1}{\ell} 
\sum _{{\cal X} _{\ell}\in {\mathbb Z}(\ell)} f({\cal X} _{\ell})
\omega _{\ell} (-{\cal X} _{\ell} {\cal P} _{\ell})
;\;\;\;\;{\cal X} _{\ell}, {\cal P} _{\ell} \in {\mathbb Z}(\ell).
\end{eqnarray}
This factorization is basically Good's factorization of the Fourier transform, given in Eq.(\ref{82}).

\subsubsection{The finite Heisenberg-Weyl group ${\bf HW}[{\mathbb Z}(\ell),{\mathbb Z}(\ell),{\mathbb Z}(\ell)]$:} \label{HW2}

We denote as $D^{(\ell)}(\alpha _{\ell}, \beta _{\ell}, \gamma _{\ell})$ the
restriction of the displacement operators $D({\mathfrak a},b,{\mathfrak c})$
in the $\ell$-dimensional subspace ${\mathfrak B}[{\mathbb Z}(\ell),{\mathbb Z}(\ell)]$. 
The relationship between the $\beta _{\ell}$ and $b$
is analogous to the relationship between ${\cal X}_\ell$ and $x$ in Eqs.(\ref{Z1}).
Also the relationship between $\gamma _{\ell}$ and ${\mathfrak c}$
is analogous to the relationship between ${\cal P}_\ell$ and ${\mathfrak p}$ in Eqs.(\ref{Z1}). 
For the relationship between $\alpha _{\ell}$ and ${\mathfrak a}$
we recall Eq.(\ref{fd}),(\ref{4xc}),  where there is difference between $p\ne 2$ and $p=2$.
The implication of this here, is that for odd and even $\ell$, we have
${\mathfrak a}=\ell^{-1}\alpha _{\ell}$ and ${\mathfrak a}=(2\ell)^{-1}\alpha _{\ell}$, correspondingly
(in the sense of Eq.(\ref{468})).
Then
\begin{eqnarray}\label{444}
&&{\rm odd}\;\ell \;\;\rightarrow\;\;D^{(\ell)}(\alpha _{\ell}, \beta _{\ell}, \gamma _{\ell})
=D\left (\ell^{-1}\alpha _{\ell}, \beta _{\ell}, \ell ^{-1}\gamma _{\ell}\right )\nonumber\\
&&{\rm even}\;\ell \;\;\rightarrow\;\;D^{(\ell)}(\alpha _{\ell}, \beta _{\ell}, \gamma _{\ell})
=D\left ((2\ell) ^{-1}\alpha _{\ell}, \beta _{\ell}, \ell^{-1}\gamma _{\ell}\right )
\end{eqnarray}
Now Eqs(\ref{4500}) reduce to
\begin{eqnarray}\label{708}      
&&[D^{(\ell)}(\alpha _{\ell}, \beta _{\ell}, \gamma _{\ell})f]({\cal X} _{\ell})=\omega _{\ell}\left (\gamma _{\ell}-\alpha _{\ell}\beta _{\ell}+
2\alpha _{\ell}{\cal X} _{\ell}\right )f({\cal X} _{\ell}-\beta _{\ell})\nonumber\\
&&[D^{(\ell)}(\alpha _{\ell}, \beta _{\ell}, \gamma _{\ell})\widetilde f]({\cal P} _{\ell})=
\omega _{\ell}\left (\gamma _{\ell}+\alpha _{\ell}\beta _{\ell}-\beta _{\ell}{\cal P} _{\ell}\right )
\widetilde f({\cal P} _{\ell}-2\alpha _{\ell})\nonumber\\
&&\alpha _{\ell}, \beta _{\ell}, \gamma _{\ell}, {\cal X} _{\ell}, {\cal P} _{\ell} \in {\mathbb Z}(\ell).
\end{eqnarray}
We call ${\bf HW}[{\mathbb Z}(\ell),{\mathbb Z}(\ell),{\mathbb Z}(\ell)]$ the group of the operators $D^{(\ell)}(\alpha _{\ell}, \beta _{\ell}, \gamma _{\ell})$.
We next use the notation
\begin{eqnarray}  
&&X=D^{(\ell)}(0,1,0)=D(0,1,0)\nonumber\\   
&&{\rm odd}\;\ell \;\;\rightarrow\;\;Z=D^{(\ell)}(2^{-1},0,0)=D\left ((2\ell) ^{-1},0,0\right )\nonumber\\
&&{\rm even}\;\ell \;\;\rightarrow\;\;Z=D^{(\ell)}(1,0,0)=D\left (\ell ^{-1},0,0\right )
\end{eqnarray}
For odd $\ell$ the $2^{-1}$ exists in ${\mathbb Z}(\ell)$.
Now Eqs(\ref{708}) give
\begin{eqnarray} 
&&[Zf^{(\ell )}]({\cal X} _{\ell })=\omega _{\ell }\left ({\cal X} _{\ell }\right )f^{(\ell )}({\cal X} _{\ell });\;\;\;\;\;
[Xf^{(\ell )}]({\cal X} _{\ell })=f^{(\ell )}({\cal X} _{\ell }-1)
\nonumber\\
&&[Z{\widetilde f}^{(\ell )}]({\cal P} _{\ell })={\widetilde f}^{(\ell )}({\cal P} _{\ell }-1);\;\;\;\;\;
[X{\widetilde f}^{(\ell )}]({\cal X} _{\ell })=\omega _{\ell }\left (-{\cal P} _{\ell }\right ){\widetilde f}^{(\ell )}({\cal X} _{\ell })
\nonumber\\
&&Z^{\alpha _{\ell}}X^{\beta _{\ell}}=X^{\beta _{\ell}}Z^{\alpha _{\ell}}\omega _{\ell}\left (\alpha _{\ell}\beta _{\ell}\right );\;\;\;\;\;\;
Z^{\ell }=X^{\ell }={\bf 1}.
\end{eqnarray}
These relations are valid for all values of $\ell$ and they generalize Eqs(\ref{finite}).
They are used in the formalism for finite quantum systems with variables in the ring ${\mathbb Z}(\ell)$\cite{F1,F2,F3,F4,F5}, and they show clearly that 
the formalism of this paper is a generalization of it.

With regard to the factorization in Eq.(\ref{fa}), the diplacement operators factorize as
\begin{eqnarray}\label{900}  
D^{(\ell)}(\alpha _{\ell}, \beta _{\ell}, \gamma _{\ell})=\prod D^{(p^{e_p})}({\widehat \alpha} _{p^{e_p}}, \beta _{p^{e_p}}, {\widehat \gamma} _{p^{e_p}})
\end{eqnarray}
The relationship between the $\beta _{\ell}$ and $\beta _{p^{e_p}}$
is analogous to the map in Eqs.(\ref{m1}).
Also the relationship between $\alpha _{\ell},\gamma _{\ell}$ and ${\widehat \alpha} _{p^{e_p}},{\widehat \gamma }_{p^{e_p}}$ correspondingly,
is analogous to the map in Eqs.(\ref{m2}). 

\subsection{Embeddings and their compatibility}\label{EE2}

Using the embedding ${\bf E}_{k\ell}$ of definition \ref{def2}, we define an embedding of the space
${\mathfrak B}[{\mathfrak Z}(k),\widetilde {\mathfrak Z}(k)]$
into the space ${\mathfrak B}[{\mathfrak Z}(\ell),\widetilde {\mathfrak Z}(\ell)]$.
For $k,\ell \in {\mathbb N}_S$ with $k|\ell$, we consider the linear maps
\begin{eqnarray}\label{abn1}
&&\overline {\bf E}_{k\ell}:\;\;f(x)\;\;\rightarrow \;\;f(x')\nonumber\\
&&\overline {\bf E}_{k\ell}:\;\;F(\mathfrak p)\;\;\rightarrow \;\;F(\mathfrak p ').
\end{eqnarray}
where the relationship between $x$ and $x'$ and also between
$\mathfrak p$ and $\mathfrak p '$ is given in definition \ref{def2}.
These embeddings are compatible in the sense that
\begin{eqnarray}
&&k|\ell|m\;\;\rightarrow\;\;\overline {\bf E}_{\ell m}\circ \overline {\bf E}_{k\ell}=\overline 
{\bf E}_{km}\nonumber\\
&&\overline {\bf E}_{k\ell}\circ {\mathfrak F}^{(k)}={\mathfrak F}^{(\ell)}\circ 
\overline {\bf E}_{k\ell}.
\end{eqnarray}

\paragraph*{Embeddings of the Heisenberg-Weyl groups:}
For $k|\ell$, we introduce an embedding of 
${\bf HW}[\widetilde {\mathfrak Z}(k),{\mathfrak Z}(k),\widetilde {\mathfrak Z}(k)]$
into ${\bf HW}[\widetilde {\mathfrak Z}(\ell),{\mathfrak Z}(\ell),\widetilde {\mathfrak Z}(\ell)]$, as follows:
\begin{eqnarray}
&&D^{(k)}(\alpha _{k}, \beta _{k}, \gamma _{k});\;\rightarrow\;\;  
D^{(\ell)}(\alpha _{\ell}, \beta _{\ell}, \gamma _{\ell})
\end{eqnarray}
The relationship between $\alpha _{k}, \beta _{k}, \gamma _{k}$ and  
$\alpha _{\ell}, \beta _{\ell}, \gamma _{\ell}$ is analogous to the one described in definition \ref{def2}.
We can prove the compatibility condition
\begin{eqnarray}
&&k|\ell\;\;\rightarrow\;\;\overline {\bf E}_{k \ell }\circ D^{(k)}(\alpha _{k}, \beta _{k}, \gamma _{k})=D^{(\ell)}(\alpha _{\ell}, \beta _{\ell}, \gamma _{\ell})
\circ \overline {\bf E}_{k\ell}
\end{eqnarray}

The quantum formalism for $\Sigma [{\widehat {\mathbb Z}},({\mathbb Q}/{\mathbb Z})]$
generalizes the formalism for the finite quantum systems $\Sigma [{\mathbb Z}(k),{\mathbb Z}(k)]$, 
in two different ways.
Firstly, for a given prime we consider systems with dimension $p^n$ and we can increase 
arbitrarily the exponent $n$ (the $p$-adic formalism does that).
Secondly, in the product  of $p^n$-dimensional systems, we can increase arbitrarily 
the number of primes, and the 
Schwartz-Bruhat space (conditions 1B and 2B in definition \ref{d1}) ensures that the integrals converge.

\paragraph*{Embedding of a set of finite quantum systems into $\Sigma [{\widehat {\mathbb Z}},({\mathbb Q}/{\mathbb Z})]$:} 
The formalism of a set of finite quantum systems, can also be embedded into 
the formalism for $\Sigma [{\widehat {\mathbb Z}},({\mathbb Q}/{\mathbb Z})]$.
Then the building blocks are `mathematical subsystems', each of which has power of prime dimension.

As an example, we consider a bi-partite system $AB$ comprised of two `physical subsystems' $A$ and $B$, with 
variables in ${\mathbb Z}(k_A)$ and ${\mathbb Z}(k_B)$, where the $k_A,k_B$ are factorized in terms of prime numbers as
$k_A=p_1^2p_2$ and $k_B=p_1p_3^2$.
The dimension of the system $AB$ is $k_Ak_B=p_1^3p_2p_3^2$.
The general wavefunction of the system $AB$, can be written as
\begin{eqnarray}
&&F(x)=\sum _k F_{1k}(x_1)F_{2k}(x_2)F_{3k}(x_3)\nonumber\\
&&x\in {\mathbb Z}(k_Ak_B);\;\;\;\;x_1\in {\mathbb Z}(p_1^3);\;\;\;\;x_2\in {\mathbb Z}(p_2)
;\;\;\;\;x_3\in {\mathbb Z}(p_3^2).
\end{eqnarray} 
Here the system $AB$ is comprised of three `mathematical subsystems', with dimensions $p_1^3$, $p_2$, $p_3^2$.

We next consider the special case of a factorizable 
wavefunction $f(x_A)g(x_B)$, where the wavefunction $f(x_A)$ of the system $A$ is
\begin{eqnarray}
f(x_A)=f_{11}(x_1')f_{21}(x_2)+f_{12}(x_1')f_{22}(x_2);\;\;\;\;x_A\in {\mathbb Z}(k_A);\;\;\;\;x_1'\in {\mathbb Z}(p_1^2);\;\;\;\;x_2\in {\mathbb Z}(p_2),
\end{eqnarray}
and the wavefunction $g(x_B)$ of the system $B$ is
\begin{eqnarray}
g(x_B)=g_{11}(x_1'')g_{31}(x_3);\;\;\;\;x_B\in {\mathbb Z}(k_B);\;\;\;\;x''_1\in {\mathbb Z}(p_1);\;\;\;\;
x_3\in {\mathbb Z}(p_3^2).
\end{eqnarray}
In this special case the wavefunction of the system is
\begin{eqnarray}
&&F(x)=f(x_A)g(x_B)=\sum _{k=1}^2 F_{1k}(x_1)F_{2k}(x_2)F_{3k}(x_3)\nonumber\\
&&F_{11}(x_1)=f_{11}(x_1')g_{11}(x_1'');\;\;\;\;F_{12}(x_1)=f_{12}(x_1')g_{11}(x_1'');\;\;\;\;x_1\in {\mathbb Z}(p_1^3)\nonumber\\
&&F_{21}(x_2)=f_{21}(x_2);\;\;\;\;F_{22}(x_2)=f_{22}(x_2)\nonumber\\
&&F_{31}(x_3)=g_{31}(x_3);\;\;\;\;F_{32}(x_3)=g_{31}(x_3).
\end{eqnarray}
For $F_{11}$ (and also for $F_{12}$) we need to use a bijective map from ${\mathbb Z}(p_1^3)$ to ${\mathbb Z}(p_1^2)\times {\mathbb Z}(p_1)$ 
in order to write $f_{11}(x_1')g_{11}(x_1'')$ as $F_{11}(x_1)$ with $x_1\in {\mathbb Z}(p_1^3)$.
This example shows explicitly, how a set of finite quantum systems can be viewed as being comprised of mathematical subsystems 
each of which has power of prime dimension.

\subsection{The partially ordered sets $\Sigma $ and $\Sigma _S$ as $T_0$ topological spaces}

The partial order between the $({\mathfrak Z}(k),\widetilde {\mathfrak Z}(k))$
in Eq.(\ref{478}),
induces the partial order `subsystem' between the $\Sigma [{\mathfrak Z}(k),\widetilde {\mathfrak Z}(k)]$.
Indeed, for $k|\ell$ the $\widetilde {\mathfrak Z}(k)$ is a subgroup of $\widetilde {\mathfrak Z}(\ell)$ 
and therefore the
$\Sigma [{\mathfrak Z}(k),\widetilde {\mathfrak Z}(k)]$ is a subsystem of
$\Sigma [{\mathfrak Z}(\ell),\widetilde {\mathfrak Z}(\ell)]$:
\begin{eqnarray}
k|\ell \;\;\rightarrow\;\;\Sigma[{\mathfrak Z}(k),\widetilde {\mathfrak Z}(k)]
\prec \Sigma[{\mathfrak Z}(\ell),\widetilde {\mathfrak Z}(\ell)];\;\;\;\;k,\ell \in {\mathbb N}_S.
\end{eqnarray}
The sets
\begin{eqnarray}\label{bbb}
&&\Sigma (n)=\{\Sigma[{\mathbb Z}(k),{\mathbb Z}(k)]\;|\;k|n\}\nonumber\\
&&\Sigma=\{\Sigma[{\mathbb Z}(\ell),{\mathbb Z}(\ell)]\;|\;\ell \in {\mathbb N}\}\nonumber\\
&&\Sigma _S=\{\Sigma[{\mathfrak Z}(n),\widetilde {\mathfrak Z}(n)]\;|\;n\in {\mathbb N}_S\}
\end{eqnarray}  
with subsystem as order, are directed partially ordered sets.

The set $\Sigma (n)$ is order isomorphic to ${\mathbb N}(n)$.
The factorization discussed in section \ref{6vb}, is analogous to the factorization of a positive integer  in terms of prime numbers. 
It expresses $\Sigma (n)$ as tensor product of systems with power of prime dimension, which can be viewed as `building blocks' of all finite quantum systems.
Our statements for ${\mathbb N}(n)$ in section \ref{refe}, are also valid for $\Sigma (n)$.
For example, if we factorize $n$ as in Eq.(\ref{factorize10}), 
the partially ordered set $\Sigma (n)$, has width equal to $\ell$. 
Physically, this means that it can be regarded as an $\ell$-partite system.
Also $\Sigma (n)$ can be partitioned in a way analogous to Eq.(\ref{901}).

The set $\Sigma $ is order isomorphic to ${\mathbb N}$ (with divisibility as an order). It contains all $\Sigma (n)$ where $n\in {\mathbb N}$.
But it is not a directed-complete partially ordered set, and it has no maximal elements.

The set $\Sigma _S$ is order isomorphic to ${\mathbb N}_S$, and it
is a directed-complete partially ordered set.
The $\Sigma(\widehat {\mathbb Z},{\mathbb Q}/{\mathbb Z})$ is maximum element in $\Sigma _S$.
It is the smallest quantum system that contains all the $\Sigma({\mathfrak Z}(n),\widetilde {\mathfrak Z}(n))$
(where $n\in {\mathbb N}_S$) as subsystems.

\paragraph*{The topological spaces $(\Sigma , T_{\Sigma})$ and $(\Sigma _S, T_{\Sigma _S})$:}
There is a bijective map between ${\mathbb N}_S$ and $\Sigma _S$:
\begin{eqnarray}
&&n\;\;\leftrightarrow \;\;\Sigma({\mathfrak Z}(n),\widetilde {\mathfrak Z}(n))\nonumber\\
&&\Omega\;\;\leftrightarrow \;\;\Sigma [{\widehat {\mathbb Z}}, ({\mathbb Q}/{\mathbb Z})]
\end{eqnarray}
Therefore we can make $\Sigma _S$ a topological space and $({\mathbb N}_S, T_{{\mathbb N}_S})\sim (\Sigma _S, T_{\Sigma _S})$ (as described in section \ref{topology}).
The `points' in this topology are the systems $\Sigma({\mathfrak Z}(n),\widetilde {\mathfrak Z}(n))$,
and an open (resp., closed set) contains some systems and their subsystems (resp. supersystems).
In a similar way we can make $\Sigma $ a topological space and $({\mathbb N}, T_{{\mathbb N}})\sim (\Sigma , T_{\Sigma })$.
From proposition \ref{ttt}, it follows that both $(\Sigma _S, T_{\Sigma _S})$ and also $(\Sigma , T_{\Sigma })$ are $T_0$-topological spaces
but they are not $T_1$-topological spaces.
The physical meaning of this is discussed below.
Also the $(\Sigma _S, T_{\Sigma _S})$ is compact and the $(\Sigma , T_{\Sigma })$ is locally compact.

In summary, from a partial order theory point of view, $\Sigma _S$ is a directed-complete partial order, while
$\Sigma $ is not a directed-complete partial order.
From a topological point of view, $\Sigma _S$ is a compact $T_0$-space, while
$\Sigma $ is a locally compact $T_0$-space (which is not compact).  
This is shown in table I.

\subsection{Physical importance of the $T_0$-topology}

The fact that all the  $(\Sigma _S, T_{\Sigma _S})$, $(\Sigma , T_{\Sigma })$,  $(\Sigma ^{(p)}, T_{{\Sigma }^{(p)}})$,
$(\Sigma _S^{(p)}, T_{{\Sigma }_S^{(p)})}$) are $T_0$-topological spaces and they are not $T_1$-topological spaces,
reflects very fundamental aspects of the logical relationship between a quantum system and its subsystems and 
supersystems. 
We present our arguments for the system $(\Sigma _S, T_{\Sigma _S})$.

The fact that $(\Sigma _S, T_{\Sigma _S})$ is a  $T_0$-space,
means that for two distint elements (points) $\Sigma({\mathfrak Z}(m),\widetilde {\mathfrak Z}(m))$ and 
$\Sigma({\mathfrak Z}(n),\widetilde {\mathfrak Z}(n))$, 
there is an open set containing one 
of them but not the other. 
If one of the systems is a subsystem of the other, 
which is the case if for example $m|n$, then the open set
\begin{eqnarray}
U(m)=\{\Sigma({\mathfrak Z}(k),\widetilde {\mathfrak Z}(k))\;|\;k|m\},
\end{eqnarray}
is such that
\begin{eqnarray}
\Sigma({\mathfrak Z}(m),\widetilde {\mathfrak Z}(m))\in U(m);\;\;\;\;
\Sigma({\mathfrak Z}(n),\widetilde {\mathfrak Z}(n))\notin U(m).
\end{eqnarray}
If none of the $\Sigma({\mathfrak Z}(m),\widetilde {\mathfrak Z}(m))$,
$\Sigma({\mathfrak Z}(n),\widetilde {\mathfrak Z}(n))$ is a subsystem of the other,  
this means that $m$ is not a divisor or a multiple of $n$.
Then $U(m)$ contains $\Sigma({\mathfrak Z}(m),\widetilde {\mathfrak Z}(m))$ but it does not contain
$\Sigma({\mathfrak Z}(n),\widetilde {\mathfrak Z}(n))$.

On the other hand $(\Sigma _S, T_{\Sigma _S})$ is not a $T_1$-space. In a $T_1$-space, for any pair of
$\Sigma({\mathfrak Z}(m),\widetilde {\mathfrak Z}(m))$,
$\Sigma({\mathfrak Z}(n),\widetilde {\mathfrak Z}(n))$
there exist two open sets $U_1$ and $U_2$ such that
\begin{eqnarray}
&&\Sigma({\mathfrak Z}(m),\widetilde {\mathfrak Z}(m))\in U_1;\;\;\;\;\Sigma({\mathfrak Z}(m),\widetilde {\mathfrak Z}(m))\notin U_2\nonumber\\
&&\Sigma({\mathfrak Z}(n),\widetilde {\mathfrak Z}(n))\notin U_1;\;\;\;\;\Sigma({\mathfrak Z}(n),\widetilde {\mathfrak Z}(n))\in U_2
\end{eqnarray}
But this is impossible if $\Sigma({\mathfrak Z}(m),\widetilde {\mathfrak Z}(m))$ is a subsystem of
$\Sigma({\mathfrak Z}(n),\widetilde {\mathfrak Z}(n))$, because any open set that contains $\Sigma({\mathfrak Z}(n),\widetilde {\mathfrak Z}(n))$
will also contain its subsystems.
Consequently, $(\Sigma _S, T_{\Sigma _S})$ is not a $T_1$-space.

It is seen that the properties of the $T_0$ topology reflect very fundamental logical
relationships between a system and its subsystems.
The logical concept `subsystem' is intimately connected to the $T_0$ topology and it is incompatible with 
the $T_1$ topology.

\subsection{Continuity of physical quantities in the set $\Sigma $ as a function of $n$}\label{3xc}

The fact that the set $\Sigma $ in Eq.(\ref{bbb}) is a topological space, can be used to define 
continuity of physical quantities in the systems $\Sigma [{\mathbb Z}(n),{\mathbb Z}(n)]$, as a function of $n\in {\mathbb N}$.

In each system $\Sigma [{\mathbb Z}(n),{\mathbb Z}(n)]$ we define a quantity ${\cal L} _n$ as a map from the space
${\mathfrak B}[{\mathbb Z}(n),{\mathbb Z}(n)]$ to the set ${\mathbb R}$ of real numbers (e.g., entropic quantities);
or to the set of $n\times n$ matrices (e.g., Wigner functions); etc.
We therefore have a set ${\mathbb L}=\{{\cal L}_2,{\cal L}_3,...\}$ of such quantities which are defined in the various `points' in the set $\Sigma$.
We calculate the quantity ${\cal L}_n [f^{(n)}({\cal X}_n)]$ for a function $f^{(n)}({\cal X}_n) \in {\mathfrak B}[{\mathbb Z}(n),{\mathbb Z}(n)]$
of the system $\Sigma [{\mathbb Z}(n),{\mathbb Z}(n)]$. 
We then embed this state within the supersystem $\Sigma [{\mathbb Z}(r),{\mathbb Z}(r)]$, where $n|r$,
(i.e., we consider the state ${\overline {\bf E}}_{nr}[f^{(n)}({\cal X}_n)]$) 
and we calculate ${\cal L}_r\{{\overline {\bf E}}_{nr}[f^{(n)}({\cal X}_n)]\}$.
In ref\cite{vourdas}, we have called ubiquitous quantities the ones for which
\begin{eqnarray}\label{rf}
{\cal L}_r\{{\overline {\bf E}}_{nr}[f^{(n)}({\cal X}_n)]\}={\cal L}_n [f^{(n)}({\cal X}_n)]
\end{eqnarray}
for all $n|r$, and for all states $f^{(n)}({\cal X}_n) \in {\mathfrak B}[{\mathbb Z}(n),{\mathbb Z}(n)]$.
This compatibility condition, ensures that the various ${\cal L}_n$ represent the same quantity.
We have shown that various entropic quantities, and also the Wigner and Weyl functions, are ubiquitous quantities.

We then make the set ${\mathbb L}$ a topological space as described in section \ref{topology}.
Then the function from ${\mathbb N}$ to ${\mathbb L}$ which maps $n$ into ${\cal L}_n$, is a continuous function.
This makes formal the intuitive concept of continuity between a quantity defined in a large system and the same quantity in a small system.

\subsection{Physical importance of the profinite topology and of the Schwartz-Bruhat space}

The important concept for going from the quantum formalism for $\Sigma[{{\mathbb Z}_p},({\mathbb Q}_p/{\mathbb Z}_p)]$
into the quantum formalism for $\Sigma [{\widehat {\mathbb Z}},({\mathbb Q}/{\mathbb Z})]$ is the restricted tensor product.
The Schwartz-Bruhat space ${\mathfrak B}[{\widehat {\mathbb Z}},({\mathbb Q}/{\mathbb Z})]$
consists of finite linear combinations of $\prod f_p(x_p)$ where  $f_p(x_p)=1$, for all but a finite number of $p$
(it is the restricted tensor product of the ${\mathfrak B}[{{\mathbb Z}_p},({\mathbb Q}_p/{\mathbb Z}_p)]$). 
As we explained earlier, for this we need the topology where the open sets of ${\widehat {\mathbb Z}}$
are $\prod U_p$ where $U_p$ is an open sets in ${\mathbb Z}_p$, and $U_p={\mathbb Z}_p$ for all but a finite number of $p$.

From this and our discussion in section \ref{a45}, on the role of the $p$-adic topology for  the system
$\Sigma[{{\mathbb Z}_p},({\mathbb Q}_p/{\mathbb Z}_p)]$, it follows that in $\Sigma [{\widehat {\mathbb Z}},({\mathbb Q}/{\mathbb Z})]$ 
we have a formalism that includes all 
$\Sigma[\prod _p{\mathbb Z}(p^e),\prod _p{\mathbb Z}(p^e)]$, as subsystems.
Here $e=n+k$ where $n,k$ are the degrees of local constancy and compact support, correspondingly, of the functions $f_p(x_p)$. 
The quantum formalism of $\Sigma[\prod _p{\mathbb Z}(p^e),\prod _p{\mathbb Z}(p^e)]$ 
is precisely the quantum formalism for $\Sigma[{\mathbb Z}(n),{\mathbb Z}(n)]$ with $n=\prod p^e$
(the proof is based on the Chinese remainder theorem \cite{facto1,facto2}).                   
Therefore in $\Sigma[{\widehat {\mathbb Z}},({\mathbb Q}/{\mathbb Z})]$ we have a rich quantum formalism that includes all 
$\Sigma[{\mathbb Z}(n),{\mathbb Z}(n)]$ as subsystems.
We have discussed in detail the partial order between these systems.

\section{Discussion}

We have considered the quantum system $\Sigma[{\mathbb Z}_p, ({\mathbb Q}_p/{\mathbb Z}_p)]$.
The profinite group ${\mathbb Z}_p$ of positions is the inverse limit of ${\mathbb Z}(p^n)$, and
the group ${\mathbb Q}_p/{\mathbb Z}_p$ of momenta is the direct limit of ${\mathbb Z}(p^n)$.
The homomorphisms in the inverse and direct limits, have been used 
in the embeddings $E_{k\ell}$ (definition \ref{def1}), which in turn have been used to define the  embeddings ${\overline E}_{k\ell}$ 
of $\Sigma[{\mathbb Z}(p^k), {\mathbb Z}(p^k)]$ into 
$\Sigma[{\mathbb Z}(p^\ell), {\mathbb Z}(p^\ell)]$ (section \ref{EE1}).
The set $\Sigma ^{(p)}_S$ of all these quantum systems 
ordered with the binary relation `subsystem' is a complete chain which has the
$\Sigma[{\mathbb Z}_p, ({\mathbb Q}_p/{\mathbb Z}_p)]$ as supremum.
This is the smallest quantum system that contains all the $\Sigma[{\mathbb Z}(p^k), {\mathbb Z}(p^k)]$ as subsystems.
The set $\Sigma ^{(p)}$ is also a chain, but it is not complete.

We have also considered the quantum system $\Sigma[{\widehat {\mathbb Z}}, ({\mathbb Q}/{\mathbb Z})]$.
The profinite group ${\widehat {\mathbb Z}}$ of positions is the inverse limit of ${\mathbb Z}(n)$, and
the group ${\mathbb Q}/{\mathbb Z}$ of momenta is the direct limit of ${\mathbb Z}(n)$.
Here the homomorphisms in the inverse and direct limits, have been used 
in the embeddings ${\bf E}_{k\ell}$ (definition \ref{def2}), which in turn have been used for the  embedding ${\overline {\bf E}}_{k\ell}$ 
of $\Sigma[{\mathbb Z}(k), {\mathbb Z}(k)]$ into 
$\Sigma[{\mathbb Z}(\ell), {\mathbb Z}(\ell)]$ (section \ref{EE2}).
The set  $\Sigma _S$ of all these systems, ordered with the relation subsystem, is a directed-complete partial order.
The $\Sigma[{\widehat {\mathbb Z}}, ({\mathbb Q}/{\mathbb Z})]$ is maximum element in this set, i.e.,
it is the smallest system that contains all the systems in $\Sigma _S$ as subsystems.
The set $\Sigma $ is directed partially ordered set, but it is not directed-complete partial order.

There is a strong link between partial order theory and topology 
and the sets $\Sigma ^{(p)}_S$ and $\Sigma _S$ in Eq.(\ref{bbb}), have been studied as $T_0$ topological spaces.
The axioms of the $T_0$-topology, express fundamental logical relations between the quantum systems.
The topology of the set $\Sigma ^{(p)}_S$ (resp. $\Sigma _S$)
can be used to define continuity of a physical quantity in the various systems in $\Sigma ^{(p)}_S$ (resp. $\Sigma _S$).
The work reveals an interesting partial order theory and $T_0$-topology structure in finite quantum systems, the full use of which remains to be explored.

In both cases of the quantum systems $\Sigma[{\mathbb Z}_p, ({\mathbb Q}_p/{\mathbb Z}_p)]$
and $\Sigma[{\widehat {\mathbb Z}}, ({\mathbb Q}/{\mathbb Z})]$, we have defined the Schwartz-Bruhat spaces and the
Heisenberg-Weyl groups. We have also discussed the properties  of displacement and parity operators.

The paper can be viewed as a study of `large finite quantum systems'.
It factorizes them as tensor products of `mathematical component systems' which are labelled with prime numbers, and each of which has dimension $p^e$.
They are fundamental building blocks of finite quantum systems (analogous to the prime numbers which are fundamental building blocks of all positive integers).
Both, the number of the component systems and also the dimension of each component system can become arbitrarily large, 
but structures like the Schwartz-Bruhat space, the restricted tensor product of spaces, the restricted direct product of locally compact groups, etc, ensure that there are no divergencies. 

The work uses profinite groups, algebraic number theory, partial order theory and $T_0$ topology in a quantum mechanical context, with emphasis on the 
physical meaning of the mathematical concepts.

\vspace{1cm}
\begin{table}[htbp]
\caption{Some sets of quantum systems and their characterization as partial orders and as topologies}
\begin{tabular}{|c|c|c|}
  \hline
  set of quantum systems& partial order&topology\\ \hline 
  $\Sigma ^{(p)}$& not complete chain&locally compact, $T_0$\\ \hline
 $\Sigma _S^{(p)}$& complete chain&compact, $T_0$\\ \hline
$\Sigma $& not directed-complete partial order&locally compact, $T_0$\\ \hline
$\Sigma _S$& directed-complete partial order&compact, $T_0$\\ \hline
  \end{tabular}
\end{table}


\begin{thebibliography}{99}

\bibitem{PRO1}
L. Ribes, P. Zalesskii, `Profinite groups', Springer, Berlin, 2000

\bibitem{PRO2}
J. Wilson, `Profinite groups', (Clarendon, Oxford, 1998)

\bibitem{PRO3}
B. Klopsch, N. Nikolov, C. Voll, `Lectures on profinite topics in group theory',
(Cambridge Univ. Press, Cambridge, 2011)

\bibitem{F1}
A. Vourdas, Rep. Prog. Phys. 67, 1 (2004)

\bibitem{F2}
G. Bjork, A.B. Klimov, L.L. Sanchez-Soto, Prog. Optics 51, 469 (2008)

\bibitem{F3}
M. Kibler, J. Phys. A42, 353001 (2009)

\bibitem{F4}
N. Cotfas, J.P. Gazeau, J.Phys. A43, 193001(2010)

\bibitem{F5}
T. Durt, B.G. Englert, I. Bengtsson, K. Zyczkowski, Int. J. Quantum Comp. 8, 535 (2010)

\bibitem{To1}
P. Stovicek, J. Tolar, Rep. Math. Phys. 20, 157 (1984) 

\bibitem{To2}
V.S. Varadarajan, Lett. Math. Phys. 34, 319 (1995)

\bibitem{To3}
T. Digernes, E. Husstad, V.S. Varadarajan, Math. Scand. 84, 261 (1999)

\bibitem{facto1}
A. Vourdas, C. Bendjaballah, Phys. Rev. 47, 3523 (1993)

\bibitem{facto2}
A. Vourdas,  J.Phys. A36, 5645 (2003)

\bibitem{Galois}
A. Vourdas, J. Phys. A40, R285 (2007)

\bibitem{NP1}
F.Q. Gouvea, `p-adic numbers', Springer, Berlin, 1993

\bibitem{NP2}
A.M. Robert, `A course in p-adic analysis' Springer, Berlin, 2000

\bibitem{N1}
J.W.S. Cassels, A. Frohlich (Editors), `Algebraic number theory', Academic, London, 1968

\bibitem{N2}
A. Weil, `Basic number theory' Springer, Berlin, 1973

\bibitem{N3}
A. Weil, Acta Math, 111, 143 (1964)

\bibitem{N4}
S. Lang, `Algebraic number theory', Springr, Berlin,1970
 
\bibitem{r1}
L. Brekke, P. Freund, M. Olson, E. Witten, Nucl. Phys. B302, 365 (1988)

\bibitem{r2}
Ph. Ruelle, E. Thiran, D. Verstegen, J. Weyers, J. Math. Phys. 30, 2854-2874 (1989)

\bibitem{r3}
V.S. Vladimirov, I.V. Volovich, Commun. Math. Phys. 123, 659-676 (1989)

\bibitem{r4}
Y. Meurice,  Commun. Math. Phys. 135, 303 (1991)

\bibitem{r5}
L. Brekke, P. Freund, Phys. Rep. 233, 1 (1993)

\bibitem{r6}
S. Haran,  Ann. Inst. Fourier 43, 997 (1993)

\bibitem{r7}
E.I. Zelenov, Commun. Math. Phys. 155, 489-502 (1993)

\bibitem{r8}
V.S. Vladimirov, I.V. Volovich,E.I. Zelonov, `$p$-adic analysis and mathematical physics',
World Scientific, Singapore, 1994

\bibitem{r9}
S. Albeverio, R. Cianci, A. Khrennikov,  J. Phys. A30, 881 (1997)

\bibitem{r10}
S. Albeverio, R. Cianci, A. Khrennikov,  J. Phys. A30, 5767 (1997)

\bibitem{r11}
V.S. Varadarajan, Lett. Math. Phys. 39, 97 (1997)

\bibitem{r12}
A. Vourdas, J. Phys. A41, 455303 (2008)

\bibitem{r13}
A. Vourdas, J. Fourier Anal. Appl. 16 , 748 (2010) 

\bibitem{r14}
A. Vourdas, J. Math. Phys. 52, 062103 (2011)

\bibitem{r15}
S. Albeverio, A. Khrennikov, V.M. Shelkovich, `Theory of p-adic distributions', Cambridge U. P., 
Cambridge, 2010

\bibitem{r16}
B. Dragovich, A. Khrenikov, S.V. Kozyrev, I.V. Volovich,  
P-adic numbers, Ultrametic analysis and Applications, 1, 1 (2009)

\bibitem{ap1}
R. Rammal, G, Toulouse, M.A. Virasoro,  Rev. Mod. Phys. 58, 765 (1986)

\bibitem{ap2}
A.Yu. Khrennikov, S.V. Kozyrev, Physica A359, 222 (2006)
 
\bibitem{ap3}
A.Yu. Khrennikov, S.V. Kozyrev, Physica A359, 241 (2006) 

\bibitem{ap4}
A.Yu. Khrennikov, S.V. Kozyrev, Physica A378, 283 (2007) 

\bibitem{comp}
E. Hehner, R.N. Horspool, SIAM J. Comp. 8, 124 (1979)

\bibitem{wav1}
J.J. Benedetto, R.L. Benedetto, J. Geom. Anal.14, 423 (2004)

\bibitem{wav2}
R.L. Benedetto, Contemporary Math. 435, 27 (2004)

\bibitem{wav3}
W.C. Lang, SIAM J. Math. Anal. 27, 305 (1996)

\bibitem{wav4}
A. Khrennikov, V.M. Shelkovich, Appl. Comput. Harmon. Anal. 28, 1 (2010) 

\bibitem{b1}
D. Bump, `Automorphic forms and representations', Cambridge Univ. Press, Cambridge, 1998

\bibitem{b2}
I.M. Gel'fand, M.I.Graev, Russian Math. Surveys 18, 29-109 (1963)

\bibitem{b3}
I.M. Gel'fand, M.I.Graev,I.I. Piatetskii-Shapiro, `Representation Theory and Automorphic Functions',
Academic, London, 1990

\bibitem{b4}
M.S. Osborne, J. Funct. Anal. 19, 40 (1975)

\bibitem{b5}
D. Ramakrishnan, R.J. Valenza, `Fourier analysis on number fields', Springer, Berlin, 1999

\bibitem{b6}
M.H. Taibleson, `Fourier analysis on local fields', (Princeton Univ. Press, Princeton, 1975)

\bibitem{b7}
V.S. Vladimirov, Russian Math. Surveys 43, 19 (1988)

\bibitem{AD1} 
B. Dragovich, Int. J. Mod. Phys. A10, 2349 (1995)

\bibitem{AD2}
B. Dragovich, Integral Transf. Spec. Funct. 6, 197-203 (1998) 

\bibitem{AD3}
G.S. Djordjevic, L.J. Nesic, B. Dragovich, Mod. Phys. Lett. A14, 317 (1999)

\bibitem{AD4}
A. Vourdas, J. Math. Anal. Appl. 394, 48 (2012)

\bibitem{L1}
H. Jacquet, R.P. Langlands, `Automorphic forms of $GL(2)$', Lecture Notes Math. (Springer, Berlin, 1970)

\bibitem{Bi}
G. Birkhoff, `Lattice theory' (Amer. Math. Soc. New York, 1948)

\bibitem{Sz}
G. Szasz, `Introduction to lattice theory' (Academic, London, 1963)

\bibitem{Gr}
G.A.Gratzer, `General lattice theory' (Springer, Berlin, 2003)

\bibitem{D1}
D.S. Scott, Lecture Notes in Mathematics 274, 97 (1972)

\bibitem{D2}
S. Abramsky, A. Jung, in S. Abramsky, D.M. Gabbay, T.S.E. Malbaum (Eds), Handbook of logic in Computer Science, 
(Oxford Univ. Press, Oxford, 1994)
 
\bibitem{D3}
G. Gierz, K.H. Hofmann, K. Keimel, J.D. Lawson, M.W. Mislove, D.S. Scott, 
`Continuous lattices and domains' (Cambridge Univ. Press, Cambridge, 2003) 

\bibitem{LO1}
G. Birkhoff, J. von Neumann, Ann. Math. 37, 823 (1936)

\bibitem{LO2}
C. Piron, `Foundations of quantum physics', Benjamin, New York, 1976

\bibitem{LO3}
G.W.  Mackey, `Mathematical foundations of quantum mechanics' Benjamin, New York, 1963

\bibitem{LO4}
K. Engesser, D.M. Gabbay, D. Lehmann, `Handbook of quantum logic and quantum structures', Elsevier, Amsterdam, 2009

\bibitem{T1}
M.H. Stone, Proc. Nat. Acad. Sci. 20,197 (1934)

\bibitem{T2}
H. Wallman, Ann. Math. 39, 112 (1938)

\bibitem{TO1}
N. Bourbaki, `General topology', part 1 (Hermann, paris, 1966)

\bibitem{TO2}
S. Willard, `General topology', (Dover, New York, 1970)

\bibitem{TO3}
J.L. Kelly, `General topology', (Springer, Berlin, 1955)

\bibitem{TO4}
L.A. Steen, J.A. Seebach, `Counterexamples in topology' (Dover, New York, 1995)

\bibitem{TO5}
L. Nachbin, `Topology and Order' (Van Nostrand, New York, 1965)

\bibitem{vourdas}
A. Vourdas, J. Math. Phys. 53, 122101 (2012)

\bibitem{Po}
L.S. Pontryagin, `Topological groups' (Gordon and Breach, New York, 1966)

\bibitem{DIL}
R.P. Dilworth, Ann. Math. 51,161 (1950)

\bibitem{ER}
D.F. Elliott, K.R. Rao, `Fast transforms', (Academic, London, 1982)

\bibitem{G1}
I.J. Good, IEEE Trans. Computers C20, 310 (1971)

\bibitem{G2}
J.H. McClellan, C.M. Rader, `Number theory in digital signal processing', (Prentice Hall, London, 1979)

\bibitem{LE}
U. Leonhardt, Phys. Rev. A53, 2998 (1996)

\bibitem{ZA}
J. Zak, J. Phys. A44, 345305 (2011)

\bibitem{ZV}
S. Zhang, A. Vourdas, J. Math. Phys. 44, 5084 (2003)

\end{thebibliography}
\end{document}